\documentclass[equations,11pt]{article}
\usepackage{amsfonts}
\usepackage{a4,tikz}
\usepackage{geometry}
\usepackage{color} 
\usepackage{subfigure}
\makeatletter
\newcommand\ackname{Acknowledgements}
\if@titlepage
  \newenvironment{acknowledgements}{%
      \titlepage
      \null\vfil
      \@beginparpenalty\@lowpenalty
      \begin{center}%
        \bfseries \ackname
        \@endparpenalty\@M1
      \end{center}}%
     {\par\vfil\null\endtitlepage}
\else
  \newenvironment{acknowledgements}{%
      \if@twocolumn
        \section*{\abstractname}%
      \else
        \small
        \begin{center}%
          {\bfseries \ackname\vspace{-.5em}\vspace{\z@}}%
        \end{center}%
        \quotation
      \fi}
      {\if@twocolumn\else\endquotation\fi}
\fi
\makeatother

\usepackage{amsmath}
\usepackage{amssymb}
\usepackage{amsthm}
\usepackage{mathrsfs}
\usepackage{eucal}
 \usepackage[usenames,dvipsnames]{pstricks}
 \usepackage{epsfig}
 \usepackage{pst-grad} 
 \usepackage{pst-plot} 
\renewcommand{\theequation}{\arabic{equation}}
\usepackage{amssymb,amsmath}
\usepackage{color}
\usepackage{accents}
\usepackage{texdraw}
\usepackage{eucal}
\usepackage{epic,epsfig}
\usepackage{graphicx}
\newtheorem{theorem}{Theorem}[section]

\newtheorem{prop}[theorem]{Proposition}
\theoremstyle{definition}

\numberwithin{equation}{section}
\DeclareMathAccent{\wtilde}{\mathord}{largesymbols}{"65}
\DeclareMathAccent{\what}{\mathord}{largesymbols}{"62}

\def\m@th{\mathsurround=0pt}
\mathchardef\bracell="0365
\def\upbrall{$\m@th\bracell$}
\def\undertilde#1{\mathop{\vtop{\ialign{##\crcr
    $\hfil\displaystyle{#1}\hfil$\crcr
     \noalign
     {\kern1.5pt\nointerlineskip}
     \upbrall\crcr\noalign{\kern1pt
   }}}}\limits}
\def\m@th{\mathsurround=0pt}
\mathchardef\bracell="0365
\def\upbrall{$\m@th\bracell$}
\def\underhat#1{\mathop{\vtop{\ialign{##\crcr
    $\hfil\displaystyle{#1}\hfil$\crcr
     \noalign
     {\kern1.5pt\nointerlineskip}
     \upbrall\crcr\noalign{\kern1pt
   }}}}\limits}
\usepackage{pgf}
\usepackage{tikz}
\usepackage{verbatim}
\usepackage[utf8]{inputenc}
\usetikzlibrary{arrows,automata}
\usetikzlibrary{positioning}

\def\theequation{\arabic{section}.\arabic{equation}}

\newcommand{\wh}{\widehat}
\newcommand{\wt}{\widetilde}

\def\hypotilde#1#2{\vrule depth #1 pt width 0pt{\smash{{\mathop{#2}
\limits_{\displaystyle\widetilde{}}}}}}
\def\hypohat#1#2{\vrule depth #1 pt width 0pt{\smash{{\mathop{#2}
\limits_{\displaystyle\widehat{}}}}}}

\newcommand{\tbu}{\,^{t\!}{\bu}}

\newcommand{\bblu}{\begin{color}{blue}}
\newcommand{\bred}{\begin{color}{red}}
\newcommand{\ecl}{\end{color}}

\newcommand{\bI}{\boldsymbol{I}}

\newcommand{\bK}{\boldsymbol{K}}
\newcommand{\bL}{\boldsymbol{L}}
\newcommand{\bM}{\boldsymbol{M}}
\newcommand{\bN}{\boldsymbol{N}}

\newcommand{\bU}{\boldsymbol{U}}

\newcommand{\bsX}{\boldsymbol{X}}
\newcommand{\bsY}{\boldsymbol{Y}}

\newcommand{\bLam}{\boldsymbol{\Lambda}}

\newcommand{\ld}{\lambda}

\newcommand{\be}{\begin{equation}}
\newcommand{\ee}{\end{equation}}
\newcommand{\bea}{\begin{eqnarray}}
\newcommand{\eea}{\end{eqnarray}}
\newcommand{\bse}{\begin{subequations}}
\newcommand{\ese}{\end{subequations}}
\newcommand{\nn}{\nonumber}

\newcommand{\ol}{\overline}

\newcommand{\bu}{\boldsymbol{u}}

\begin{document}

\def\theequation{\arabic{section}.\arabic{equation}}

\newtheorem{thm}{Theorem}[section]
\newtheorem{lem}{Lemma}[section]
\newtheorem{defn}{Definition}[section]
\newtheorem{ex}{Example}[section]
\newtheorem{rem}{}
\newtheorem{criteria}{Criteria}[section]
\newcommand{\ra}{\rangle}
\newcommand{\la}{\langle}

\title{\textbf{Discrete-time Ruijsenaars-Schneider system and Lagrangian 1-form structure}}
\author{\\\\Sikarin Yoo-Kong$^{\dagger,1 } $, Frank Nijhoff$^{\ddagger,2} $ \\
\small $^\dagger $\emph{Theoretical and Computational Physics (TCP) Group, Department of Physics,}\\ 
\small \emph{Faculty of Science, King Mongkut's University of Technology Thonburi,}\\
\small \emph{Thailand, 10140.}\\
\small $^\ddagger $\emph{Department of Applied Mathematics, School of Mathematics, University of Leeds,}\\
\small \emph{United Kingdom, LS2 9JT.}\\
\small  $^1$sikarin.yoo@kmutt.ac.th,
\small $^2$nijhoff@maths.leeds.ac.uk}
\maketitle


\abstract
We study the Lagrange formalism of the (rational) Ruijsenaars-Schneider (RS) system, both in discrete time as well 
as in continuous time, as a further example of a Lagrange 1-form structure in the sense of the recent paper \cite{Sikarin1}. The discrete-time model of the RS system 
was established some time ago arising via an Ansatz of a Lax pair, and was shown to lead to an exactly integrable 
correspondence (multivalued map)\cite{FrankN8}. In this paper we consider an extended system 
representing a family of commuting flows of this type, and establish a connection with the lattice KP system. 
In the Lagrangian 1-form structure of this extended model, the closure relation is verified making  
use of the equations of motion. Performing successive continuum limits on the RS system, we 
establish the Lagrange 1-form structure for the corresponding continuum case of the RS model.


\section{Introduction}\label{intro}
\setcounter{equation}{0}
The Ruijsenaars-Schneider (RS) system \cite{R1,R2}, i.e., the relativistic version of the Calogero-Moser (CM) system, is integrable both in the 
classical and quantum regimes. The classical model was discovered in \cite{R1} by considering the Poincar\'e Poisson algebra associated with 
sine-Gordon solitons, and was motivated by the discovery in the late 1970s of explicit soliton-type S-matrices for some relativistic 
two-dimensional quantum field theories (such as the massive Thirring model, the quantum sine-Gordon theory and the $O(N)$ $\sigma$-models). 
For reasons elucidated below, 
we are interested in Lagrangian aspects of the RS model, which have hardly received attention. 
An apparent reason for this is that the Hamiltonian description corresponding to the system is not of 
Newtonian form, and hence the usual connection between the Hamiltonian and the Lagrangian description 
through the Legendre transformation becomes quite convoluted. At the same 
time we are interested in the integrable time-discrete version of the RS system, which was proposed and studied in \cite{FrankN8}, where 
the Lagrangian description is more natural than the Hamiltonian one, because the finite-step time-iterate can be naturally viewed as a 
canonical transformation where the Lagrangian plays the role of its generating function. In \cite{FrankN8} the corresponding 
discrete-time Lagrangian was found, but the continuum limits were not considered so far. As we shall show, the latter can be used to derive 
a natural Lagrangian description for the continuous RS model as well, but in the context of what we call a Lagrangian 1-form structure. We will 
now explain what we mean by this latter notion. 

Recently, a novel point of view was developed on the role of the Lagrangian structure in integrable systems, cf. \cite{SF1}, where it was proposed 
that the fundamental property of \emph{multidimensional consistency} can be made manifest in the Lagrangians by thinking of the latter as 
components 
of a difference (or differential) ``Lagrange-form'' when the flows are embedded in a multidimensional space-time. A new variational principle 
was formulated which involves not only variations with respect to the dependent variables of the theory, but also with respect to the geometry 
in the space of independent (discrete or continuous) variables. In \cite{SF1}, this was laid out in the case of two-dimensional lattice equations, 
whilst in \cite{LNQ} it was extended also to the case of the 3-dimensional bilinear Kadomtsev–Petviashvili (KP) equation (Hirota's equation). Furthermore, in  \cite{XNL} 
a universal Lagrangian structure was established for quadrilateral affine-linear lattice equations as well as for their corresponding continuous 
counterparts, the so-called \textit{generating PDEs} of the system. The key property in all these systems, in which in a sense 
the integrability of the system  resides, is that the Lagrangian form is closed on solutions of the equations of the motion (but not identically 
closed for arbitrary field values). This can be viewed as a manifestation of the multidimensional consistency 
of the system under consideration on the Lagrangian level.  

In the case of integrable systems of ODEs, like system of equations of motion of integrable many-body systems,
the Lagrangian form structure is that of Lagrange 
1-forms. Recently, in collaboration with S Lobb, the authors studied a first example of such a Lagrange 1-form structure, namely the case of the discrete-time 
(rational) Calogero-Moser (CM) system, \cite{Sikarin1,Rinthesis}. The multidimensional consistency of the system in that case is represented by the co-existence of two 
or more independent \emph{commuting} discrete-time flows in the case of three or more particles. 
Starting with the discrete-time case, we furthermore established the Lagrange structure of the corresponding continuous case by 
performing systematic continuum limits on the discrete-time equations and Lagrangians. Of course, these systems exhibit also a multi-time Hamiltonian structure, where the various time-flows generated by the 
Hamiltonians, which are in involution with respect to a canonical Poisson structure, commute. 
However, it is not the case that one can perform naively a Legendre transformation on each of these 
Hamiltonians separately to yield a proper Lagrangian structure that makes sense as a coherent system. 
In fact, the higher-order Lagrangians emerging from such a naive approach 
would yield rather complicated algebraic expressions which seem unsuitable for further study. However, as we have shown in \cite{Sikarin1}, 
a proper Lagrangian 1-form structure can be defined for the CM system, in which the components of the form are 
mixed Lagrangians, of polynomial form in the time-derivatives, obeying the crucial \emph{closure property}, expressing the commutativity of the flows, \emph{on solutions of the equations of the motion}.  
To derive these Lagrangians, the connection between the semi-discrete KP equation and the discrete-time CM system, which arises as the pole-reduction of the former, 
was instrumental in order to guide the proper choice of higher-order continuum limits obtained by systematic expansions performed on the discrete-time model, thus 
leading to the Lagrangians in the continuum case. (Unfortunately, we do not know at this stage a Lagrangian 
of the semi-discrete KP equation in the relevant form, 
which would have allowed us to do the pole-reduction on the Lagrangians directly.) 

In the present paper, we proceed in the same spirit as in the paper \cite{Sikarin1}, to establish the Lagrange 
1-form structure of the discrete-time rational RS system. However, compared to the the case of the 
discrete-time CM system where there is a direct connection between the Lax matrices and the relevant 
Lagrangians, and where the closure relation is a direct consequence of a zero-curvature condition, 
such a direct connection seems absent in the RS case. Thus, in the latter case, the establishment of the 
closure relation for the Lagrangians which essentially were provided in \cite{FrankN8}, has to be 
verified by an explicit computation, and seems to be governed by a different mechanism. It is this aspect 
that makes the study of the RS system a worthwhile addition to the emerging theory of Lagrangian 
multi-form structures, confirming that the latter is universal structure underlying integrable systems.   
Furthermore, whereas the discrete-time CM system arises from the pole-reduction of a \textit{semi-discrete}  
KP equation (with one continuous and two discrete independent variables), the discrete-time RS is connected 
by an analogous reduction to the fully discrete KP equation (with three discrete independent variables). 
Thus, it is evident that (discrete-time) RS case, viewed as a relativistic analogue of the corresponding 
CM system, is richer than the CM case of \cite{Sikarin1}, containing an additional (deformation) 
parameter which can be viewed as the reciprocal of the speed of light. It is observed that this 
non-relativistic limit also corresponds to a particular continuum limit on the lattice KP system.


Although our focus in this paper is on the rational case of the RS model, \textcolor{red}{most} of our 
results on the Lagrangian structure can be extended straightforwardly to the 
trigonometric/hyperbolic case and even the elliptic case, however, as in the case of the CM system, we prefer all the statements we make to be backed up by the explicit solution of 
the equations of motion that can be constructed in the rational case. For instance, an important ingredient in the structure is what we call "constraint equations", which 
in addition to the one-dimensional equations of the motion can be verified explicitly for the solution obtained. These constraint equations involve the dynamics 
in two discrete variables (corresponding to trajectories in the space of independent variables which involve corners). Since, in contrast to the paper \cite{Sikarin1}, the starting 
point in the present paper is an Ansatz of a Lax pair rather than a reduction from a KP system (the connection with the lattice KP is established \emph{a posteriori}), the verification that all relations (equations of the motion, 
constraint equation and closure relation) hold true for a nontrivial family of solutions backs up the consistency of the 
whole structure of this system. 

The organization of the paper is as follows. In Section 2, we  review the construction of the 
single-flow discrete-time RS system and its exact solution in terms of a secular problem. Next, we extend the system 
by imposing additional commuting discrete flows in additional time directions, and derive the conditions (i.e., the 
constraint relations) for their compatibility. In Section 3, the Lagrangian 1-form structure of the discrete multi-time 
RS system is studied, and we verify the relevant closure relation by a direct computation involving the equations of 
motion as well as the constraints. Thus, we establish that this system possesses a Lagrangian 1-form 
structure in the sense of the paper \cite{Sikarin1}. In Section 4, a ``skew'' continuum  limit is taken,  
guided by the exact solution of the discrete-time RS system system, yielding what we call the 
\textit{semi-continuous RS system}. The latter, in fact, acts as a generating system for the 
continuous RS system, thus allowing us in Section 5 to derive the full continuum limit, by which we recover the 
continuous-time RS hierarchy albeit in a form involving mixed derivatives w.r.t. the multiple times of the RS 
hierarchy. Applying the same limit to the Lagrangian of the semi-continuous RS system, we obtain the Lagrangian 
components of the relevant continuous higher-time flows of the RS system, in a form (namely involving mixed 
higher-time derivatives) which is suitable for the interpretation as a Lagrangian 1-form structure, where the 
Lagrangian components obey the (continuous) closure relation subject to the solutions of the equations of the motion. In Section 6, the connection to the lattice KP system is presented, showing that 
the exact solution of the discrete multiple-time RS model leads to solutions of the lattice KP dependent variable 
as function of these multiple times.  In particular, the characteristic polynomial associated with the exact 
solution can be identified with the corresponding lattice KP $\tau$-function. Finally in Section 7 summary and discussion will be presented along with some open problems.

\section{The discrete-time Ruijsenaars-Schneider system and commuting flows}\label{exactsolution}
\setcounter{equation}{0} 
In this Section we review the discrete-time RS system which has been introduced in \cite{FrankN8}. This gives us an occasion to introduce appropriate notation which we will use throughout the paper. Furthermore, we derive the exact solution of the discrete equations of the motion and identify 
the constraint relations on commuting flows that can coexist in the system.
\\
\\
\subsection{The single-flow RS system}\label{SF}
 Following \cite{FrankN8} the discrete-time RS model is obtained on the basis of a Lax pair of the form:
\begin{subequations}\label{LM} 
\begin{equation}\label{LM1}
\wt{\boldsymbol \phi}=\boldsymbol M_\kappa\boldsymbol \phi\  ,\;\;\;\;\;\; \boldsymbol L_\kappa\boldsymbol \phi=\zeta \boldsymbol \phi\  , 
\end{equation}
for a vector function $\boldsymbol\phi$ and an eigenvalue $\zeta $, in which the matrices $\boldsymbol L_\kappa$ and $\boldsymbol M_\kappa$ 
are given by 
\begin{eqnarray}\label{LM121}
\boldsymbol L_{\kappa}&=&\frac{hh^T}{\kappa}+\boldsymbol L_0\;,\label{LM1}\\
\boldsymbol M_{\kappa}&=&\frac{\wt{h}h^T}{\kappa}+\boldsymbol M_0\;,\label{aLM2}
\end{eqnarray}\end{subequations} 
where in the rational case
\begin{eqnarray}\label{LM0}
\boldsymbol L_0=\sum_{i,j=1}^N\frac{h_ih_j}{x_i-x_j+\lambda}E_{ij}\;,\;\;\mbox{and}\;\;\;\boldsymbol M_0=\sum_{i,j=1}^N\frac{\wt{h}_ih_j}{\wt{x}_i-x_j+\lambda}E_{ij}\;.\label{LM2}
\end{eqnarray}
In \eqref{LM} the $x_i$ are the particle positions, whilst the $h_i$ are auxiliary variables which will be determined later. 
The \emph{tilde} ~$\wt{\phantom{a}}$~ is a shorthand notation for the discrete-time shift, i.e. for $x_i(n)=x_i$, and we write $x_i(n+1)=\wt{x}_i$, 
and $x_i(n-1)=\undertilde{x_i}$. 
The variable $\kappa$ is the additional spectral parameter, whereas $\lambda$ is a parameter of the system related to the 
non-relativistic limit. The matrix $E_{ij}$ are the standard elementary matrices whose entries are given by 
$(E_{ij})_{kl}=\delta_{ik}\delta_{jl}$.
\\
\\
The compatibility condition of the system \eqref{LM}:
\begin{eqnarray}\label{coM}
&&\wt{\boldsymbol L}_{\kappa}\boldsymbol M_{\kappa}=\boldsymbol M_{\kappa}\boldsymbol L_{\kappa}\Rightarrow \nn\\
&&\left( \frac{\wt h\wt h^T}{\kappa}+\wt{\boldsymbol L}_0\right)\left( \frac{\wt{h}h^T}{\kappa}+\boldsymbol M_0\right)
=\left(\frac{\wt{h}h^T}{\kappa}+\boldsymbol M_0\right)\left( \frac{hh^T}{\kappa}+\boldsymbol L_0\right)
\end{eqnarray}
%
%
From the coefficients of $1/\kappa^2$ we derive the conservation law $\mbox{Tr}\wt{\boldsymbol L}_{\kappa}=\mbox{Tr}{\boldsymbol L}_{\kappa}$ leading to
\begin{equation}
\sum_{j=1}^N\wt{h}_j^2=\sum_{j=1}^Nh_j^2\;,\label{conh}
\end{equation}
and furthermore, the coefficients of $1/\kappa$ give
\begin{equation}\label{iden1}
\wt{\boldsymbol L}_{0}\wt h h^T+\wt h\wt h^T\boldsymbol M_{0}=\boldsymbol M_{0}hh^T+\wt h h^T\boldsymbol L_{0}\;,
\end{equation}
and the rest produces the equation
\begin{eqnarray}\label{iden2}
\wt{\boldsymbol L}_{0}\boldsymbol M_{0}=\boldsymbol M_{0}\boldsymbol L_{0}\;.
\end{eqnarray}
\eqref{iden1} and \eqref{iden2} produce the relations
\begin{equation}\label{idenforh}
\sum_{j=1}^N
\left(\frac{\wt{h}_j^2}{\wt{x}_i-\wt{x}_j+\lambda}-\frac{h_j^2}{\wt{x}_i-x_j+\lambda}\right)
=\sum_{j=1}^N
\left(\frac{h_j^2}{x_j-x_l+\lambda}-\frac{\wt{h}_j^2}{\wt{x}_j-x_l+\lambda}\right)\;,
\end{equation}
for all $i,j=1,2,...,N$. Consequently, both sides of \eqref{idenforh} must be independent of the external particle label. Thus, we find a coupled system of equations in terms of the variables $h_i$, and $x_i$ of the form: 
\begin{subequations}\label{idenforh2aa}
\begin{eqnarray}
\sum_{j=1}^N
\left(\frac{\wt{h}_j^2}{\wt{x}_i-\wt{x}_j+\lambda}-\frac{h_j^2}{\wt{x}_i-x_j+\lambda}\right)&=&-p\;\;\;\;\;, \forall i\;,
\end{eqnarray}
\begin{eqnarray}
\sum_{j=1}^N
\left(\frac{h_j^2}{x_j-x_l+\lambda}-\frac{\wt{h}_j^2}{\wt{x}_j-x_l+\lambda}\right)&=&-p \;\;\;\;\;,\forall l\;,
\end{eqnarray}
\end{subequations}
where the quantity $p=p(n)$ does not carry a particle label, but could still be a function of $\;n\;$. 
%

In order to derive a closed set of equations of motion for the variables $x_i$ we have to determine the variables 
$h_i$  in terms of the $x_i$ and their time-shifts. To do this most effectively,we use 
the Lagrange interpolation formula, which is given in the following lemma:  

%
\textbf{Lemma 2.1.}\emph{
\textbf{Lagrange interpolation formula}: Consider $2N$ noncoinciding complex numbers $x_k$ and $y_k$, where $k=1,2,...,N$. Then the following formula holds true:
\begin{subequations}
\begin{equation}\label{LaInab}
\prod_{k=1}^N\frac{(\xi-x_k)}{(\xi-y_k)}=1+\sum_{k=1}^N\frac{1}{(\xi-y_k)}\frac{\prod_{j=1}^N(y_k-x_j)}{\prod_{j=1,j\neq k}^N(y_k-y_j)}\;.
\end{equation}
As a consequence
\begin{equation}\label{LaIn2}
-1=\sum_{k=1}^N\frac{1}{(x_i-y_k)}\frac{\prod_{j=1}^N(y_k-x_j)}{\prod_{j=1,j\neq k}^N(y_k-y_j)}\;,\;\;\;\;\;\;i=1,...,N\;,
\end{equation}
\end{subequations}
which is obtained by substituting $\xi=x_i$ into \eqref{LaInab}.}
\\
\\
%
Applying the Lagrange interpolation formula to the following rational function of 
the indeterminate variable $\xi$:   
\begin{eqnarray}\label{poly}
\prod_{j=1}^N\frac{(\xi-x_j+\lambda)(\xi-\wt{x}_j-\lambda)}{(\xi-x_j)(\xi-\wt{x}_j)}\;,
\end{eqnarray}
we obtain 
\begin{subequations} \label{heqs} 
\begin{eqnarray}\label{heqsa}
h_i^2&=&-p\frac{\prod_{j=1}^N(x_i-x_j+\lambda)(x_i-\wt{x}_j-\lambda)}{\prod_{j\ne i}^N(x_i-x_j)\prod_{j=1}^N(x_i-\wt{x}_j)}\;, \\
\wt{h}_i^2&=&p\frac{\prod_{j=1}^N(\wt{x}_i-x_j+\lambda)(\wt{x}_i-\wt{x}_j-\lambda)}{\prod_{j\ne i}^N(\wt{x}_i-\wt{x}_j)\prod_{j=1}^N(\wt{x}_i-x_j)} \;, 
\label{heqsb} 
\end{eqnarray}\end{subequations} 
for $i=1,2,...,N$ which we obtain the following system of equations
\begin{equation}\label{eqmotion1}
\frac{p }{\hypotilde 0 {p }}\prod\limits_{\mathop {j = 1}\limits_{j \ne i} }^N \frac{(x_i-x_j+\lambda)}{(x_i-x_j-\lambda)}=\prod_{j=1}^N\frac{(x_i-\wt{x}_j)(x_i-{\hypotilde 0 x}_j+\lambda)}{(x_i-{\hypotilde 0 x}_j)(x_i-\wt{x}_j-\lambda)}\;.
\end{equation}
Eq. \eqref{eqmotion1} can be considered to be the product version of \eqref{heqs} which is a system of $N$ equations for $N+1$ unknowns, $x_1,...,x_N$ and $p$. 
There is so far no separate equation for $p $, which amounts to a freedom in determining the centre of mass motion, 
and fixing a specific choice of the time-evolution for $p$ we would get a closed set of equations 
(for more details, see \cite{FrankN8}). For simplicity we will often take, in what follows, \emph{$p$ to be 
constant} as a function of the discrete-time variable $n$. 

The exact solution of the equations of motion can be derived in a way similar to the continuous case of the rational RS model, cf. 
\cite{FrankN8}, cf. also \cite{R1,R5} using the explicit form of the Lax matrices \eqref{LM1}, and is given 
by the following statement. 

\begin{prop}
Let $\bLam$ be a constant diagonal matrix, and let $p(n)$ be a given function of the discrete-time 
variable $n$ such that $p(n)\boldsymbol{I}+\bLam$ is invertible for all $n \geq 0$, and let the   
$N\times N$ matrix function of the discrete variable $n$, $\boldsymbol{Y}(n)$, be given by\footnote{The factors in the first term of \eqref{Yn2} come out directly from the computation but they can be removed by conjugation. For clarity, we write the inverses of (diagonal) matrices such as $p\boldsymbol{I}+\boldsymbol{\Lambda}$ in fractional form, where it does not lead to ambiguity.}  
\begin{subequations}
\begin{equation}\label{Yn2}
\boldsymbol{Y}(n)=\left[\prod_{k=0}^{n-1}(p(k)\boldsymbol{I}+\bLam)^{-1}\right]\,\boldsymbol{Y}(0)\left[\prod_{k=0}^{n-1}(p(k)\boldsymbol{I}+\bLam)\right]-\sum_{k=0}^{n-1}\frac{p(k)\lambda}{p(k)\boldsymbol{I}+\bLam} \;,
\end{equation}
subject to the following condition on the initial value matrix
\begin{equation}\label{constraint(n)}
[\boldsymbol Y(0),\bLam]=-\lambda\bLam+\mbox{\rm rank 1}\;,
\end{equation}
\end{subequations}
then the eigenvalues $x_i(n)$ of the matrix $\boldsymbol{Y}(n)$ coincide with the solutions for particle position of the discrete-time RS system, i.e. they solve the discrete equations of motion \eqref{eqmotion1}. 
\end{prop} 

The details of the proof are given in are Appendix A. In fact, the matrix $\bLam$ can be identified with 
the (diagonal) matrix of eigenvalues of the matrix $\boldsymbol{L}_0$, cf. \eqref{XX2}, which coincide 
with the eigenvalues of $\boldsymbol{L}_0(0)$ at the initial value, since we are dealing with an 
isospectral problem. However, as far as the above proposition is concerned, both $\boldsymbol\Lambda$ and 
the initial value matrix $\boldsymbol{Y}(0)$ can be chosen in lieu of posing initial conditions on the particle 
positions\footnote{In fact, specifying $x_i(0)$ and $x_i(1)$, $i=1,2,...,N$, the matrices $\boldsymbol{Y}(0)$ and $\boldsymbol\Lambda$ can be computed from the initial values by using the $h_i^2$ from \eqref{heqsa} and the matrix $\boldsymbol{L}_0$ from \eqref{LM0} at $n=0$, where for simplifity we assume that the latter can be diagonalised.}.

The functions $p(n)$ determine the centre of mass motion, which can be separated from the relative 
motion of the particles. In the special case of constant $p$: $p=\widetilde{p}$ 
(which amounts to choosing a frame in which the centre of mass remains stationary) 
the expression for the matrix $\boldsymbol{Y}(n)$ becomes
\begin{equation}
\boldsymbol{Y}(n)=(p\boldsymbol{I}+\bLam)^{-n}\boldsymbol Y(0)(p\boldsymbol{I}+\bLam)^n-\frac{np\lambda}{p\boldsymbol{I}+\bLam} \;.
\end{equation}
As a corollary, we have that the solutions can be found from the secular problem for the matrix: 
\begin{equation}\label{Yn3}
\boldsymbol{Y}(0)-np\lambda(p\boldsymbol{I}+\boldsymbol L_0(0))^{-1}\;, 
\end{equation}
i.e., the eigenvalues of $\boldsymbol{Y}(n)$ are the values of the particle positions at discrete-time $n$. 
%
\subsection{Commuting discrete flows} 
Following the construction in \cite{Sikarin1} we introduce another temporal Lax matrix $\boldsymbol N_\kappa$ 
which generates a shift $\;\;\widehat{\phantom{a}}\;\;$ in an additional discrete time direction. Thus, we 
impose for the same vector function $\boldsymbol\phi$ as before, also the system of equations: 
%
\begin{equation}\label{KN} 
\wh{\boldsymbol \phi}=\boldsymbol N_\kappa\boldsymbol \phi\  \;,\;\;\;\;\;\boldsymbol{L}_{\kappa}\phi =\zeta \phi \;,
\end{equation}
where 
\begin{eqnarray}
\boldsymbol N_{\kappa}&=&\frac{\wh{h}h^T}{\kappa}+\boldsymbol N_0\;,\;\;\;\;\mbox{where}\;\;\;\;\boldsymbol N_0=\sum_{i,j=1}^N\frac{\wh{h}_ih_j}{\wt{x}_i-x_j+\lambda}E_{ij}\;.\label{aKN2}
\end{eqnarray}
This describes the flow in terms of an additional discrete-time variable $m$, where 
\emph{hat} is a shorthand notation for the discrete-time shift, i.e. for $x_i(n,m)=x_i$, and we write $x_i(n,m+1)=\wh{x}_i$, and 
 $x_i(n,m-1)=\hypohat 0 {x}_i$. 
\\
\\
Obviously, the compatibility relations for \eqref{KN} can be analysed in a very similar manner as to the ones for \eqref{LM}. 
Thus, we find a coupled system of equations in terms of the variables $h_i$, and $x_i$ in the form
\begin{subequations}\label{idenforh2as}
\begin{eqnarray}
\sum_{j=1}^N
\left(\frac{\wh{h}_j^2}{\wh{x}_i-\wh{x}_j+\lambda }-\frac{h_j^2}{\wh{x}_i-x_j+\lambda }\right)&=&-q\;,\;\;\;\;\forall i\;, \\
\sum_{j=1}^N
\left(\frac{h_j^2}{x_j-x_l+\lambda }-\frac{\wh{h}_j^2}{\wh{x}_j-x_l+\lambda }\right)&=&-q\;,\;\;\;\;\forall l \;,
\end{eqnarray}
\end{subequations}
where the quantity $q=q(m)$ does not carry a particle label, but may still be a function of $m$. 
The system \eqref{idenforh2as} can be resolved once again by using the Lagrange interpolation formula 
\eqref{LaIn2} yielding the resolution: 
\begin{subequations} \label{keqs} 
\begin{eqnarray}\label{keqsa}
h_i^2&=&-q\frac{\prod_{j=1}^N(x_i-x_j+\lambda)(x_i-\wh{x}_j-\lambda)}{\prod_{j\ne i}^N(x_i-x_j)\prod_{j=1}^N(x_i-\wt{x}_j)}\;, \\
\wh{h}_i^2&=&q\frac{\prod_{j=1}^N(\wh{x}_i-x_j+\lambda)(\wh{x}_i-\wh{x}_j-\lambda)}{\prod_{j\ne i}^N(\wh{x}_i-\wh{x}_j)\prod_{j=1}^N(\wh{x}_i-x_j)} \;, 
\label{keqsb} 
\end{eqnarray}\end{subequations} 
for $i=1,2,...,N$, and from which we obtain the following system of equations
\begin{equation}\label{eqmotion12}
\frac{q }{\hypohat 0 {q}}\prod\limits_{\mathop {j = 1}\limits_{j \ne i} }^N \frac{(x_i-x_j+\lambda)}{(x_i-x_j-\lambda)}=
\prod_{j=1}^N\frac{(x_i-\wh{x}_j)(x_i-{\hypohat 0 x}_j+\lambda)}{(x_i-{\hypohat 0 x}_j)(x_i-\wh{x}_j-\lambda)}\;.
\end{equation}
The product version \eqref{eqmotion12} of \eqref{keqs}, thus yields a system of $N$ equations for $N+1$ unknowns, $x_1,...,x_N$ and $q$. 
There is again no equation for $q $ separately, and thus it should be a priori given in order to get a closed set of equations. Thus far, so similar. 

Assuming now that the dependent variables depend simultaneously on both discrete time variables $n$ and $m$, then to obtain a univalent solution of the 
equations of motion we must require that both flows, in the ``$\;\;\wt{}\;\;$" direction  and the ``$\;\;\wh{}\;\;$" direction, commute. If so, then we can fix a 
value for $n$ and solve the equations in the ``$\;\;\wh{}\;\;$" direction similarly as before, leading to the matrix $\boldsymbol{Y}$ which depends on $n$ and $m$ 
as follows: 
\\
\begin{prop}
Let the $N\times N$ matrix function of the discrete variable $m$, $\boldsymbol{Y}(m)$, be given by  
\begin{subequations}
\begin{equation}\label{Ym2}
\boldsymbol{Y}(m)=\left[\prod_{k=0}^{m-1}(q(k)\boldsymbol{I}+\bLam)^{-1}\right]\,\boldsymbol{Y}(0)\left[\prod_{k=0}^{m-1}(q(k)\boldsymbol{I}+\bLam)\right]-\sum_{k=0}^{m-1}\frac{q(k)\lambda}{q(k)\boldsymbol{I}+\bLam} \;,
\end{equation}
subject to the following condition on the initial value matrix
\begin{equation}\label{constraint(m)}
[\boldsymbol Y(0),\bLam]=\lambda\bLam+\mbox{\rm rank 1}\;.
\end{equation}
\end{subequations}
Then eigenvalues $x_i(m)$ of the matrix $\boldsymbol{Y}(m)$ coincide with the solutions for particle position of the discretetime RS system, i.e. 
they solve the discrete equations of motion \eqref{eqmotion12}. 
\end{prop}
From now on we will restrict ourselves for simplicity to the case of constant $q$: $q=\widehat{q}$ leading to 
what we would like to call the discrete-time RS system corresponding to 
the ``$\;\;\wh{}\;\;$" direction in terms of the discrete-time variable $m$ and the matrix $\boldsymbol{Y}(m)$ becomes
\begin{equation}
\boldsymbol{Y}(m)=(q\boldsymbol{I}+\bLam)^{-m}\boldsymbol Y(0)(q\boldsymbol{I}+\bLam)^m-\frac{mq\lambda}{q\boldsymbol{I}+\bLam} \;.
\end{equation}

In order for this scenario to work there must be further constraints on the flows. This will lead to a system of ``constraints'' which can be readily 
obtained from the compatibility between Lax pairs \eqref{aLM2} and \eqref{aKN2}

%
%
%
%
%
\begin{subequations}
\begin{eqnarray}
\frac{p}{q}&=&\prod_{j=1}^N\frac{(x_i-\wh{x}_j-\lambda)(x_i-\wt{x}_j)}{(x_i-\wt{x}_j-\lambda)(x_i-\wh{x}_j)}\;,\label{CON1}\\
\frac{p}{q }&=&\prod_{j=1}^N\frac{(x_i-{\hypohat 0 x}_j+\lambda)(x_i-{\hypotilde 0 x}_j)}{(x_i-{\hypotilde 0 x}_j+\lambda)(x_i-{\hypohat 0 x}_j)}\;.\label{CON2}
\end{eqnarray}\end{subequations} 
We will refer to relations \eqref{CON1} and \eqref{CON2} as the \emph{constraint equations} which guarantee the commutativity between the discrete-time flows with 
shifts ``$\;\;\wt{}\;\;$" and ``$\;\;\wh{}\;\;$" in the variables $n$ and $m$ respectively, and will play a major role in the discrete-time variational principle (see Appendix \ref{example}). Equating \eqref{CON1} with \eqref{CON2}, we get
\begin{equation}\label{eq12}
\prod_{j=1}^N\frac{(x_i-\wt{x}_j)(x_i-{\hypotilde 0 x}_j+\lambda)}{(x_i-{\hypotilde 0 x}_j)(x_i-\wt{x}_j-\lambda)}=
\prod_{j=1}^N\frac{(x_i-\wh{x}_j)(x_i-{\hypohat 0 x}_j+\lambda)}{(x_i-{\hypohat 0 x}_j)(x_i-\wh{x}_j-\lambda)}\;,
\end{equation}
which is a consequence of \eqref{eqmotion1} and \eqref{eqmotion12}, and which expresses the compatibility with the set of O$\Delta$Es. 
\begin{prop}
 The eigenvalues $x_1(n,m), \dots, x_N(n,m)$ of the $N\times N$ matrix 
\begin{subequations}
\begin{eqnarray}\label{EXACT}
\bsY(n,m)&=&(p\bI+\bLam)^{-n}(q\bI+\bLam)^{-m}\bsY(0,0)(p\bI+\bLam)^{n}(q\bI+\bLam)^{m}\nn\\
&&-np\lambda(p\bI+\bLam)^{-1}-mq\lambda(q\bI+\bLam)^{-1}\  
\end{eqnarray} 
in which the initial value matrix $\bsY(0,0)$ is subject to the condition 
\begin{equation}\label{EXACT2}
[\bsY(0,0)\,,\,\bLam]=\lambda\bLam + {\rm rank\,1} \  , 
\end{equation}
\end{subequations} 
obey simultaneously both the discrete-time Ruijsenaars-Schneider systems given by \eqref{eqmotion1} and \eqref{eqmotion12} as well as the systems of constraint equations given by \eqref{CON1} and \eqref{CON2}.
\end{prop}

In order to make a connection with an initial value problem, we mention that the initial value matrix 
$\bsY(0,0)$ can be obtained from the diagonal matrix of initial values $\bsX(0,0)$ by a similarity transformation with a matrix $\bU(0,0)$ 
which is an invertible matrix diagonalizing the initial Lax matrix $\bL(0,0)$. To find the latter, we need the initial values 
$x_i(0,0)$, $x_i(1,0)$ and $x_i(0,1)$, ($i=1,\dots,N$)\footnote{The description of the initial value problem can be imposed the same fashion with the CM case \cite{Sikarin1}}. We note that the secular problem can, hence, be reformulated as one for the following 
matrix 
\begin{equation}\label{XLK}
\boldsymbol{\mathcal X}(n,m)=\bsX(0,0)-np\lambda \left(p\bI+\bL(0,0)\right)^{-1}-mq\lambda
\left(\bI+\bK(0,0)\right)^{-1}\  ,  
\end{equation} 
and hence the solution is provided by the roots of the characteristic equation:
\begin{equation}\label{Sec}
P_{\boldsymbol{\mathcal X}}(x)=\det(x \boldsymbol I- \boldsymbol{\mathcal{X}}(n,m))=\prod\limits_{i=1}^N(x-x_i(n,m))\;. 
\end{equation}

\noindent 
\textbf{Remark 1:} We would like to mention that in the more general case where  $p$ and $q$ may depend 
nontrivially on $n$ and $m$, they should be subject to compatibility relations between \eqref{Yn2} and \eqref{Ym2}, 
which produces the conditions
\[
\wh{p}q=\wt{q}p\;\;\;\;\mbox{and}\;\;\;\wh{p}+q=\wt{q}+p\;.
\]
From these two equations, the only possible answers would be $p=p(n)$ and $q=q(m)$ implying that $p$ and $q$ can only 
depend on the discrete variable associated with their own respective directions. 
\vspace{.2cm} 

\noindent 
\textbf{Remark 2:} 
In order to understand why we can regard the model described in this section as relativistic version of the 
discrete-time Calogero-Moser system, we perform the non-relativistic limit which is obtained by letting the 
parameter $\lambda\rightarrow 0$. To implement the limit, we note that as a 
function of $\lambda$ the variable $h_i^2$ as given in \eqref{heqsa} behave as:   
\begin{equation}
h_i^2\rightarrow -p\lambda \left[ 1+\lambda\mathsf{p}_i+\mathcal{O}(\lambda^2)\, \right] \;,
\end{equation}
where $\mathsf{p}_i$ are momenta for the discrete-time Calogero-Moser system \cite{FP} given by
\begin{equation}\mathsf{p}_i=\sum_{j=1\atop j \neq i}^N\frac{1}{x_i-x_j}-\sum_{j=1}^N\frac{1}{x_i-\wt{x}_j}\;.\end{equation}
The spatial Lax matrix $\boldsymbol L_\kappa$ given in \eqref{LM1} becomes
\begin{equation}
\boldsymbol L_\kappa\rightarrow-p\boldsymbol{I}-p\lambda\left( \frac{1}{\kappa}\boldsymbol {E}+\boldsymbol L_{CM}\right)+\mathcal{O}(\lambda^2)\;,
\end{equation}
in which $\boldsymbol {E}=\sum_{i,j} E_{ij}$ and 
where $\boldsymbol L_{CM}$ is the spatial Lax matrix for the discrete-time Calogero-Moser system \cite{FP} given by
\begin{equation}
\boldsymbol L_{CM}=\sum_{i=1}^N\mathsf{p}_iE_{ii}+\sum_{i \neq j}^N\frac{E_{ij}}{x_i-x_j}\;.
\end{equation}
The temporal Lax matrix $\boldsymbol M_\kappa$ given in \eqref{aLM2} expands  as
\begin{equation}
\boldsymbol M_\kappa\rightarrow-p\lambda\left( \frac{1}{\kappa}\boldsymbol {E}+\boldsymbol M_{CM}\right)+\mathcal{O}(\lambda^2)\;,
\end{equation}
where $\boldsymbol M_{CM}$ is the temporal Lax matrix for the discrete-time Calgoero-Moser system \cite{FP} given by
\begin{equation}
\boldsymbol M_{CM}=-\sum_{i,j=1}^N\frac{E_{ij}}{\wt{x}_i-x_j}\;.
\end{equation}
Thus, as $\lambda\rightarrow 0$, the compatibility \eqref{coM} produces
\begin{equation}
\wt{\boldsymbol L}_{CM}\;\frac{\boldsymbol {E}}{\kappa}+\frac{\boldsymbol {E}}{\kappa}\;\boldsymbol M_{CM}+\wt{\boldsymbol L}_{CM}\boldsymbol M_{CM}=
{\boldsymbol M}_{CM}\;\frac{\boldsymbol {E}}{\kappa}+\frac{\boldsymbol {E}}{\kappa}\;\boldsymbol L_{CM}+{\boldsymbol M}_{CM}\boldsymbol L_{CM}\;
\end{equation}
consequently implying
\begin{eqnarray}
\wt{\boldsymbol L}_{CM}\boldsymbol M_{CM}&=&{\boldsymbol M}_{CM}\boldsymbol L_{CM}\;,\\
\left(\wt{\boldsymbol L}_{CM}-\boldsymbol M_{CM}\right)\boldsymbol {E}&=&\boldsymbol {E}\left(\boldsymbol M_{CM}-\boldsymbol L_{CM} \right)\;.
\end{eqnarray}
These two equations give what we know as the discrete-time equations of motion corresponding to the Calogero-Moser system \cite{FP}.

With the suitable choices of $\boldsymbol\Lambda\rightarrow -e^{-\lambda\boldsymbol\Lambda_{CM}}$, 
$p\rightarrow e^{-\lambda p_{CM}}$ and $q\rightarrow e^{-\lambda q_{CM}}$, up to order 
$\mathcal{O}(\lambda)$, the exact solution \eqref{EXACT} becomes
\begin{equation}
\boldsymbol{Y}(n,m)\rightarrow \boldsymbol Y(0,0)-\frac{n}{p_{CM}\boldsymbol{I}+\boldsymbol\Lambda_{CM}}-\frac{m}{q_{CM}\boldsymbol{I}+\boldsymbol\Lambda_{CM}}\;,
\end{equation}
and  in oder $\mathcal{O}(\lambda)$, \eqref{EXACT2} yields 
\begin{equation}\label{EXACT23}
[\bsY(0,0)\,,\,\bLam_{CM}]=\boldsymbol{I}+ {\rm rank\,1} \  .
\end{equation}
These two equations are just identical to the defining relations for the exact solution for the 
discrete-time Calogero-Moser system \cite{Sikarin1}. Thus, both the discrete Lax representation as well 
as the solution for the discrete CM model is obtained through the above limits from the 
discrete-time Ruijsenaars-Schneider model. In the continuum case the non-relativistic limit of the 
Ruijsenaars model was discussed in \cite{Braden}.


\section{The Lagrangian 1-form and the closure relation}\label{Lagrange}
\setcounter{equation}{0} 

In this section we consider the Lagrange formulation of the discrete multi-time RS model and show that it 
possesses a Lagrange 1-form structure.
In the CM case \cite{Sikarin1}, we obtained Lagrangians 1-form structure through the connection of the Lax 
representation. Here we also have the Lagrangian 1-form structure for the RS system, but the establishment is more 
difficult as connection through the Lax representation is no longer relevant. In this Section, we will first derive the 
Lagrangian 1-form for the discrete-time RS system and then establish the closure relation. 
%
%
It is easy to show that the actions corresponding to equations of motion \eqref{eqmotion1} and \eqref{eqmotion12} are given by
\begin{subequations}\label{actions}\begin{eqnarray}
{S}_{(n)}&=&\sum_{n}\mathscr{L}_{(n)}(\boldsymbol{x}(n),\boldsymbol{x}(n+1))\;,\label{actionn}\\
{S}_{(m)}&=&\sum_{m}\mathscr{L}_{(m)}(\boldsymbol{x}(m),\boldsymbol{x}(m+1))\;,\label{actionm}
\end{eqnarray}\end{subequations} 
where
\begin{subequations}\label{Lagrangians} \begin{eqnarray}
\mathscr{L}_{(n)}&=&\sum_{i,j=1}^N\left( f(x_i-\wt{x}_j)-f(x_i-\wt{x}_j-\lambda)\right)-\frac{1}{2}\sum\limits_{\mathop {i,j = 1}\limits_{j \ne i} }^N\left(f(x_i-x_j+\lambda)\right.\nn\\
&&\left.+f(\wt x_i-\wt x_j+\lambda)\right)-\ln\left|p\right|(\Xi -\wt{\Xi }) \;,\label{2Lagn}\\
\mathscr{L}_{(m)}&=&\sum_{i,j=1}^N\left( f(x_i-\wh{x}_j)-f(x_i-\wh{x}_j-\lambda)\right)-\frac{1}{2}\sum\limits_{\mathop {i,j = 1}\limits_{j \ne i} }^N\left(f(x_i-x_j+\lambda)\right.\nn\\
&&\left.+f(\wh x_i-\wh x_j+\lambda)\right)-\ln\left|q \right|(\Xi -\wh{\Xi })\;,\label{2Lagm}
\end{eqnarray}\end{subequations} 
with $\Xi =\sum_{i=1}^Nx_i$ and the function $f(x)$ is given by $f(x)=x\ln(x)$. We assume in this and the following sections that the parameters $p$ and $q$ are constant.
The discrete-time Euler-Lagrange equations read
\begin{eqnarray}
\frac{\partial\mathscr{L}_{(n)}}{\partial{\wt{x}_i}}+\wt{\frac{\partial\mathscr{L}_{(n)}}{\partial{x_i}}}=0\;,\;\;\;\mbox{and}\;\;\;
\frac{\partial\mathscr{L}_{(m)}}{\partial{\wh{x}_i}}+\wh{\frac{\partial\mathscr{L}_{(m)}}{\partial{x_i}}}=0\;,\nn
\end{eqnarray}
which lead to \eqref{eqmotion1} and \eqref{eqmotion12}, respectively.

The additional terms $\ln\left|p\right|(\Xi -\wt{\Xi })$ in \eqref{2Lagn} and $\ln\left|q\right|(\Xi -\wh{\Xi })$ in \eqref{2Lagm}, containing the
differences of the centre of mass, are needed in order to account for the constraint
equations \eqref{CON1} and \eqref{CON2} as they arrive from the Euler-Lagrange (EL) equations on discrete curves, which is a
connected collection of line segments (i.e. elementary links on the lattice) with or without end
points (i.e. closed or non-closed), involving corners (vertices connecting line segments with
different directions).
\\
\begin{theorem}
For the Lagrangians \eqref{2Lagn} and \eqref{2Lagm}, the closure relation
\label{closure}
\begin{equation}\label{closure}
 \widehat{\mathscr{L}_{(n)}(\boldsymbol{x},\widetilde{\boldsymbol{x}})} - \mathscr{L}_{(n)}(\boldsymbol{x},\widetilde{\boldsymbol{x}}) -
 \widetilde{\mathscr{L}_{(m)}(\boldsymbol{x},\widehat{\boldsymbol{x}})} + \mathscr{L}_{(m)}(\boldsymbol{x},\widehat{\boldsymbol{x}}) = 0\;,
\end{equation}
holds on the solutions of the equations of motion \eqref{eqmotion1} and \eqref{eqmotion12} as well as the constraint equations \eqref{CON1} and 
\eqref{CON2}.
\end{theorem}

\begin{proof} \eqref{closure} can be written in the form
\begin{subequations}
\begin{eqnarray}\label{closureproof1}
 && \widehat{\mathscr{L}_{(n)}(\boldsymbol{x},\widetilde{\boldsymbol{x}})} - \mathscr{L}_{(n)}(\boldsymbol{x},\widetilde{\boldsymbol{x}}) -
 \widetilde{\mathscr{L}_{(m)}(\boldsymbol{x},\widehat{\boldsymbol{x}})} + \mathscr{L}_{(m)}(\boldsymbol{x},\widehat{\boldsymbol{x}})\;\nn\\
 &&=\sum_{i,j=1}^N\wh{x}_i\left(\ln\left|\frac{\wh{x_i}-\wh{\wt x}_j}{\wh{x}_i-\wh{\wt{x}}_j-\lambda} \frac{\wh{x_i}-x_j+\lambda}{\wh{x}_i-x_j}\right| -\ln\left| \frac{\wh{x}_i-\wh{x}_j+\lambda}{\wh{x}_i-\wh{x}_j-\lambda}\right|\right)\nn\\
 &&-\sum_{i,j=1}^N\wt{x}_i\left(\ln\left|\frac{\wt{x_i}-\wh{\wt x}_j}{\wt{x}_i-\wh{\wt{x}}_j-\lambda} \frac{\wt{x_i}-x_j+\lambda}{\wt{x}_i-x_j}\right| -\ln\left| \frac{\wt{x}_i-\wt{x}_j+\lambda}{\wt{x}_i-\wt{x}_j-\lambda}\right|\right)\nn\\
&&+\sum_{i,j=1}^N\left(\wh{\wt{x}}_i\ln\left|\frac{\wt{x}_i-\wh{\wt{x}}_j-\lambda}{\wt{x_i}-\wh{\wt x}_j} \frac{\wt{x}_i-\wh{\wt{x}}_j}{\wt{x_i}-\wh{\wt{x}}_j-\lambda}\right| -x_i\ln\left| \frac{x_i-\wt{x}_j}{x_i-\wt{x}_j-\lambda}\frac{x_i-\wh{x}_j-\lambda}{x_i-\wh{x}_j}\right|\right)\nn\\
&&+\left(\ln\left|q\right|-\ln\left| p\right|\right)\left( \wt\Xi-\wh{\wt\Xi}-\Xi+\wh\Xi\right)\nn\\
&&+\lambda\sum_{i,j=1}^N\left( \ln\left| \frac{\wh{x}_i-\wh{\wt x}_j-\lambda}{\wt{x}_i-\wh{\wt x}_j-\lambda}\frac{x_i-\wh{x}_j-\lambda}{x_i-\wt{x}_j-\lambda}\right|-\ln\left|\frac{\wh{x}_i-\wh{x}_j+\lambda}{\wt{x}_i-\wt{x}_j+\lambda} \right|\right)\;,
\end{eqnarray}
Using \eqref{eqmotion1}, \eqref{eqmotion12}, \eqref{CON1} and \eqref{CON2}, we have
\begin{eqnarray}\label{closureproof2}
 &&\widehat{\mathscr L_{(n)}(\boldsymbol{x},\widetilde{\boldsymbol{x}})} - \mathscr L_{(n)}(\boldsymbol{x},\widetilde{\boldsymbol{x}}) -
 \widetilde{\mathscr L_{(m)}(\boldsymbol{x},\widehat{\boldsymbol{x}})} + \mathscr L_{(m)}(\boldsymbol{x},\widehat{\boldsymbol{x}})\;\nn\\
 &&=\sum_{i=1}^N(\wh{x}_i+\wt{x}_i-\wh{\wt x}_i-x_i)\ln\left| \frac{p}{q  }\right|
 +\ln\left| \frac{q }{p }\right|\left( \wt\Xi-\wh{\wt\Xi}-\Xi+\wh\Xi\right)\nn\\
&&+\lambda\sum_{i,j=1}^N\left( \ln\left| \frac{\wh{x}_i-\wh{\wt x}_j-\lambda}{\wt{x}_i-\wh{\wt x}_j-\lambda}\frac{x_i-\wh{x}_j-\lambda}{x_i-\wt{x}_j-\lambda}\right|-\ln\left|\frac{\wh{x}_i-\wh{x}_j+\lambda}{\wt{x}_i-\wt{x}_j+\lambda} \right|\right)\;.
\end{eqnarray}
Using the fact that the last line of \eqref{closureproof2} vanishes on the exact solution \eqref{EXACT} and $\wt\Xi-\wh{\wt\Xi}-\Xi+\wh\Xi=0$, then we have
\begin{eqnarray}\label{closureproof3}
 \widehat{\mathscr{L}_{(n)}(\boldsymbol{x},\widetilde{\boldsymbol{x}})} - \mathscr{L}_{(n)}(\boldsymbol{x},\widetilde{\boldsymbol{x}}) -
 \widetilde{\mathscr{L}_{(m)}(\boldsymbol{x},\widehat{\boldsymbol{x}})} + \mathscr{L}_{(m)}(\boldsymbol{x},\widehat{\boldsymbol{x}})=0\;.
\end{eqnarray}
\end{subequations}
\end{proof}
%
\begin{figure}[h]
\begin{center}
\begin{tikzpicture}[scale=0.5]
 \draw[->] (0,0) -- (8,0) node[anchor=west] {$n_{i}$};
 \draw[->] (0,0) -- (0,8) node[anchor=south] {$n_{j}$};
 \fill (4,1) circle (0.2);
 \draw (4,1) --(1,1)  -- (1,2) -- (3,2) -- (3,3) -- (4,3) -- (4,4) -- (5,4) -- (5,6) -- (6,6);
 \fill (6,6) circle (0.2);
 \draw (4,4)node[anchor=north west] {$\Gamma$};
\end{tikzpicture}
\begin{tikzpicture}[scale=0.5]
 \draw[white] (0,0) -- (6.5,0);
 \draw[white] (0,0) -- (0,6);
 \draw[->,thick] (2,4) -- (4,4);
\end{tikzpicture}
\begin{tikzpicture}[scale=0.5]
 \draw[->] (0,0) -- (8,0) node[anchor=west] {$n_{i}$};
 \draw[->] (0,0) -- (0,8) node[anchor=south] {$n_{j}$};
\fill (4,1) circle (0.2);
 \draw (4,1) --(1,1)  -- (1,2) -- (4,2) -- (4,4) -- (5,4) -- (5,6) -- (6,6);
 \draw[dashed] (3,2) -- (3,3) -- (4,3);
 \fill (6,6) circle (0.2);
  \draw(4,4)node[anchor=north west] {$\Gamma^{\prime}$};
\end{tikzpicture}
\end{center}
\caption{Deformation of the discrete curve $\Gamma$.}\label{curve_deformation}
\end{figure}
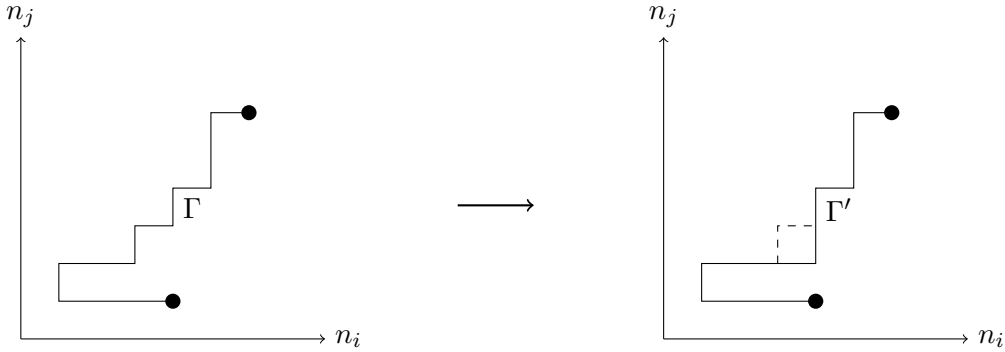
%
In \cite{Sikarin1} we described what we mean by the Lagrangian 1-form, but let us reiterate this here for the sake of self-contained of this paper. 
\\
\\
\textbf{Definition}.\emph{
Let $\boldsymbol{e}_{i}$ represent the unit vector in the lattice direction labeled by $i$ and let any position in the 
lattice be identified by the vector $\boldsymbol{n}$, so that an elementary shift in the lattice can be created by the
 operation $\boldsymbol{n}\mapsto\boldsymbol{n}+\boldsymbol{e}_{i}$. Since the Lagrangian depends on $\boldsymbol{x}$ and 
its elementary shift in one discrete direction, it can be associated with an oriented vector $\boldsymbol{e}_i$ on a curve 
$\Gamma_i(\boldsymbol{n})=(\boldsymbol{n},\boldsymbol{n}+\boldsymbol{e}_i)$, and we can treat these Lagrangians as defining 
a discrete 1-form $\mathscr{L}_i(\boldsymbol{n})$
\begin{equation}\label{eq:3.12}
\mathscr{L}_i(\boldsymbol{n})=\mathscr{L}_i(\boldsymbol{x}(\boldsymbol{n}),\boldsymbol{x}(\boldsymbol{n}+\boldsymbol{e}_{i})),
\end{equation}
which satisfies the following relation
\begin{eqnarray}\label{d}
&& \mathscr{L}_{i}(\boldsymbol{x}(\boldsymbol{n}+\boldsymbol{e}_{j}),
                                                      \boldsymbol{x}(\boldsymbol{n}+\boldsymbol{e}_{i}+\boldsymbol{e}_{j}))
                                     -\mathscr{L}_{i}(\boldsymbol{x}(\boldsymbol{n}),\boldsymbol{x}(\boldsymbol{n}+\boldsymbol{e}_{i}))\nn\\
                                   &&\;\;\;\;\;\;\;\;-\mathscr{L}_{j}(\boldsymbol{x}(\boldsymbol{n}+\boldsymbol{e}_{i}),
                                                      \boldsymbol{x}(\boldsymbol{n}+\boldsymbol{e}_{j}+\boldsymbol{e}_{i}))
                                     +\mathscr{L}_{j}(\boldsymbol{x}(\boldsymbol{n}),\boldsymbol{x}(\boldsymbol{n}+\boldsymbol{e}_{j}))=0\;.
\end{eqnarray}
Equation \eqref{d} represents the closure relation of the Lagrangian $1$-form for the RS system and it can be explicitly shown holding on the level of the equations of motion, and as well as constraints.}
\\
\\
Choosing a discrete curve $\Gamma$ consisting of connected elements $\Gamma_{i}$, we can define an action on the curve by summing up the contributions $\mathscr{L}_i$ from each of the oriented links $\Gamma_{i}$ in the curve, to get
\begin{equation}\label{eq:3.13}
S(\boldsymbol{x}(\textbf{n});\Gamma) = \sum\limits_{\boldsymbol{n} \in \Gamma } {\mathscr{L}_i(\boldsymbol{x}(\boldsymbol{n}),\boldsymbol{x}(\boldsymbol{n} +\boldsymbol{e}_i ))}. 
\end{equation}
The closure relation \eqref{d} is actually equivalent to the invariance of the action under local deformations of 
the curve. To see this, suppose we have an action $ S$ evaluated on a curve $\Gamma$, and we deform this 
(keeping end points fixed) to get a curve $\Gamma^{\prime}$ on which an action $ S^{\prime}$ is evaluated,
 such as in Figure \ref{curve_deformation}.
\\
\\
Then $ S^{\prime}$ is related to $ S$ by the following:
\begin{eqnarray}\label{eq:3.14ss}
 S^{\prime} & = &  S -\mathscr{L}_{i}(\boldsymbol{x}(\boldsymbol{n}+\boldsymbol{e}_{j}),
                                                      \boldsymbol{x}(\boldsymbol{n}+\boldsymbol{e}_{i}+\boldsymbol{e}_{j}))
                                     +\mathscr{L}_{i}(\boldsymbol{x}(\boldsymbol{n}),\boldsymbol{x}(\boldsymbol{n}+\boldsymbol{e}_{i}))\nn\\
                                  && +\mathscr{L}_{j}(\boldsymbol{x}(\boldsymbol{n}+\boldsymbol{e}_{i}),
                                                      \boldsymbol{x}(\boldsymbol{n}+\boldsymbol{e}_{j}+\boldsymbol{e}_{i}))
                                     -\mathscr{L}_{j}(\boldsymbol{x}(\boldsymbol{n}),\boldsymbol{x}(\boldsymbol{n}+\boldsymbol{e}_{j})).
\end{eqnarray}
Equation \eqref{eq:3.14ss} shows that the independence of the action under such a deformation is locally equivalent to the closure relation. The invariance of the action under the local deformation is a crucial aspect of the underlying variational principle.

The basic relations constituting the discrete multi-time EL equations were first given in 2011 in 
\cite{Rinthesis}\footnote{Chapter 3 \emph{The variational principle for Lagrangian 1-form} of \cite{Rinthesis} 
provides the system of actions on the elementary discrete curves as indicated in Fig.  \ref{curve12}, together 
with the corresponding EL equations. In a later paper \cite{Suris}, which appeared after these results were presented 
at the SIDE X (2012) meeting by the first author, these equations were restated as a Theorem and applied to a discrete-time Toda system.}  
and arise as the EL equations for actions on a set of basic curves given in Fig. \ref{curve12}. 
We now use the Lagrangian in \eqref{2Lagn} 
$\mathscr{L}_{(n)} $ and the Lagrangian in \eqref{2Lagm} $\mathscr{L}_{(m)}$ and derive the full set of EL equations for these basic curves and associated actions, and apply them to the 
case at hand of the discrete RS model. 
%
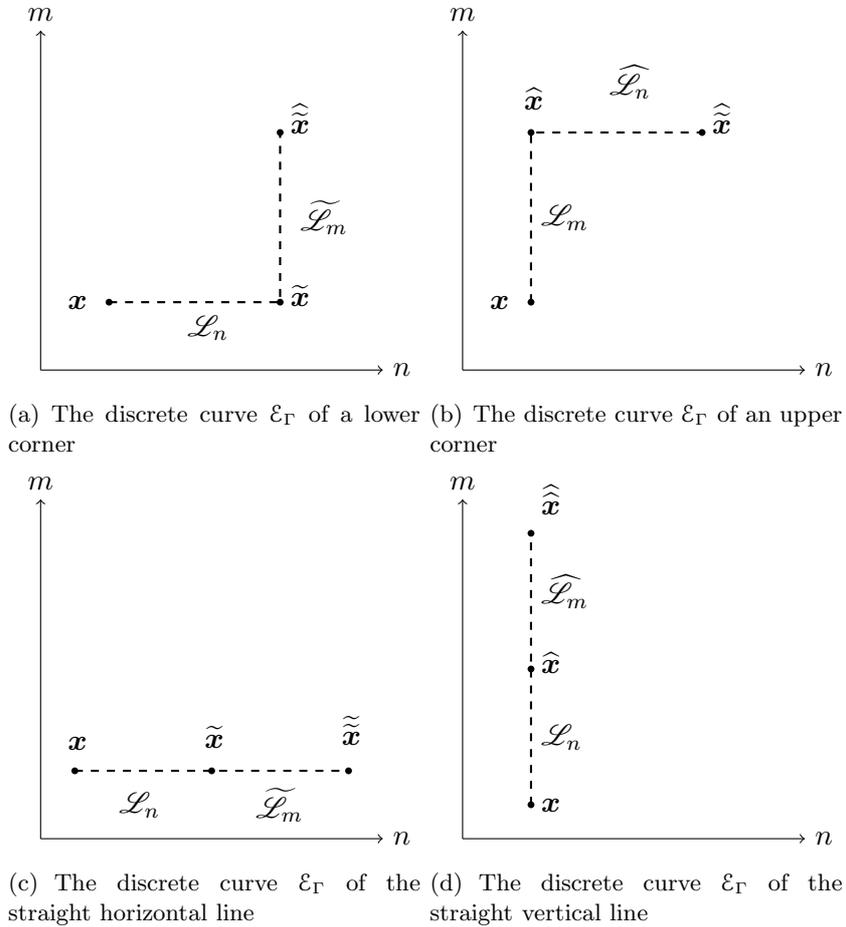
\begin{figure}[h]
\begin{center}
\subfigure[The discrete curve $\mathcal{E}_{\Gamma}$ of a lower corner]{
\begin{tikzpicture}[scale=0.45]
 \draw[->] (0,0) -- (10,0) node[anchor=west] {$n$};
 \draw[->] (0,0) -- (0,10) node[anchor=south] {$m$};
 \draw[thick,dashed] (2,2)--(7,2)--(7,7);
 \fill (2,2) circle (0.1);
 \fill (7,2) circle (0.1);
 \fill (7,7) circle (0.1);
 \draw (.5,2) node[anchor=west] {$\boldsymbol x$};
 \draw (7,2.2) node[anchor=west] {$\wt{\boldsymbol x}$};
 \draw (7,7.4) node[anchor=west] {$\wh{\wt{\boldsymbol x}}$};
 \draw (4,1.25) node[anchor=west] {$\mathscr{L}_{n}$};
 \draw (7.3,4.5) node[anchor=west] {$\wt{\mathscr{L}}_{m}$};
\end{tikzpicture}}
\subfigure[The discrete curve $\mathcal{E}_{\Gamma}$ of an upper corner]{
\begin{tikzpicture}[scale=0.45]
 \draw[->] (0,0) -- (10,0) node[anchor=west] {$n$};
 \draw[->] (0,0) -- (0,10) node[anchor=south] {$m$};
 \draw[thick,dashed] (2,2)--(2,7)--(7,7);
 \fill (2,2) circle (0.1);
 \fill (2,7) circle (0.1);
 \fill (7,7) circle (0.1);
 \draw (.5,2) node[anchor=west] {$\boldsymbol x$};
 \draw (1.5,8) node[anchor=west] {$\wh{\boldsymbol x}$};
 \draw (7,7.4) node[anchor=west] {$\wh{\wt{\boldsymbol x}}$};
 \draw (4,8.5) node[anchor=west] {$\wh{\mathscr{L}}_{n}$};
 \draw (2,4.5) node[anchor=west] {$\mathscr{L}_{m}$};
\end{tikzpicture}}
\subfigure[The discrete curve $\mathcal{E}_{\Gamma}$ of the straight horizontal line]{
\begin{tikzpicture}[scale=0.45]
 \draw[->] (0,0) -- (10,0) node[anchor=west] {$n$};
 \draw[->] (0,0) -- (0,10) node[anchor=south] {$m$};
 \draw[thick,dashed] (1,2)--(5,2)--(9,2);
 \fill (1,2) circle (0.1);
 \fill (5,2) circle (0.1);
 \fill (9,2) circle (0.1);
 \draw (.5,2.8) node[anchor=west] {$\boldsymbol x$};
 \draw (4.5,3) node[anchor=west] {$\wt{\boldsymbol x}$};
 \draw (8.5,3.2) node[anchor=west] {$\wt{\wt{\boldsymbol x}}$};
 \draw (2,1) node[anchor=west] {$\mathscr{L}_{n}$};
 \draw (6,1) node[anchor=west] {$\wt{\mathscr{L}}_{m}$};
\end{tikzpicture}}
\subfigure[The discrete curve $\mathcal{E}_{\Gamma}$ of the straight vertical line]{
\begin{tikzpicture}[scale=0.45]
 \draw[->] (0,0) -- (10,0) node[anchor=west] {$n$};
 \draw[->] (0,0) -- (0,10) node[anchor=south] {$m$};
 \draw[thick,dashed] (2,1)--(2,5)--(2,9);
 \fill (2,1) circle (0.1);
 \fill (2,5) circle (0.1);
 \fill (2,9) circle (0.1);
 \draw (2,1) node[anchor=west] {$\boldsymbol x$};
 \draw (2,5.2) node[anchor=west] {$\wh{\boldsymbol x}$};
 \draw (2,10) node[anchor=west] {$\wh{\wh{\boldsymbol x}}$};
 \draw (2,3) node[anchor=west] {$\mathscr{L}_{n}$};
 \draw (2,7.3) node[anchor=west] {$\wh{\mathscr{L}}_{m}$};
\end{tikzpicture}}
\end{center}
\caption{Simple discrete curves for $n$ and $m$ variables.}\label{curve12}
\end{figure}
\vspace{.3cm} 
\\
\textbf{case (a)}: The action for a discrete curve in Fig. \ref{curve12}(a) is
\begin{eqnarray}\label{DisEL}
S[\boldsymbol{x}]=\mathscr{L}_p(\boldsymbol{x},\wt{\boldsymbol{x}})+\wt{\mathscr{L}}_q(\wt{\boldsymbol{x}},\wh{\wt{\boldsymbol{x}}})\;.
\end{eqnarray}
We now vary the variable $\boldsymbol x \mapsto \boldsymbol x+\delta \boldsymbol x$ with the end points fixed: $\delta \boldsymbol x=0$ and $\delta \wh{\wt{\boldsymbol x}}=0$. Then the variation of the action is
\begin{eqnarray}\label{deltaS}
\delta S=\frac{\mathscr{L}_p}{\partial \boldsymbol{x}}\delta \boldsymbol{x}
+\frac{\mathscr{L}_p}{\partial \wt{\boldsymbol{x}}}\delta \wt{\boldsymbol{x}}
+\frac{\wt{\mathscr{L}}_q}{\partial \wt{\boldsymbol{x}}}\delta \wt{\boldsymbol{x}}
+\frac{\wt{\mathscr{L}}_q}{\partial \wh{\wt{\boldsymbol{x}}}}\delta \wh{\wt{\boldsymbol{x}}}\; 
=\left(\frac{\mathscr{L}_p}{\partial \wt{\boldsymbol{x}}}
+\frac{\wt{\mathscr{L}}_q}{\partial \wt{\boldsymbol{x}}}\right)\delta \wt{\boldsymbol{x}}\;.
\end{eqnarray}
The first and last terms vanish according to the condition on end points. The $\delta S=0$ once the coefficient of $\delta \wt{\boldsymbol x}$ is zero yielding
\begin{eqnarray}\label{ELnm}
\frac{\mathscr{L}_p(\boldsymbol{x},\wt{\boldsymbol{x}})}{\partial \wt{\boldsymbol{x}}}
+\frac{\wt{\mathscr{L}}_q(\wt{\boldsymbol{x}},\wh{\wt{\boldsymbol{x}}})}{\partial \wt{\boldsymbol{x}}}=0\;.
\end{eqnarray}
Using \eqref{2Lagn} and \eqref{2Lagm}, \eqref{ELnm} gives
\begin{equation}\label{eq12a}
\frac{p}{q }=\prod_{j=1}^N\frac{(x_i-{\hypohat 0 x}_j+\lambda)(x_i-{\hypotilde 0 x}_j)}{(x_i-{\hypotilde 0 x}_j+\lambda)(x_i-{\hypohat 0 x}_j)}\;.
\end{equation}
which is the constraint equations given in \eqref{CON1}.
\\
\\
\textbf{Case (b)}:
The action for a discrete curve in Fig. \ref{curve12}(b) is
\begin{eqnarray}\label{DisEL}
S[\boldsymbol{x}]=\mathscr{L}_q(\boldsymbol{x},\wh{\boldsymbol{x}})+\wh{\mathscr{L}}_p(\wh{\boldsymbol{x}},\wh{\wt{\boldsymbol{x}}})\;.
\end{eqnarray}
We now vary the variable $\boldsymbol x \mapsto \boldsymbol x+\delta \boldsymbol x$ with the end points fixed: $\delta \boldsymbol x=0$ and $\delta \wh{\wt{\boldsymbol x}}=0$. Then the variation of the action is
\begin{eqnarray}\label{deltaS2}
\delta S=\frac{\mathscr{L}_q}{\partial \boldsymbol{x}}\delta \boldsymbol{x}
+\frac{\mathscr{L}_q}{\partial \wh{\boldsymbol{x}}}\delta \wh{\boldsymbol{x}}
+\frac{\wh{\mathscr{L}}_p}{\partial \wh{\boldsymbol{x}}}\delta \wh{\boldsymbol{x}}
+\frac{\wh{\mathscr{L}}_p}{\partial \wh{\wt{\boldsymbol{x}}}}\delta \wh{\wt{\boldsymbol{x}}}\; 
=\left(\frac{\mathscr{L}_q}{\partial \wh{\boldsymbol{x}}}
+\frac{\wh{\mathscr{L}}_p}{\partial \wh{\boldsymbol{x}}}\right)\delta \wh{\boldsymbol{x}}\;.
\end{eqnarray}
The first and last terms vanish according to the condition end points. The $\delta S=0$ once the coefficient of $\delta \wt{\boldsymbol x}$ is zero yielding
\begin{eqnarray}\label{ELnm2}
\frac{\mathscr{L}_q(\boldsymbol{x},\wh{\boldsymbol{x}})}{\partial \wh{\boldsymbol{x}}}
+\frac{\wh{\mathscr{L}}_p(\wh{\boldsymbol{x}},\wh{\wt{\boldsymbol{x}}})}{\partial \wh{\boldsymbol{x}}}=0\;.
\end{eqnarray}
Using \eqref{2Lagn} and \eqref{2Lagm}, \eqref{ELnm2} gives
\begin{eqnarray}\label{CON22b} 
\frac{p}{q}&=&\prod_{j=1}^N\frac{(x_i-\wh{x}_j-\lambda)(x_i-\wt{x}_j)}{(x_i-\wt{x}_j-\lambda)(x_i-\wh{x}_j)}\;,
\end{eqnarray}
which is the constraint equations given in \eqref{CON2}.
\\
\\
\textbf{Case (c)}:
The action for a discrete curve in Fig. \ref{curve12}(c) is
\begin{eqnarray}\label{DisEL}
S[\boldsymbol{x}]=\mathscr{L}_p(\boldsymbol{x},\wt{\boldsymbol{x}})+\wt{\mathscr{L}}_p(\wt{\boldsymbol{x}},\wt{\wt{\boldsymbol{x}}})\;.
\end{eqnarray}
We now vary the variable $\boldsymbol x \mapsto \boldsymbol x+\delta \boldsymbol x$ with the end points fixed: $\delta \boldsymbol x=0$ and $\delta \wt{\wt{\boldsymbol x}}=0$. Then the variation of the action is
\begin{eqnarray}\label{deltaS2}
\delta S=\frac{\mathscr{L}_p}{\partial \boldsymbol{x}}\delta \boldsymbol{x}
+\frac{\mathscr{L}_p}{\partial \wt{\boldsymbol{x}}}\delta \wt{\boldsymbol{x}}
+\frac{\wt{\mathscr{L}}_p}{\partial \wt{\boldsymbol{x}}}\delta \wt{\boldsymbol{x}}
+\frac{\wt{\mathscr{L}}_p}{\partial \wt{\wt{\boldsymbol{x}}}}\delta \wh{\wt{\boldsymbol{x}}} 
=\left(\frac{\mathscr{L}_p}{\partial \wt{\boldsymbol{x}}}
+\frac{\wt{\mathscr{L}}_p}{\partial \wt{\boldsymbol{x}}}\right)\delta \wt{\boldsymbol{x}}\;.
\end{eqnarray}
The first and last terms vanish according to the condition end points. The $\delta S=0$ once the coefficient of $\delta \wt{\boldsymbol x}$ is zero yielding
\begin{eqnarray}\label{ELnm2c}
\frac{\mathscr{L}_p(\boldsymbol{x},\wt{\boldsymbol{x}})}{\partial \wt{\boldsymbol{x}}}
+\frac{\wt{\mathscr{L}}_p(\wt{\boldsymbol{x}},\wt{\wt{\boldsymbol{x}}})}{\partial \wt{\boldsymbol{x}}}=0\;.
\end{eqnarray}
Using \eqref{2Lagn} and \eqref{2Lagm}, \eqref{ELnm2c} gives
\begin{equation}\label{eqmotion1a}
\prod\limits_{\mathop {j = 1}\limits_{j \ne i} }^N \frac{(x_i-x_j+\lambda)}{(x_i-x_j-\lambda)}=\prod_{j=1}^N\frac{(x_i-\wt{x}_j)(x_i-{\hypotilde 0 x}_j+\lambda)}{(x_i-{\hypotilde 0 x}_j)(x_i-\wt{x}_j-\lambda)}\;.
\end{equation}
which is the equations of motion given by \eqref{eqmotion1}.
\\
\\
\textbf{Case (d)}:
The action for a discrete curve in Fig. \ref{curve12}(d) is
\begin{eqnarray}\label{DisEL}
S[\boldsymbol{x}]=\mathscr{L}_q(\boldsymbol{x},\wh{\boldsymbol{x}})+\wh{\mathscr{L}}_q(\wh{\boldsymbol{x}},\wh{\wh{\boldsymbol{x}}})\;.
\end{eqnarray}
We now vary the variable $\boldsymbol x \mapsto \boldsymbol x+\delta \boldsymbol x$ with the end points fixed: $\delta \boldsymbol x=0$ and $\delta \wh{\wh{\boldsymbol x}}=0$. Then the variation of the action is
\begin{eqnarray}\label{deltaS2}
\delta S=\frac{\mathscr{L}_q}{\partial \boldsymbol{x}}\delta \boldsymbol{x}
+\frac{\mathscr{L}_q}{\partial \wh{\boldsymbol{x}}}\delta \wh{\boldsymbol{x}}
+\frac{\wh{\mathscr{L}}_q}{\partial \wh{\boldsymbol{x}}}\delta \wh{\boldsymbol{x}}
+\frac{\wh{\mathscr{L}}_q}{\partial \wh{\wh{\boldsymbol{x}}}}\delta \wh{\wh{\boldsymbol{x}}} 
=\left(\frac{\mathscr{L}_q}{\partial \wh{\boldsymbol{x}}}
+\frac{\wh{\mathscr{L}}_q}{\partial \wh{\boldsymbol{x}}}\right)\delta \wh{\boldsymbol{x}}\;.
\end{eqnarray}
The first and last terms vanish according to the condition end points. The $\delta S=0$ once the coefficient of $\delta \wh{\boldsymbol x}$ is zero yielding
\begin{eqnarray}\label{ELnm2d}
\frac{\mathscr{L}_q(\boldsymbol{x},\wh{\boldsymbol{x}})}{\partial \wh{\boldsymbol{x}}}
+\frac{\wh{\mathscr{L}}_q(\wh{\boldsymbol{x}},\wh{\wh{\boldsymbol{x}}})}{\partial \wh{\boldsymbol{x}}}=0\;.
\end{eqnarray}
Using \eqref{2Lagn} and \eqref{2Lagm}, \eqref{ELnm2d} gives
\begin{equation}\label{eqmotion1b}
\prod\limits_{\mathop {j = 1}\limits_{j \ne i} }^N \frac{(x_i-x_j+\lambda)}{(x_i-x_j-\lambda)}=\prod_{j=1}^N\frac{(x_i-\wh{x}_j)(x_i-{\hypohat 0 x}_j+\lambda)}{(x_i-{\hypohat 0 x}_j)(x_i-\wh{x}_j-\lambda)}\;.
\end{equation}
which is the equations of motion given by \eqref{eqmotion12}.
\\
\\
In conclusion, the discrete variational principle which comprises the basic set of equations 
\eqref{ELnm}, \eqref{ELnm2}, \eqref{ELnm2c} and \eqref{ELnm2d},   produces the system of EL equations and 
constraints for the rational discrete-time RS model. Furthermore, the closure relation \eqref{closure} 
expresses the compatibility of these four basic equations, and as a consequence on the solutions of the variational 
system the action is stationary under deformations of any discrete curve (with fixed end points) such as 
indicated in Fig. \ref{curve_deformation}. In Appendix \ref{example}, we demonstrate how to derive explicitly 
the discrete Euler-Lagrange equation for some specific discrete curves.

\section{The semi-continuous limit: The skew limit}\label{skewlimit}
\setcounter{equation}{0} 
In this Section, we study a continuum analogue of a previous construction in Section \ref{exactsolution} by considering a particular semi-continuous limit.
Since the exact solution \eqref{EXACT} contains two discrete variables $n$ and $m$, we could perform a continuum limit on one of these 
variables separately, while leaving the other discrete variable intact, and thus obtain a semi-continuous equation with one remaining discrete 
and two continuous independent variables. Alternatively, we can first perform a change of independent variables on the lattice and subsequently perform 
the limit on one of the new variables. \emph{The advantage of the latter approach over the former is that it often leads in a more direct way to a hierarchy of 
higher order flows}. Adopting the latter approach in this section, we use a new discrete variable $\mathsf N :=n+m$, and perform the 
transformation on the dependent variables by setting $x(n,m)\mapsto \mathsf x(\mathsf N,m)=:{\mathsf x}$, which leads to the following expressions for 
the shifted variables:
\begin{eqnarray}\label{changingvariables}
x=x(n+1,m) &\mapsto & \mathsf x(\mathsf N+1,m)=:\wt{\mathsf x}\;, \nn\\
\widehat x=x(n,m+1) &\mapsto & \mathsf x(\mathsf N+1,m+1)=:\widehat{\widetilde{ {\mathsf x}}}\;,\nn\\ 
\widetilde x=x(n+1,m+1) &\mapsto & \mathsf x(\mathsf N+2,m+1)=:\widehat{\widetilde{\widetilde {\mathsf x}}}\;.\nn
\end{eqnarray}
Rearranging the terms in \eqref{EXACT}, we have
\begin{eqnarray}\label{EXACTSKEW}
\mathsf {\bsY}(\mathsf N,m)&=&(p\bI+\bLam)^{-\mathsf N}\left(\frac{q\bI+\bLam}{p\bI+\bLam}\right)^{-m}\left[\bsY(0,0)-\frac{\mathsf{N}p\lambda}{p\bI+\bLam}\right.\nn\\
&&\left.+m\lambda\left(\frac{p}{p\bI+\bLam}-\frac{q}{q\bI+\bLam}\right)\right](p\bI+\bLam)^{\mathsf N}\left(\frac{q\bI+\bLam}{p\bI+\bLam}\right)^{m}\;.  
\end{eqnarray} 
We perform the limit $n\rightarrow -\infty$, $m \rightarrow \infty$, $\varepsilon \rightarrow 0$ while
keeping $\mathsf N$ fixed and setting $\varepsilon=p-q$,
such that $\varepsilon m=\tau$ remains finite. Focusing on the penultimate factor in \eqref{EXACTSKEW} we have that
\begin{equation}\label{eq:4.2}
\lim\limits_{\mathop {m \rightarrow \infty} \limits_{\mathop {\varepsilon \rightarrow 0}\limits_{\varepsilon
 m \rightarrow \tau}} }\left(1-\frac{\varepsilon}{p\bI+\bLam} \right)^m =\lim\limits_{m \rightarrow \infty}\left( 1-\frac{\tau}{m(p\bI+\bLam)}
\right)^m=e^{-\frac{\tau}{p\bI+\bLam}}, 
\end{equation}
so that the exact solution takes the form
\begin{eqnarray}\label{EXACTSKEW2}
\mathsf {\bsY}(\mathsf N,\tau)=(p\bI+\bLam)^{-\mathsf N}e^{\frac{\tau}{p\bI+\bLam}}\left[\bsY(0,0)-\frac{\mathsf{N}p\lambda}{p\bI+\bLam}
+\frac{\tau\lambda\bLam}{(p\bI+\bLam)^2}\right](p\bI+\bLam)^{\mathsf N}e^{-\frac{\tau}{p\bI+\bLam}}\;. 
\end{eqnarray} 
This equation represents the full solution after taking the skew limit. The position of the particles $\mathrm x_i(\mathsf N,\tau )$ can be determined by computing the eigenvalues of \eqref{EXACTSKEW2}.
\subsection{The skew limit on equations of motion and constraints}
We first rewrite the equations of motion \eqref{eqmotion1}, taking \emph{$p$ to be constant}, in terms of the variables $(\mathsf{N},m)$ as follows 
\begin{eqnarray}
\sum\limits_{\mathop {j = 1}\limits_{j \ne i} }^N\left( \ln(\mathrm x_i-\mathrm x_j+\lambda)-\ln(\mathrm  x_i-\mathrm  x_j-\lambda)\right)&=&\sum_{j=1}^N\left(\ln(\mathrm  x_i-\wt{\wh{\mathrm  x}}_j)-\ln(\mathrm x_i-{\hypohat 0 {\hypotilde 0 {\mathrm {x}}}}_j+\lambda) \right.\nn\\
&&\left.+\ln(\mathrm x_i-{\hypotilde 0 {\hypohat 0 {\mathrm x}}}_j)-\ln(\mathrm  x_i-\wt{\wh{\mathrm  x}}_j-\lambda)\right)\;,\label{eqln2skew}
\end{eqnarray}
Introducing the notations $\wh{ \mathrm{x}}=\mathrm{x}(\mathsf N,\tau +\varepsilon )$ and $\hypohat 0 { \mathrm{x}}=\mathrm{x}(\mathsf N,\tau -\varepsilon)$ with the use of the Taylor expansion, we obtain
\begin{subequations}
\begin{eqnarray}\label{Taylor}
\mathrm{x}(\mathsf N,\tau\pm \varepsilon )=\mathrm{x}(\mathsf N,\tau)\pm \varepsilon\frac{\partial \mathrm{x}(\mathsf N,\tau)}{\partial {\tau}}+\frac{\varepsilon^2}{2}\frac{\partial^2 \mathrm{x}(\mathsf N,\tau)}{\partial {\tau^2}}\pm ...\;.
\end{eqnarray}
\end{subequations}
Collecting terms in order $\mathcal O(\varepsilon^0 )$, we have the equations of motion for the RS system corresponding to the ``$\;\mathsf N\;$" variable
\begin{eqnarray}\label{eqmotionskew1}
\sum\limits_{\mathop {j = 1}\limits_{j \ne i} }^N\left[ \ln(\mathrm x_i-\mathrm x_j+\lambda)-\ln(\mathrm  x_i-\mathrm  x_j-\lambda)\right]&=&\sum_{j=1}^N\left[\ln(\mathrm  x_i-\wt{\mathrm  x}_j)-\ln(\mathrm x_i-{{\hypotilde 0 {\mathrm {x}}}}_j+\lambda) \right.\nn\\
&&\;\;\;\;\;\;\;+\left.\ln(\mathrm x_i-{{\hypotilde 0 {\mathrm x}}}_j)-\ln(\mathrm  x_i-\wt{\mathrm  x}_j-\lambda)\right]\;,\label{eqln2skew}
\end{eqnarray}
and $\mathcal O(\varepsilon )$, we have
\begin{eqnarray}\label{eqmotionskew2}
\sum_{j=1}^N\left[ \frac{\partial {\wt {\mathrm x}_j}}{\partial \tau}\left( \frac{1}{\mathrm x_i-\wt{\mathrm x}_j-\lambda}-\frac{1}{\mathrm x_i-\wt{\mathrm x}_j}\right)+\frac{\partial {\hypotilde 0 {\mathrm x}_j}}{\partial \tau}\left( \frac{1}{\mathrm x_i-\hypotilde 0 {\mathrm x}_j+\lambda}-\frac{1}{\mathrm x_i-\hypotilde 0 {\mathrm x}_j}\right)\right]=0\;.
\end{eqnarray}
which are the equations of motion for the RS system corresponding to the ``$\;\;\tau\;\;$'' variable.
\\

Similarly, changing the variables $x(n,m)\mapsto \mathrm x(\mathrm{N},\tau )$ in constraints \eqref{CON1} and \eqref{CON2} and collecting terms in order $\mathcal O(\varepsilon )$, we have
\begin{subequations}
\begin{eqnarray}
-\frac{1}{p} &=&\sum_{j=1}^N\frac{\partial {\wt {\mathrm x}_j}}{\partial \tau}\left(\frac{1}{\mathrm x_i-\wt{\mathrm x}_j-\lambda}-\frac{1}{\mathrm x_i-\wt{\mathrm x}_j} \right)\;,\label{CONSKEW1}\\
\frac{1}{p} &=&\sum_{j=1}^N\frac{\partial {\hypotilde 0 {\mathrm x}_j}}{\partial \tau}\left(\frac{1}{\mathrm x_i-\hypotilde 0 {\mathrm x}_j+\lambda}-\frac{1}{\mathrm x_i-\hypotilde 0 {\mathrm x}_j}\right)\;,\label{CONSKEW2}
\end{eqnarray}
\end{subequations}
\eqref{CONSKEW1} and \eqref{CONSKEW2} represent the constraints after taking the skew limit. The summation of these two yields \eqref{eqmotionskew2}.
%
%
%

\subsection{The skew limit on action}
Next, to obtain the continuum limit of the action, we proceed exactly with the same steps as in \cite{Sikarin1}.
First, we observe that eq. \eqref{eqmotionskew1} can be once again be obtained by implementing the usual variational principle 
on the following action $ S_{(\mathsf N)}$ given by
\begin{eqnarray}\label{eq:4.17a}
{S}_{(\mathsf{N})}=\sum_{\mathsf N}\mathscr{L}_{(\mathsf N)}&=&\sum_{\mathsf N}\left(\sum_{i,j=1}^N\left( f(\mathrm x_i-\wt{\mathrm x}_j)-f(\mathrm x_i-\wt{\mathrm x}_j-\lambda)\right)-\frac{1}{2}\sum\limits_{\mathop {i,j = 1}\limits_{j \ne i} }^Nf(\mathrm x_i-\mathrm x_j+\lambda)\right.\nn\\
&&\left.-\frac{1}{2}\sum\limits_{\mathop {i,j = 1}\limits_{j \ne i} }^Nf(\wt{\mathrm x}_i-\wt{\mathrm x}_j+\lambda)-\ln\left|p\right|\sum_{i=1}^N(\mathrm x_i -\wt{\mathrm x}_i)\right) \;,\label{Lagn}
\end{eqnarray}
where now the Lagrangian $\mathscr{L}_{(\mathsf N)}$ involves variables $\wt{\mathsf x}_i$ shifted in the discrete variable $\mathsf N$ instead of 
the original variable $n$, and the corresponding discrete Euler-Lagrange equation reads: 
\begin{equation}\label{eq:4.20s}
\widetilde{\frac{\partial\mathscr{L}_{(\mathsf N)}}{\partial {\mathsf x}_i}}+
\left( \frac{\partial\mathscr{L}_{(\mathsf N)}}{\partial {\widetilde {\mathsf x}_i}}\right)=0 ,
\end{equation}
yielding \eqref{eqmotionskew1}.  

Second, we observe that eq. \eqref{eqmotionskew2} can be once again be obtained by implementing the usual variational principle 
on the following action $ S_{(\tau)}$ given by
\begin{equation}\label{actionskew3}
{S}_{(\tau)}=\int_{\tau_1}^{\tau_2}d\tau\mathscr{L}_{(\tau)}
\left(\boldsymbol{\mathrm{x}}(\mathsf N_{0}-1,\tau),\frac{\partial{\boldsymbol{\mathrm{x}}(\mathsf N_0,\tau)}}{\partial \tau}\right)\;,
\end{equation}
which is obtained by taking the skew limit together with anti-Taylor expansion of \eqref{actionm} and
\begin{eqnarray}\label{Lagrangeskew3}
\mathscr{L}_{(\tau)}&=&\sum_{i,j=1}^N\left( \frac{\partial \mathrm{\wt{x}}_j}{\partial \tau}(\ln\left|\mathrm x_i-\wt{\mathrm{x}}_j-\lambda\right|-\ln\left|\mathrm x_i-\wt{\mathrm{x}}_j\right|)\right)\nn\\
&&-\frac{1}{2}\sum\limits_{\mathop {i,j = 1}\limits_{j \ne i} }^N \left(\frac{\partial \mathrm{\wt{x}}_j}{\partial \tau}\left(\ln\left|\wt{\mathrm x}_i-\wt{\mathrm{x}}_j+\lambda\right|-\ln\left|\wt{\mathrm x}_i-\wt{\mathrm{x}}_j-\lambda\right| \right)+\frac{\partial{\wt x_i}}{\partial \tau }-\frac{\partial{\wt x_j}}{\partial \tau }\right)\nn\\
&&+\sum_{i=1}^N\left( \frac{1}{p}(\mathrm x_i-\wt{\mathrm{x}}_i)+\frac{\partial \mathrm{\wt{x}}_i}{\partial \tau}
\ln\left| p  \right|\right)\;.
\end{eqnarray}
The Euler Lagrange equations 
\begin{equation}\label{xds}
\frac{\partial{\mathscr{L}_{(\tau)}}}{\partial{\mathrm{x}_i}}-\frac{d}{d\tau}\left(\frac{\partial{\mathscr{L}_{(\tau)}}}{\partial{(d\mathrm{x}_i/d\tau})}\right)=0\;,
\end{equation}
yield \eqref{eqmotionskew2}.

\section{The full continuum limit}\label{fullLIMIT}
In the previous Section, we took the continuum limit on the discrete variable $m$, leading to a system of
differential-difference equations. The full continuum limit, performed on the remaining discrete variable $\mathsf N$ as well as $\tau $, will lead to a coupled system
of poles in the first instance, from which a hierarchy of ODEs can be retrieved, which is the RS hierarchy. How to perform this limit is inspired by the
structure of the solutions of \eqref{EXACTSKEW2}. Performing the following computation, 
\begin{eqnarray}\label{fulllimitsolution}
 \boldsymbol{\mathsf {Y}}(\mathsf {N},\tau)&=&
\left(\bI+\frac{\boldsymbol{\Lambda}}{p}\right)^{-\mathsf{ N}}e^{\frac{\tau}{p}\left(\bI+\frac{\boldsymbol{\Lambda}
}{p}\right)^{-1}}\boldsymbol{\mathsf Y}(0,0)
e^{-\frac{\tau}{p}\left(\bI+\frac{\boldsymbol{\Lambda}}{p}\right)^{-1}}
\left(\bI+\frac{\boldsymbol{\Lambda}}{p}\right)^{\mathsf N}\nn\\
&&-\mathsf {N}\lambda\left(1+\frac{\boldsymbol{\boldsymbol{\Lambda}}}{p}\right)^{-1}+
\frac{\tau\lambda\boldsymbol{\Lambda}}{p^2}\left(1+\frac{\boldsymbol{\Lambda}}{p}\right)^{-2}\;\nn\\
&=&e^{-\mathsf{ N}\ln\left(1+\frac{\boldsymbol{\Lambda}}{p}\right)+\frac{\tau}{p}
\left(1+\frac{\boldsymbol{\Lambda}}{p}\right)^{-1}}\boldsymbol{\mathsf Y}(0,0)
e^{\mathsf{ N}\ln\left(1+\frac{\boldsymbol{\Lambda}}{p}\right)-\frac{\tau}{p}\left(1+\frac{\boldsymbol{\Lambda}}{p}\right)^{-1}}\nn\\
&&-\mathsf {N}\lambda\left(1+\frac{\boldsymbol{\boldsymbol{\Lambda}}}{p}\right)^{-1}+
\frac{\tau\lambda\boldsymbol{\Lambda}}{p^2}\left(1+\frac{\boldsymbol{\Lambda}}{p}\right)^{-2}\;.
\end{eqnarray}
We now introduce
%
\begin{eqnarray}\label{t1t2t3}
t_1=\frac{\tau}{p^2}+\frac{\mathsf {N}}{p}\;,\;\;\;
 t_2=-\frac{2\tau}{p^3}-\frac{\mathsf {N}}{p^2}\;,\;\;\;
 t_3=\frac{3\tau}{p^4}+\frac{\mathsf {N}}{p^3}\;,\;\;\;
\mbox{~~}.....\;\;,
\end{eqnarray}
%
and expand \eqref{fulllimitsolution} with respect to variable $ p$. We have
\begin{eqnarray}\label{fullY}
 \boldsymbol{\mathsf {Y}}(t_1,t_2,t_3,...,\mathsf {N})  
&=&e^{-\boldsymbol{\Lambda}t_1+\boldsymbol{\Lambda}^2\frac{t_2}{2}-\boldsymbol{\Lambda}^3\frac{t_3}{3}+...}
\boldsymbol{\mathsf Y}(0,0,...)
e^{\boldsymbol{\Lambda}t_1-\boldsymbol{\Lambda}^2\frac{t_2}{2}+\boldsymbol{\Lambda}^3\frac{t_3}{3}+...}\nn\\
&&-\mathsf{N}\lambda+\boldsymbol{\Lambda}\lambda t_1+\boldsymbol{\Lambda}^2\lambda t_2+\boldsymbol{\Lambda}^3\lambda t_3+...\;,
\end{eqnarray}
which is a function of time variables $(t_1,t_2,t_3,...,\mathsf {N})$. The positions of the particles $X_i(t_1,t_2,t_3,...,\mathsf {N})$ can be computed by
looking for the eigenvalues of \eqref{fullY}. The explicit 
expression of the solution for the RS can be obtained from the secular problem for the matrix
\begin{equation}\label{eq:exact}
 \boldsymbol{ X}(0,0)-\xi+\boldsymbol{ L}(0,0)\lambda t_1+\boldsymbol{L}^2(0,0)\lambda t_2+\boldsymbol{L}^3(0,0)\lambda t_3\;,
\end{equation}
where $\xi=\mathsf {N}\lambda$ and $\boldsymbol{ X}(0,0)=\boldsymbol{ U}^{-1}(0,0)\boldsymbol{ Y}(0,0)\boldsymbol{ U}(0,0)$ and 
$\boldsymbol{ L}(0,0)=\boldsymbol{ U}^{-1}(0,0)\boldsymbol{\Lambda}\boldsymbol{ U}(0,0)$. The solution \eqref{fullY} involves $N$-time flows for 
the RS system. The next solutions in the hierarchy can be generated by pushing further on with the expansion.

\subsection{The full limit on the equations of motion}
We now would like to see what would result from 
taking the limit on the equations of motion \eqref{eqmotionskew2}. First, we introduce
\begin{eqnarray}\label{eq:5.4}
\dot {\mathsf x}_i&=&\frac{\partial {\mathsf x}_i}{\partial \tau}=\frac{\partial X_i}{\partial t_1}
\frac{\partial t_1}{\partial \tau}+\frac{\partial X_i}{\partial t_2}\frac{\partial t_2}{\partial \tau}+\frac{\partial X_i}{\partial t_3}
\frac{\partial t_3}{\partial \tau}+...\nonumber\\
&=&\frac{1}{p^2}\frac{\partial X_i}{\partial t_1}-\frac{2}{p^3}\frac{\partial X_i}{\partial t_2}+\frac{3}{p^4}\frac{\partial X_i}{\partial t_3}+...\;,
\end{eqnarray}
and
\begin{eqnarray}\label{eq:5.5}
 {\mathsf x}_i(\mathsf {N}\pm 1)&=&\mathsf x_i\mp  \lambda\pm \frac{1}{p}\frac{\partial X_i}{\partial t_1} 
+\frac{1}{p^2}\left(\frac{1}{2}\frac{\partial^2 X_i}{\partial t_1^2} \mp \frac{\partial X_i}{\partial t_2}\right)
+\frac{1}{p^3}\left(\pm \frac{1}{6}\frac{\partial X_i}{\partial t_3}-\frac{\partial^2 X_i}{\partial t_1\partial t_2} 
\pm \frac{\partial X_i}{\partial t_3}\right)\nn\\
&&+\frac{1}{p^4}\left(\frac{1}{24}\frac{\partial^4 X_i}{\partial t_4^4}\mp \frac{1}{2}\frac{\partial^3 X_i}{\partial t_1^2\partial t_2}+\frac{1}{2}\frac{\partial^2 X_i}{\partial t_2^2}+\frac{\partial^2 X_i}{\partial t_1\partial t_3}\right)
+\mathcal O(1/p^5)\;.
\end{eqnarray}
%
%
%
Then we expand \eqref{eqmotionskew2} with respect to the variable $p$ together with \eqref{t1t2t3}. We find that
\\
\\
$\blacklozenge{}$The leading term of order $\mathcal O(1/p^3)$ gives us
\begin{eqnarray}\label{EQMOTIONRS1}
\frac{\partial^2 X_i}{\partial t_1^2}\mbox{\huge{/}}\frac{\partial X_i}{\partial t_1}+\sum_{j=1}^N\frac{\partial X_j}{\partial t_1}\left( \frac{1}{X_i-X_j+\lambda}+\frac{1}{X_i-X_j-\lambda}-\frac{2}{X_i-X_j}\right)=0\;,
\end{eqnarray}
which is the equations of the motion for the usual continuous RS system \cite{Bru}.
\\
\\
$\blacklozenge{}$The term of order $\mathcal O(1/p^4)$ gives us
\begin{eqnarray}\label{EQMOTIONRS2}
&&2\frac{\partial^2 X_i}{\partial t_1\partial t_2}\mbox{\huge{/}}\frac{\partial X_i}{\partial t_1}
-\frac{\partial^2 X_i}{\partial t_1^2}\frac{\partial X_i}{\partial t_2}\mbox{\huge{/}}\left(\frac{\partial X_i}{\partial t_1}\right)^2
-\frac{1}{\lambda}\frac{\partial^2 X_i}{\partial t_1^2} \nn\\
&&+\sum_{j=1}^N\left[\frac{\partial X_j}{\partial t_2}\left( \frac{1}{X_i-X_j+\lambda}+\frac{1}{X_i-X_j-\lambda}-\frac{2}{X_i-X_j}\right)\right.\nn\\
&&+\frac{1}{2 }\frac{\partial^2 X_j}{\partial t_1^2}\left( \frac{1}{X_i-X_j-\lambda}-\frac{1}{X_i-X_j+\lambda}\right)\nn\\
&&+\frac{1}{2}\left(\frac{\partial X_j}{\partial t_2}\right)^2\left.\left(\frac{1}{(X_i-X_j-\lambda)^2}-\frac{1}{(X_i-X_j+\lambda)^2} \right)\right]=0\;.
\end{eqnarray}
This equation represents the next equation of motion beyond the usual continuous RS in the hierarchy.
We will stop at this equation, but we can actually get the higher terms of the equation in which the variable $t_3$ and higher 
order time-flows must be taken into account.
\\
\\
The full limit of the addition between \eqref{CONSKEW1} and \eqref{CONSKEW2} in the order $\mathcal O(1/p^2)$ gives
\begin{eqnarray}\label{CONcon}
\frac{2}{\lambda }\frac{\partial X_i}{\partial t_1}-2\frac{\partial X_i}{\partial t_2}\mbox{\huge{/}}\frac{\partial X_i}{\partial t_1} +\sum_{j=1}^N\frac{\partial X_j}{\partial t_1}\left(\frac{1}{X_i-X_j+\lambda}-\frac{1}{X_i-X_j-\lambda} \right)=0\;,
\end{eqnarray}
which is the constraint equations for the full limit which will play a crucial role in the next Subsection, see \eqref{eq:5.10s}.
\\
\\
Using \eqref{CONcon}, we can simplify \eqref{EQMOTIONRS2} into
\begin{eqnarray}\label{EQMOTIONRS3}
&&\frac{\partial^2 X_i}{\partial t_1\partial t_2}\mbox{\huge{/}}\frac{\partial X_i}{\partial t_1}
+\sum_{j=1}^N\left[\frac{\partial X_j}{\partial t_2}\left( \frac{1}{X_i-X_j+\lambda}+\frac{1}{X_i-X_j-\lambda}-\frac{2}{X_i-X_j}\right)\right.\nn\\
&&-\frac{1}{2}\frac{\partial X_i}{\partial t_1}\frac{\partial X_j}{\partial t_1}\left.\left(\frac{1}{(X_i-X_j-\lambda)^2}-\frac{1}{(X_i-X_j+\lambda)^2} \right)\right]=0\;.
\end{eqnarray}
Note that \eqref{EQMOTIONRS3} can be obtained directly from the full limit in order $\mathcal O(1/p^3)$ from the combination of the relations
\begin{eqnarray}
-\frac{1}{p} &=&\sum_{j=1}^N\frac{\partial {{\mathrm x}_j}}{\partial \tau}\left(\frac{1}{{\hypotilde 0 {\mathrm x}_i}-{\mathrm x}_j-\lambda}-\frac{1}{{\hypotilde 0 {\mathrm x}_i}-{\mathrm x}_j} \right)\;,\\
\frac{1}{p} &=&\sum_{j=1}^N\frac{\partial {\mathrm x}_j}{\partial \tau}\left(\frac{1}{\wt{\mathrm x}_i-{\mathrm x}_j-\lambda}-\frac{1}{\wt{\mathrm x}_i- {\mathrm x}_j}\right)\;,
\end{eqnarray}
which are the backward shift and forward shift of \eqref{CONSKEW1} and \eqref{CONSKEW2}, respectively.

\subsection{The full limit on the action}
We will follow the steps in \cite{Sikarin1} in order to obtain the full limit on the action.
We now take the action to be of the form
\begin{equation}
{S}[\boldsymbol{x}(\mathsf N,\tau);\Gamma]=\int_{\tau_1}^{\tau_2}d\tau\mathscr{L}_{(\tau)}(\boldsymbol{x}(\mathsf N,\tau),\dot{\boldsymbol{x}}(\mathsf N,\tau))
+\sum_{\mathsf N}\mathscr{L}_{(\mathsf N)}(\boldsymbol{x}(\mathsf N, \tau),\boldsymbol{x}(\mathsf N+1,\tau))\;,
\end{equation}
where the first term belongs to the vertical part and the second term belongs to the horizontal part of the curve $\Gamma$.

Using anti-Taylor expansion, the action now becomes
\begin{equation}
{S}[\boldsymbol{x}(\mathsf N,\tau);\Gamma]=\int_{\tau_1}^{\tau_2}d\tau\mathscr{L}_{(\tau)}(\boldsymbol{x}(\mathsf N,\tau),\dot{\boldsymbol{x}}(\mathsf N,\tau))
+\int_{\mathsf N_1}^{\mathsf N_2}d\mathsf N\mathscr{L}_{(\mathsf N)}(\boldsymbol{x}(\mathsf N, \tau),\boldsymbol{x}(\mathsf N+1,\tau))\;,
\end{equation}
where we do not need to take into account the boundary terms coming from the expansion, because they are constant at the end points and do not contribute 
to the variational process.
\\

We now perform a change of variables\footnote{We here restrict ourselves with the first two flows for simplicity.} $(\tau,\mathsf N)\mapsto (t_1,t_2)$ by using \eqref{t1t2t3},
\begin{subequations}
\begin{eqnarray}
d\tau&=&-p^3dt_2-p^2dt_1\;,\\
d\mathsf N&=&p^2dt_2+2pdt_1\;,
\end{eqnarray}
\end{subequations}
and also expand the Lagrangians with respect to variable $p$. We obtain 
\begin{eqnarray}\label{eq:5.19}
{S}[\boldsymbol{X}(t_1,t_2);\Gamma]&=&\int_{t_1(1)}^{t_1(2)}dt_1\mathscr{L}_{(t_1)}\left(\boldsymbol{X}(t_1,t_2),
\frac{\partial\boldsymbol{X}(t_1,t_2)}{\partial t_1}\right)\nn\\
&+&\int_{t_2(1)}^{t_2(2)}dt_3\mathscr{L}_{(t_2)}\left(\boldsymbol{X}(t_1,t_2),\frac{\partial\boldsymbol{X}(t_1,t_2)}{\partial t_1},
\frac{\partial\boldsymbol{X}(t_1,t_2)}{\partial t_2}\right)\;,
\end{eqnarray}
where $\mathscr{L}_{(t_1)}$ and $\mathscr{L}_{(t_2)}$ are given by
\begin{eqnarray}\label{eq:5.10}
\mathscr L_{(t_1)}=\sum\limits_{i=1}^N\frac{\partial X_i}{\partial t_1}\ln\left| \frac{\partial X_i}{\partial t_1}\right|-\sum\limits_{i \ne j}^N\frac{\partial X_j}{\partial t_1}\left( \ln|X_i-X_j-\lambda|-\ln|X_i-X_j|\right)\;, 
\end{eqnarray}
which first appeared in \cite{Braden} and the Euler-Lagrange equation 
\begin{eqnarray}\label{eq:5.10c}
\frac{\partial \mathscr L_{(t_1)}}{\partial X_i}-\frac{\partial}{\partial t_1}\left( \frac{\partial \mathscr
 L_{(t_1)}}{\partial(\frac{\partial X_i}{\partial t_1})}\right)=0\;, 
\end{eqnarray}
gives exactly eq. (\ref{EQMOTIONRS1}) and
\begin{eqnarray}\label{eq:5.10a}
\mathscr L_{(t_2)}&=&\sum\limits_{i=1}^N\left(\frac{\partial X_i}{\partial t_2}\ln\left| \frac{\partial X_i}{\partial t_1}\right|
-\frac{1}{2\lambda }\left( \frac{\partial X_i}{\partial t_1}\right)^2+3\frac{\partial X_i}{\partial t_2}\right)\nn\\
&&-\sum\limits_{i \ne j}^N\left[\frac{\partial X_j}{\partial t_2}\left( \ln|X_i-X_j-\lambda|-\ln|X_i-X_j|\right)+
\frac{1}{2}\frac{\partial X_i}{\partial t_1}\frac{\partial X_j}{\partial t_1}\frac{1}{X_i-X_j+\lambda}\right]\;.\nn\\
\end{eqnarray}
We see that the Lagrangian $\mathscr L_{(t_2)}$ contains derivatives with respect to two time flows $t_1$ and $t_2$. 
We observe that the equations of motion (\ref{EQMOTIONRS3}) require the Euler-Lagrange equation in the form
\begin{eqnarray}\label{eq:5.10d}
\frac{\partial \mathscr L_{(t_2)}}{\partial X_i}-\frac{\partial}{\partial t_2}\left( \frac{\partial \mathscr 
L_{(t_2)}}{\partial(\frac{\partial X_i}{\partial t_2})}\right)=0\;. 
\end{eqnarray}
Furthermore, we find that
\begin{equation}\label{eq:5.10s}
\frac{\partial \mathscr L_{(t_2)}}{\partial(\frac{\partial X_i}{\partial t_1})}
=\frac{2}{\lambda }\frac{\partial X_i}{\partial t_1}-2\frac{\partial X_i}{\partial t_2}\mbox{\huge{/}}\frac{\partial X_i}{\partial t_1} +
\sum_{j=1}^N\frac{\partial X_j}{\partial t_1}\left(\frac{1}{X_i-X_j+\lambda}-\frac{1}{X_i-X_j-\lambda} \right)=0\;,
\end{equation} 
which is continuum analogue of discrete constraints \eqref{CONcon} in order $\mathcal{O}(1/p^2)$. Actually there are more constraints from the expansion which will play a major role when we consider the higher Lagrangians in the hierarchy.

Here we obtained the hierarchy of continuous Lagrangians associated with the discrete-time RS-model through the full continuum limit. Obviously, higher Lagrangians in the family
can be generated by pushing further on with the expansion. Interestingly, these Lagrangians as well as the constraints also satisfy the variational principle for Lagrangian 1-form presenting in \cite{Sikarin1,Rinthesis}.

\subsection{The full limit on the closure relation}
We take the full limit on the discrete closure relation \eqref{closure} leading to
\begin{theorem}
We find that the continuous version of the closure relation between $t_1$ and $t_2$
\begin{equation}\label{conclosure}
\frac{\partial \mathscr L_{(t_2)}}{\partial t_1}=\frac{\partial \mathscr L_{(t_1)}}{\partial t_2}\;,
\end{equation}
which holds on the equations of motion and constraint.
\end{theorem}
\begin{proof}: We find that
\begin{subequations}
\begin{eqnarray}\label{p1}
\frac{\partial \mathscr L_{(t_1)}}{\partial t_2}&=&\sum_{i=1}^N\left( \frac{\partial^2 X_i}{\partial t_1\partial t_2}\ln\left| \frac{\partial X_i}{\partial t_1}\right|+\frac{\partial^2 X_i}{\partial t_1\partial t_2}\right)
-\sum\limits_{i \ne j}^N\left(\frac{\partial^2 X_j}{\partial t_2\partial t_1}\left[\ln|X_i-X_j-\lambda|-\ln|X_i-X_j| \right] \right.\nn\\
&&+\frac{\partial X_j}{\partial t_1}\frac{\partial X_i}{\partial t_2}\left[ \frac{1}{X_i-X_j-\lambda}
-\frac{1}{X_i-X_j}\right]
+\left.\frac{\partial X_j}{\partial t_1}\frac{\partial X_j}{\partial t_2}\left[ \frac{1}{X_i-X_j-\lambda}
-\frac{1}{X_i-X_j}\right]\right)\;,\nn\\
\end{eqnarray}
and
\begin{eqnarray}\label{p2}
\frac{\partial \mathscr L_{(t_2)}}{\partial t_1}&=&\sum_{i=1}^N\left( \frac{\partial^2 X_i}{\partial t_1\partial t_2}\ln\left| \frac{\partial X_i}{\partial t_1}\right|+\frac{\partial X_i}{\partial t_2}\frac{\partial^2 X_i}{\partial t_1^2}\mbox{\huge{/}}\frac{\partial X_i}{\partial t_1}
-\frac{1}{\lambda }\frac{\partial X_i}{\partial t_1}\frac{\partial^2 X_i}{\partial t_1^2}
+3\frac{\partial^2 X_i}{\partial t_1\partial t_2}\right)\nn\\
&&-\sum\limits_{i \ne j}^N\left(\frac{\partial^2 X_j}{\partial t_2\partial t_1}\left[\ln|X_i-X_j-\lambda|-\ln|X_i-X_j| \right] \right.\nn\\
&&+\frac{\partial X_i}{\partial t_1}\frac{\partial X_j}{\partial t_2}\left[ \frac{1}{X_i-X_j-\lambda}
-\frac{1}{X_i-X_j}\right]
+\frac{\partial X_j}{\partial t_1}\frac{\partial X_j}{\partial t_2}\left[ \frac{1}{X_i-X_j-\lambda}
-\frac{1}{X_i-X_j}\right]\nn\\
&&+\frac{1}{2}\frac{\partial^2 X_j}{\partial t_1^2}\frac{\partial X_i}{\partial t_1}\left[ \frac{1}{X_i-X_j+\lambda}
-\frac{1}{X_i-X_j-\lambda }\right]\nn\\
&&-\left.\frac{1}{2}\left(\frac{\partial X_j}{\partial t_1}\right)^2\frac{\partial X_i}{\partial t_1}\left[ \frac{1}{(X_i-X_j-\lambda)^2}
-\frac{1}{(X_i-X_j+\lambda)^2 }\right]\right)\;.
\end{eqnarray}
We find that $\frac{\partial \mathscr L_{(t_1)}}{\partial t_2}=\frac{\partial \mathscr L_{(t_2)}}{\partial t_1}$ gives
\begin{eqnarray}\label{p3}
&&-\sum_{i=1}^N\frac{\partial^2 X_i}{\partial t_1\partial t_2}+\frac{\partial X_j}{\partial t_1}\frac{\partial X_i}{\partial t_2}\left[ \frac{1}{X_i-X_j-\lambda}
-\frac{1}{X_i-X_j}\right]=\sum_{i=1}^N\left(\frac{\partial X_i}{\partial t_2}\frac{\partial^2 X_i}{\partial t_1^2}\mbox{\huge{/}}\frac{\partial X_i}{\partial t_1}\right.\nn\\
&&-\left.\frac{1}{\lambda }\frac{\partial X_i}{\partial t_1}\frac{\partial^2 X_i}{\partial t_1^2}
+3\frac{\partial^2 X_i}{\partial t_1\partial t_2}\right)+\sum_{i\neq j}^N\left(-\frac{\partial X_i}{\partial t_1}\frac{\partial X_j}{\partial t_2}\left[ \frac{1}{X_i-X_j-\lambda}
-\frac{1}{X_i-X_j}\right]\right.\nn\\
&&+\frac{1}{2}\frac{\partial^2 X_j}{\partial t_1^2}\frac{\partial X_i}{\partial t_1}\left[ \frac{1}{X_i-X_j+\lambda}
-\frac{1}{X_i-X_j-\lambda }\right]\nn\\
&&-\left.\frac{1}{2}\left(\frac{\partial X_j}{\partial t_1}\right)^2\frac{\partial X_i}{\partial t_1}\left[ \frac{1}{(X_i-X_j-\lambda)^2}
-\frac{1}{(X_i-X_j+\lambda)^2 }\right]\right)\;.
\end{eqnarray}
Dividing \eqref{p3} by $\frac{\partial X_i}{\partial t_1}$ we find that
\begin{eqnarray}\label{p4}
&&\frac{\partial \mathscr L_{(t_1)}}{\partial t_2}-\frac{\partial \mathscr L_{(t_2)}}{\partial t_1}
=\sum_{i=1}^N\frac{\partial X_i}{\partial t_1}\left(2\frac{\partial^2 X_i}{\partial t_1t_2}\mbox{\huge{/}}\frac{\partial X_i}{\partial t_1}
+\frac{\partial^2 X_i}{\partial t_1^2}\frac{\partial X_i}{\partial t_2}\mbox{\huge{/}}\left(\frac{\partial X_i}{\partial t_2}\right)^2
-\frac{1}{\lambda }\frac{\partial^2 X_i}{\partial t_1^2}\right.\nn\\
&&-\sum_{j=1}^N\left[\frac{\partial X_j}{\partial t_2}\left( \frac{1}{X_i-X_j+\lambda}+\frac{1}{X_i-X_j-\lambda}-\frac{2}{X_i-X_j}\right)\right.\nn\\
&&+\frac{1}{2}\frac{\partial^2 X_j}{\partial t_1^2}\left( \frac{1}{X_i-X_j+\lambda}+\frac{1}{X_i-X_j-\lambda}\right)\nn\\
&&-\left.\frac{1}{2}\left(\frac{\partial X_i}{\partial t_1}\right)^2\left.\left(\frac{1}{(X_i-X_j-\lambda)^2}-\frac{1}{(X_i-X_j+\lambda)^2} \right)\right]\right)\;.
\end{eqnarray}
Using \eqref{EQMOTIONRS2}, \eqref{p4} becomes
\begin{eqnarray}\label{p4a}
&&\frac{\partial \mathscr L_{(t_1)}}{\partial t_2}-\frac{\partial \mathscr L_{(t_2)}}{\partial t_1}
=\sum_{i=1}^N\frac{\partial X_i}{\partial t_1}\left(
2\frac{\partial^2 X_i}{\partial t_1^2}\frac{\partial X_i}{\partial t_2}\mbox{\huge{/}}\left(\frac{\partial X_i}{\partial t_1}\right)^2\right.\nn\\
&&\;\;\;\;\;\;\;\;\;\;\;-2\sum_{j=1}^N\left.\frac{\partial X_j}{\partial t_2}\left( \frac{1}{X_i-X_j+\lambda}+\frac{1}{X_i-X_j-\lambda}-\frac{2}{X_i-X_j}\right)\right)\;.
\end{eqnarray}
Inserting \eqref{EQMOTIONRS1}, we have now
\begin{eqnarray}\label{p4b}
&&\frac{\partial \mathscr L_{(t_1)}}{\partial t_2}-\frac{\partial \mathscr L_{(t_2)}}{\partial t_1}\nn\\
&&=-2\sum_{i,j=1}^N\left(\frac{\partial X_j}{\partial t_1}\frac{\partial X_i}{\partial t_2}+\frac{\partial X_i}{\partial t_1}\frac{\partial X_j}{\partial t_2}\right)\left( \frac{1}{X_i-X_j+\lambda}+\frac{1}{X_i-X_j-\lambda}-\frac{2}{X_i-X_j}\right)\nn\\
\end{eqnarray}
The second term of \eqref{p4b} is the antisymmetric function, hence vanishes. We now have
\begin{eqnarray}\label{p4bs}
\frac{\partial \mathscr L_{(t_1)}}{\partial t_2}-\frac{\partial \mathscr L_{(t_2)}}{\partial t_1}=0\;.
\end{eqnarray}
\end{subequations}
\end{proof}
\subsection{The variational principle for continuous Lagrangian 1-forms}
In \cite{Sikarin1}, we set out the key principles of the new 
variational calculus associated with the multi-time Lagrangian 1-form structure.
These principles were discovered on the basis of the careful analysis of the rational CM model, 
which formed the ``laboratory'' for studying how this new least-action principle should work. For simplicity, we focused on the case of Lagrangian 1-forms in the 
2-time case, but the general principles apply to the case of the general multi-time case in an obvious 
manner\footnote{This was done in \cite{Rinthesis} as well as in a recent preprint, \cite{Suris}, 
adopting a somewhat different point of view.}. Let us summarize here our findings.

In the 2-time case the action is defined by 
\begin{eqnarray}\label{2-time-action} 
S[\boldsymbol{x}(t_1,t_2);\Gamma]&=&\int_{\Gamma}\left(\mathscr L_{(t_1)}dt_1+\mathscr L_{(t_2)}dt_2  
\right)\nn\\
&=&\int_{s_0}^{s_1}\left( \mathscr L_{(t_1)}\frac{dt_1}{ds}+\mathscr L_{(t_2)}\frac{dt_2}{ds}\right)ds\;.
\end{eqnarray}
where $\Gamma$ is an arbitrary curve in the space of the two time-variables $t_1$ and $t_2$, which is 
parametrised by $\Gamma:\,(t_1,t_2)=(t_1(s),t_2(s))$ with the parameter $s\in [s_0,s_1]$, see 
Fig. \ref{x-curve}(a). $\mathscr{L}_{(t_1)}$ and $\mathscr{L}_{(t_2)}$ are the components of the Lagrangian 
1-form:   
\begin{eqnarray}
\mathscr L_{(t_1)}&=&\mathscr L_{(t_1)}(\boldsymbol{x}(t_1,t_2),\boldsymbol{x}_{t_1}(t_1,t_2),
\boldsymbol{x}_{t_2}(t_1,t_2))\;,\\
\mathscr L_{(t_2)}&=&\mathscr L_{(t_2)}(\boldsymbol{x}(t_1,t_2),\boldsymbol{x}_{t_1}(t_1,t_2),
\boldsymbol{x}_{t_2}(t_1,t_2))\;,
\end{eqnarray}
in which $\boldsymbol{x}_{t_1}=\partial \boldsymbol{x}/\partial t_1$ $\boldsymbol{x}_{t_2}=\partial 
\boldsymbol{x}/\partial t_2$. The dependent variable $\boldsymbol{x}=(x_1,x_2,...,x_N)$ is the 
vector of the position coordinates of the particles. The action should, in our point of 
view, be considered as a functional of both the dependent variables $\boldsymbol{x}(t_1,t_2)$ 
as well as of the curve $\Gamma$, i.e., of the functions $\boldsymbol{t}(s)=(t_1(s),t_2(s))$. 
Thus, the least-action principle for 1-forms is the implementation of the demand for criticality 
of the action under (infinitesimal) variations of the curve $\Gamma$ in the space of independent 
variables, as well as of the \textit{evaluation curve} $\mathcal{E}_\Gamma$ in the space 
of dependent variables $\boldsymbol{x}(\boldsymbol{t})$, as indicated in Fig. \ref{x-curve}(a).  
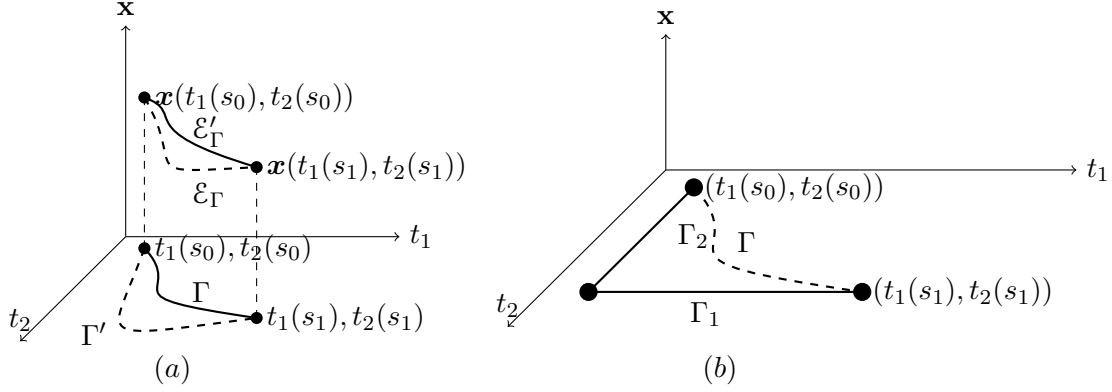
\begin{figure}[h]
\begin{center}
\begin{tikzpicture}[scale=0.4]
 \draw[->] (0,0,0) -- (9,0,0) node[anchor=west] {$t_{1}$};
 \draw[->] (0,0,0) -- (0,7,0) node[anchor=south] {$\textbf x$};
 \draw[->] (0,0,0) -- (0,0,9) node[anchor=south] {$t_{2}$};
 \fill (1,0,1) circle (0.2) node[anchor=west]{$t_1(s_0),t_2(s_0)$};
 \draw (3,0,3) node[anchor=north west] {$\Gamma$};
\draw (.7,0,6.5) node[anchor=north west] {$\Gamma^\prime$};
 \fill (7,0,7) circle (0.2) node[anchor=west]{$t_1(s_1),t_2(s_1)$} ;
\draw[dashed] (1,0,1)--(1,5,1);
\draw[dashed] (7,0,7)--(7,5,7);
\fill (1,5,1) circle (0.2) node[anchor=west]{$\boldsymbol x(t_1(s_0),t_2(s_0))$};
\fill (7,5,7) circle (0.2) node[anchor=west]{$\boldsymbol x(t_1(s_1),t_2(s_1))$};
 \draw (3,5.5,3) node[anchor=north west] {$\mathcal{E}_{\Gamma}^\prime$};
 \draw (3,3.3,3) node[anchor=north west] {$\mathcal{E}_{\Gamma}$};
 \draw (0,-5,-1.5) node[anchor=west] {$(a)$};
\draw [thick] (1,0,1) .. controls (4,0,5) and (1,0,5) .. (7,0,7);
\draw [thick,dashed] (1,0,1) .. controls (3,0,9) and (2,0,9) .. (7,0,7);
\draw [thick,dashed] (1,5,1) .. controls (4,4,5) and (2,4,5) .. (7,5,7);
\draw [thick] (1,5,1) .. controls (4,6,5) and (2,5.5,5) .. (7,5,7);
\end{tikzpicture}
\begin{tikzpicture}[scale=0.6]
 \draw[->] (0,0,0) -- (9,0,0) node[anchor=west] {$t_{1}$};
 \draw[->] (0,0,0) -- (0,3,0) node[anchor=south] {$\textbf x$};
 \draw[->] (0,0,0) -- (0,0,9) node[anchor=south] {$t_{2}$};
 \fill (1,0,1) circle (0.2) node[anchor=west]{$(t_1(s_0),t_2(s_0))$};
 \draw (1,0,2.5) node[anchor=north west] {$\Gamma_2$};
 \draw (3,0,7) node[anchor=north west] {$\Gamma_1$};
 \draw[thick](1,0,1)--(1,0,7)--(7,0,7);
 \draw (2.5,0,3) node[anchor=north west] {$\Gamma$};
 \fill (1,0,7) circle(0.2);
 \fill (7,0,7) circle (0.2) node[anchor=west]{$(t_1(s_1),t_2(s_1))$} ;
\draw [thick,dashed] (1,0,1) .. controls (4,0,5) and (1,0,5) .. (7,0,7);
\draw (0,-5,-1.5) node[anchor=west] {$(b)$};
\end{tikzpicture}
\end{center}
\caption{(a) The deformation of the curves $\Gamma \rightarrow \Gamma^\prime$ and $\mathcal{E}_{\Gamma} 
\rightarrow \mathcal{E}_{\Gamma}^\prime$ in the $\textbf x-\textbf t$ configuration. 
(b) The deformation of the curve $\Gamma$ on the space of independent variables.}\label{x-curve}
\end{figure}
Thus, we have to apply the principle of criticality of the action 
\eqref{2-time-action} under: 
\textit{i)} variations $\boldsymbol{t}(s)\,\to\,\boldsymbol{t}
(s)+\delta\boldsymbol{t}(s)$ of the 
independent variables parametrising the curve $\Gamma$, as well as \textit{ii)}
variations $\boldsymbol{x}(t_1,t_2)\,\to\,
\boldsymbol{x}(t_1,t_2)+\delta\boldsymbol{x}(t_1,t_2)$ of the dependent 
variables on an arbitrary evaluation curve $\mathcal{E}_\Gamma$. 
\\
\\
\noindent 
\textit{i)} Requiring criticality of the action w.r.t. variations of the curve   
$\delta S=S[\boldsymbol{x}(t_1+\delta t_1,t_2+\delta t_2);\Gamma^\prime]-
S[\boldsymbol{x}(t_1,t_2);\Gamma]=0$ with the condition 
$\delta\boldsymbol{t}(s_0)=\delta\boldsymbol{t}(s_1)=0$, leads to the 
continuous \textit{closure relation}: 
\begin{subequations}\label{ELsystem} 
\begin{equation}\label{eq:cont-closure}
\frac{\partial \mathscr L_{(t_1)} }{\partial t_2}=\frac{\partial 
\mathscr L_{(t_2)} }{\partial t_1}\;,
\end{equation}
\textit{ii)} Criticality of the action under variations of the 
dependent variables requires to consider two types of variations: 
variations of the variables $\boldsymbol{x}(\boldsymbol{t})$ and 
its derivatives tangential 
to the curve, and variations w.r.t. derivatives of the variables 
$\boldsymbol{x}(\boldsymbol{t})$ transversal to the curve. The latter 
should be treated as independent variations, whilst the former give rise 
to integration by parts in the integral over the variable ``\;$s$\;". 
%
%
%
This leads to the system of Euler-Lagrange (EL) equations: 
\begin{eqnarray} 
&& \frac{\partial \mathscr{L}_{(t_1)}}{\partial\boldsymbol x}\frac{dt_1}{ds}
+\frac{\partial \mathscr{L}_{(t_2)}}{\partial\boldsymbol x}\frac{dt_2}{ds} 
-\frac{d}{ds}\left\{ \frac{1}{\|d\boldsymbol{t}/ds\|^2} \times \right. \nn \\ 
&& \left. \left[ 
\left(\frac{dt_1}{ds}\right)^2\frac{\partial \mathscr{L}_{(t_1)}}{\partial\boldsymbol  x_{t_1}}
+\left(\frac{dt_1}{ds}\right)\left(\frac{dt_2}{ds}\right)\left(\frac{\partial \mathscr{L}_{(t_1)}}{\partial\boldsymbol x_{t_2}}+\frac{\partial \mathscr{L}_{(t_2)}}{\partial\boldsymbol  x_{t_1}}\right)
+ \left(\frac{dt_2}{ds}\right)^2\frac{\partial \mathscr{L}_{(t_2)}}{\partial \boldsymbol  x_{t_2}}
\right]\right\}=0\;,\nn\\\label{eq:5.32as}
\end{eqnarray}
together with 
\begin{eqnarray}
&&\frac{\partial \mathscr{L}_{(t_2)}}{\partial \boldsymbol{x}_{t_1}}\left(\frac{dt_1}{ds}\right)^2
+\left(\frac{\partial \mathscr{L}_{(t_1)}}{\partial\boldsymbol{x}_{t_2}}
-\frac{\partial \mathscr{L}_{(t_2)}}{\partial\boldsymbol{x}_{t_2}}\right)\frac{dt_1}{ds}\frac{dt_2}{ds}
-\frac{\partial \mathscr{L}_{(t_1)}}{\partial \boldsymbol{x}_{t_2}}\left(\frac{dt_1}{ds}\right)^2=0.\label{eq:5.32bs}
\end{eqnarray}\end{subequations} 
Here \eqref{eq:5.32bs} could be considered to be a system of constraints 
whilst \eqref{eq:5.32as} are EL equations along the curve $\Gamma$. 
\footnote{In \cite{Sikarin1} eq. \eqref{eq:5.32as} was given in a slightly  
different form, using a different basis of decomposition of the derivatives 
of $\boldsymbol{x}$ along and transversal to the curve, whereas this particular 
form uses an orthogonal basis as suggested by \cite{Suris}. Although the form 
of \cite{Sikarin1}, which was obtained using a non-orthogonal basis for the 
decomposition, has the slight disadvantage that it is not well-defined 
for points on the curves where $dt_1/dt_2$ becomes singular, both forms 
are equivalent when viewed as part of the system containing both 
\eqref{eq:5.32as} and \eqref{eq:5.32bs}, the latter being invariant under the 
choice of basis.}     

It is a conceptually novel point, put forward in \cite{Sikarin1} and \cite{SF1},  
that the entire set of generalized EL equations \eqref{ELsystem} should be 
considered not only as a system producing equations of the motion for a given 
Lagrange function, but should actually be considered as a system of equations 
for the Lagrangians themselves. The solutions of the system, which are the 
\textit{admissable} Lagrangians, are necessarily the ones associated with 
integrable systems. 

Since the equations \eqref{ELsystem} must hold on an arbitrary 
curve, the system must hold in particular on curves made out of straight 
segments along the $t_1$- and $t_2$-axes. Thus, invoking the closure relation 
\eqref{eq:cont-closure}, we can deform an arbitrary curve $\Gamma$ to a simpler  
curve : $\Gamma \rightarrow \Gamma_1+\Gamma_2$ as shown in Fig. \ref{x-curve}(b). 
On the curve $\Gamma_2$, where the time variable $t_1$ is ``frozen'', the 
constraint equations arises from variations of the derivative $\boldsymbol{x}_{t_1}$, whereas on $\Gamma_1$ the constraint 
arises from variations of the derivative $\boldsymbol{x}_{t_2}$. 
Thus, in the case that $\mathscr{L}_{(t_1)}$ is independent of the latter 
derivatives (as in the example of the RS system) we get the system of equations: 
\begin{subequations}\label{EL23}
\begin{eqnarray}
&&\frac{\partial{\mathscr{L}_{(t_1)}}}{\partial{\boldsymbol{x}}}
-\frac{\partial}{\partial t_1}\left(\frac{\partial{\mathscr{L}_{(t_1)}}}{\partial 
\boldsymbol{x}_{t_1}}\right)=0\;,\;\;\;\;\;\;\;\;\;\;\frac{\partial{\mathscr{L}_{(t_1)}}}{\partial {\boldsymbol{x}_{t_2}}}=0\;, \\ \label{eq:5.32c}
&&\frac{\partial{\mathscr{L}_{(t_2)}}}{\partial{\boldsymbol{x}}}
-\frac{\partial}{\partial t_2}\left(\frac{\partial{\mathscr{L}_{(t_2)}}}{\partial 
\boldsymbol{x}_{t_2}}\right)=0\;,\;\;\;\;\;\;\;\;\;\;
\frac{\partial{\mathscr{L}_{(t_2)}}}{\partial {\boldsymbol{x}_{t_1}}}=0.\label{eq:5.32b}
\end{eqnarray}
\end{subequations}
Thus, we recover the system of EL equations and closure as given in subsection 5.3. 


\section{The connection to the lattice KP systems}\label{conKP}
In contrast to the CM case \cite{Sikarin1}, where we started with a semi-discrete KP equation, and applied a pole-reduction to it to a yield a compatible CM system, 
here we start from RS system and reconnect it to the fully discrete lattice KP systems. In \cite{FrankN8}, the connection between the RS system and the KP system 
was established for the trigonometric case, but here we will focus on the (simpler) rational case as it clarifies the situation more clearly, 

We will develop now a scheme along the lines of the papers \cite{Rinthesis,NAH,FrankJames,Sikarin2}. Starting from the ``solution matrix'' $\boldsymbol{Y}(n,m)$ given in \eqref{Ynm1}, we will introduce the relevant $\tau$-function as its characteristic polynomial: 
\begin{equation}\label{Sec}
\tau (\xi )=\det(\xi \boldsymbol I-\boldsymbol Y) \;,
\end{equation}
where $\bsY=\bsY(n,m,h)$ is now a function of three discrete variables obeying the 
shift relations: 
\begin{subequations}\label{KPYrels} 
\begin{eqnarray}\label{KPYn1}
\wt{\boldsymbol{Y}}-\boldsymbol{Y}+\mu \boldsymbol{I}&=&\wt{\boldsymbol{r}}\boldsymbol{s}^T\;,\\
\wh{\boldsymbol{Y}}-\boldsymbol{Y}+\eta \boldsymbol{I}&=&\wh{\boldsymbol{r}}\boldsymbol{s}^T\;,\\
\ol{\boldsymbol{Y}}-\boldsymbol{Y}+\nu \boldsymbol{I}&=&\ol{\boldsymbol{r}}\boldsymbol{s}^T\;,
\end{eqnarray}
\end{subequations}
where $\boldsymbol{r}$ and $\boldsymbol{s}$ depend on the discrete variables via the following 
shift relations (see \eqref{rs}): 
\begin{subequations}\label{KPrs}\begin{eqnarray}
&& (p\boldsymbol{I}+\bLam)\cdot \wt{\boldsymbol r}=\boldsymbol r\;,\;\;\;\;\boldsymbol{s}^T\cdot (p\boldsymbol{I}+\bLam)=\wt{\boldsymbol s}^T\;, \\ 
&& (q\boldsymbol{I}+\bLam)\cdot \wh{\boldsymbol r}=\boldsymbol r\;,\;\;\;\;\boldsymbol{s}^T\cdot (q\boldsymbol{I}+\bLam)=\wh{\boldsymbol s}^T\;, \\ 
&& (r\boldsymbol{I}+\bLam)\cdot \ol{\boldsymbol r}=\boldsymbol r\;,\;\;\;\;\boldsymbol{s}^T\cdot (r\boldsymbol{I}+\bLam)=\ol{\boldsymbol s}^T\;.
\end{eqnarray}\end{subequations} 
As explained in Appendix A, we have introduced in \eqref{KPYrels} a slight generalization by 
introducing the parameters $\mu$, $\eta$ and $\nu$ instead of all three being equal to 
$\lambda$. The shifts $\wt{\phantom{a}}$ and $\wh{\phantom{a}}$ are, as before, lattice 
shifts associated with the lattice parameters $p$ and $q$ respectively, whereas the shift in 
the third variable is indicated by $\ol{\phantom{a}}$ and is associated with a lattice 
parameter $r$. 

To derive the equations directly from the resolvent of the matrix $\boldsymbol Y$, we proceed as 
follows. First, we perform the simple computation
\begin{eqnarray}
\wt{\tau}(\xi)&=&\det(\xi+\mu-\boldsymbol Y-\wt{\boldsymbol r}\boldsymbol s^T)\;,\nn\\
&=&\det((\xi+\mu-\boldsymbol Y)(1-\wt{\boldsymbol r}\boldsymbol s^T(\xi+\mu-\boldsymbol Y)^{-1}))\;,\nn\\
&=&\tau(\xi+\mu)(1-\boldsymbol s^T(\xi+\mu-\boldsymbol Y)^{-1}\wt{\boldsymbol r})\;,\nn
\end{eqnarray}
then we have
\begin{equation}\label{tau1}
\frac{\wt{\tau}(\xi)}{\tau(\xi+\mu)}=1-\boldsymbol s^T(\xi+\mu-\boldsymbol Y)^{-1}(p+\bLam)^{-1}\boldsymbol r=\boldsymbol {\mathrm v}_{p}(\xi+\mu)\;,
\end{equation}
in which the function $\mathrm{v}_p$ is given by 
\begin{equation}\label{va}
\mathrm{v}_a(\xi):= 1-\boldsymbol s^T(\xi-\boldsymbol Y)^{-1}(a+\bLam)^{-1}{\boldsymbol r}
\end{equation}
for a general parameter $a$ setting $a=p$. 
\\
\\
The reverse formula to Eq. \eqref{tau1} can be obtained by a similar computation: 
\begin{eqnarray}
\tau(\xi)&=&\det(\xi-\mu-\wt{\boldsymbol Y}+\wt{\boldsymbol r}\boldsymbol s^T)\;,\nn\\
&=&\det((\xi-\mu-\wt{\boldsymbol Y})(1+\wt{\boldsymbol r}\boldsymbol s^T(\xi-\mu-\wt{\boldsymbol Y})^{-1}))\;,\nn\\
&=&\wt{\tau}(\xi-\mu)(1+ {\boldsymbol s}^T(\xi-\mu-\wt{\boldsymbol Y})^{-1}\wt{\boldsymbol r})\;,\nn
\end{eqnarray}
then we have
\begin{equation}\label{tau2}
\frac{\tau(\xi)}{\wt{\tau}(\xi-\mu)}=1+\wt{\boldsymbol s}^T(p+\bLam)^{-1}(\xi-\mu-\wt{\boldsymbol Y})^{-1}\wt{\boldsymbol r}=\wt{\boldsymbol {\mathrm w}}_{p}(\xi-\mu)\;, 
\end{equation}
in which the function $\mathrm{w}_p$ is given by 
\begin{equation}\label{wa}
\mathrm{w}_a(\xi):= 1+\boldsymbol s^T(a+\bLam)^{-1}(\xi-\boldsymbol Y)^{-1}{\boldsymbol r}
\end{equation}
for again a general parameter $a$ setting $a=p$. 
\\
\\
From \eqref{tau1} and \eqref{tau2}, we have the relation
\begin{equation}\label{tauvw}
\frac{\tau(\xi)}{\wt{\tau}(\xi-\mu)}=\wt{\boldsymbol {\mathrm w}}_{p}(\xi-\mu)=\frac{1}{\boldsymbol {\mathrm v}_{p}(\xi)}\;.
\end{equation}
The same type of the relation for the other discrete directions can be obtained in the forms
\begin{subequations}
\begin{eqnarray}\label{tauvw12}
&&\frac{\tau(\xi)}{\wh{\tau}(\xi-\eta  )}=\wh{\boldsymbol {\mathrm w}}_{q}(\xi-\eta )=\frac{1}{\boldsymbol {\mathrm v}_{q}(\xi)}\;,\\
&&\frac{\tau(\xi)}{\ol{\tau}(\xi-\nu   )}=\ol{\boldsymbol {\mathrm w}}_{r}(\xi-\nu   )=\frac{1}{\boldsymbol {\mathrm v}_{r}(\xi)}\;.
\end{eqnarray}
\end{subequations}
In order to derive discrete KP equations for $\tau(\xi)$, $\boldsymbol {\mathrm w}$ and $\boldsymbol {\mathrm v}$, we introduce the N-component vectors
 \begin{subequations}
\begin{eqnarray}
\boldsymbol{ u}_a(\xi)&=&(\xi-\boldsymbol Y)^{-1}(a+\bLam)^{-1}\boldsymbol{r}\;,\label{uu1}\\
\boldsymbol{\tbu}_b(\xi)&=&\boldsymbol{s}^T(b+\bLam)^{-1}(\xi-\boldsymbol Y)^{-1}\;,\label{uu2}
\end{eqnarray}
\end{subequations}
as well as the scalar variables
\begin{equation}\label{S}
S_{ab}(\xi)=\boldsymbol{s}^T(b+\bLam)^{-1}(\xi-\boldsymbol Y)^{-1}(a+\bLam)^{-1}\boldsymbol{r}\;.
\end{equation}
We now consider \eqref{uu1} which can be written in the form
\begin{eqnarray}
\boldsymbol{ u}_a(\xi)&=&(p-a)\wt{\boldsymbol{ u}}_a(\xi-\mu )+\boldsymbol{\mathrm v}_a(\xi)\wt{\boldsymbol{ u}}_0(\xi-\mu )\;,\label{ua}
\end{eqnarray}
with $\boldsymbol{ u}_0(\xi)=(\xi-\boldsymbol Y)^{-1}\boldsymbol{r}$.
\\
\\
The same process can be applied to \eqref{uu2} and we obtain
\begin{equation}\label{ub}
\wt{\boldsymbol{ \tbu}}_b(\xi)=(p-b)\wt{\boldsymbol{ \tbu}}_b(\xi+\mu )+\wt{\boldsymbol{\mathrm w}}_b(\xi)\wt{\boldsymbol{ \tbu}}_0(\xi+\mu )\;,
\end{equation}
with $\boldsymbol{ \tbu}_0(\xi)=\boldsymbol{s}^T(\xi-\boldsymbol Y)^{-1}$.
\\
\\
Another type of relation can be obtained by multiply $\wt{\boldsymbol{s}}^T(b+\bLam)^{-1}$ on the left hand side of \eqref{ua}. We have
\begin{eqnarray}
\wt{\boldsymbol{s}}^T(b+\bLam)^{-1}\boldsymbol{ u}_a(\xi)&=&(p-a)\wt{\boldsymbol{s}}^T(b+\bLam)^{-1}\wt{\boldsymbol{ u}}_a(\xi-\mu )\nn\\
&&+\boldsymbol{\mathrm v}_a(\xi)\wt{\boldsymbol{s}}^T(b+\bLam)^{-1}\wt{\boldsymbol{ u}}_0(\xi-\mu )\;,\nn\\
\boldsymbol{s}^T(p+\bLam)(b+\bLam)^{-1}\boldsymbol{ u}_a(\xi)&=&(p-a)\wt{S}_{ab}(\xi-\mu )+\boldsymbol{\mathrm v}_a(\xi)\wt{\boldsymbol{ \mathrm{w}}}_b(\xi-\mu )\;,\nn\\
\boldsymbol{\mathrm v}_a(\xi)\wt{\boldsymbol{ \mathrm{w}}}_b(\xi-\mu )&=&1+(p-b)S_{ab}(\xi)-(p-a)\wt{S}_{ab}(\xi-\mu )\;.\label{vw1}
\end{eqnarray}
Similarly, multiplying the right hand side of \eqref{ub}, we obtain
\begin{equation}
\wt{\boldsymbol{\mathrm w}}_b(\xi)\boldsymbol{ \mathrm{v}}_a(\xi+\mu )=1+(p-b)S_{ab}(\xi+\mu )-(p-a)\wt{S}_{ab}(\xi)\;.
\label{vw2}
\end{equation}
By proceeding the similar steps, we can derive the relations in other discrete-time directions, namely
 \begin{subequations}
\begin{eqnarray}
\boldsymbol{\mathrm v}_a(\xi)\wh{\boldsymbol{ \mathrm{w}}}_b(\xi-\eta  )&=&1+(q-b)S_{ab}(\xi)-(q-a)\wh{S}_{ab}(\xi-\eta  )\;,
\label{vw3}\\
\boldsymbol{\mathrm v}_a(\xi)\ol{\boldsymbol{ \mathrm{w}}}_b(\xi-\nu  )&=&1+(r-b)S_{ab}(\xi)-(r-a)\ol{S}_{ab}(\xi-\nu  )\;,
\label{vw3}
\end{eqnarray}
\end{subequations}
%
Using the identity
\begin{equation}
\frac{\wt{\overline{\boldsymbol{\mathrm w}}}_b(\xi-\mu  -\eta )\overline{\boldsymbol{ \mathrm{v}}}_a(\xi-\mu  )}{\wh{\overline{\boldsymbol{\mathrm w}}}_b(\xi-\mu -\nu  )\overline{\boldsymbol{ \mathrm{v}}}_a(\xi-\mu  )}
=\frac{\wt{\overline{\boldsymbol{\mathrm w}}}_b(\xi-\mu  -\nu  )\wt{\boldsymbol{ \mathrm{v}}}_a(\xi-\mu   )}{\wh{\overline{\boldsymbol{\mathrm w}}}_b(\xi-\mu  -\eta )\wh{\boldsymbol{ \mathrm{v}}}_a(\xi-\mu   )}\;
\frac{\wh{\wt{\boldsymbol{\mathrm w}}}_b(\xi-\mu  -\eta  )\wh{\boldsymbol{ \mathrm{v}}}_a(\xi-\eta )}{\wh{\wt{\boldsymbol{\mathrm w}}}_b(\xi-\nu  -\eta )\wt{\boldsymbol{ \mathrm{v}}}_a(\xi-\nu  )}\;,
\label{idenvw}
\end{equation}
we can derive
\begin{eqnarray}\label{SKP}
&&\frac{1+(p-b)\overline{S}_{ab}(\xi-\nu  )-(p-a)\wt{\overline S}_{ab}(\xi-\mu  -\nu  )}{1+(q-b)\overline{S}_{ab}(\xi-\nu   )-(q-a)\wh{\overline S}_{ab}(\xi-\nu  -\eta )}\nn\\
&&\;\;\;\;\;\;\;\;\;\;\;=\frac{1+(r-b)\wt S_{ab}(\xi-\mu   )-(r-a)\wt{\overline{S}}_{ab}(\xi-\mu  -\nu  )}{1+(q-b)\wt S_{ab}(\xi-\mu   )-(q-a)\wh{\wt{S}}_{ab}(\xi-\mu  -\eta  )}\nn\\
&&\;\;\;\;\;\;\;\;\;\;\;\;\;\;\;\;\;\;\;\;\;\;\;\;\times \frac{1+(p-b)\wh S_{ab}(\xi-\eta  )-(p-a)\wh{\wt{S}}_{ab}(\xi-\mu  -\eta  )}{1+(r-b)\wh S_{ab}(\xi-\eta   )-(r-a)\wh{\overline{S}}_{ab}(\xi-\eta  -\nu  )}\;,
\end{eqnarray}
which is a three-dimensional lattice equation which appeared first (in a slightly different form) in 
\cite{NCWQ}. Effectively, this is the \emph{Schwarzian lattice KP equation} which in its canonical form was first given in \cite{DN}, cf. also \cite{Sikarin2}. 
\\
\\
We now multiply $\wt{\boldsymbol s}^T$ on the left hand side of \eqref{ua} leading to
\begin{eqnarray}\label{us}
\wt{\boldsymbol s}^T\boldsymbol{ u}_a(\xi)&=&(p-a)\wt{\boldsymbol s}^T\wt{\boldsymbol{ u}}_a(\xi-\mu )+\boldsymbol{\mathrm v}_a(\xi)\wt{\boldsymbol s}^T\wt{\boldsymbol{ u}}_0(\xi-\mu )\;,\nn\\
\boldsymbol s^T(p+\bLam)\boldsymbol{ u}_a(\xi)&=&(p-a)(1-\wt{\boldsymbol{ \mathrm{v}}}_a(\xi-\mu ))+\boldsymbol{\mathrm v}_a(\xi)\wt{\boldsymbol s}^T\wt{\boldsymbol{ u}}_0(\xi-\mu )\;.
\end{eqnarray}
Introducing 
\begin{equation}\label{u00}
u_{00}(\xi)=\boldsymbol s^T(\xi-\boldsymbol Y)^{-1}\boldsymbol r\;
\end{equation}
\eqref{us} can be written in the form
\begin{eqnarray}\label{u001}
(p+\wt{u}_{00}(\xi-\mu ))\boldsymbol{\mathrm v}_a(\xi)-(p-a)\wt{\boldsymbol{\mathrm v}}_a(\xi)=a+\boldsymbol s^T\bLam\boldsymbol{ u}_a(\xi)\;.
\end{eqnarray}
Another two relations related to the ``$\;\;\wh{}\;\;$" and ``$\;\;\bar{}\;\;$" directions can be automatically obtained
\begin{subequations}
\begin{eqnarray}
(q+\wh{u}_{00}(\xi-\eta  ))\boldsymbol{\mathrm v}_a(\xi)-(q-a)\wh{\boldsymbol{\mathrm v}}_a(\xi-\eta  )&=&a+\boldsymbol s^T\bLam\boldsymbol{ u}_a(\xi)\label{u001}\;,\\
(r+\overline{u}_{00}(\xi-\nu  ))\boldsymbol{\mathrm v}_a(\xi)-(r-a)\overline{\boldsymbol{\mathrm v}}_a(\xi-\nu  )&=&a+\boldsymbol s^T\bLam\boldsymbol{ u}_a(\xi)\label{u001}\;.
\end{eqnarray}
\end{subequations}
Eliminating the term $\boldsymbol s^T\bLam\boldsymbol{ u}_a(\xi)$, we can derive the relations
\begin{subequations}
\begin{eqnarray}\label{set1}
(p-q+\wt{u}_{00}(\xi-\mu )-\wh{u}_{00}(\xi-\eta  ))\boldsymbol{\mathrm v}_a(\xi)&=&(p-a)\wt{\boldsymbol{\mathrm v}}_a(\xi-\mu )\nn\\
&&-(q-a)\wh{\boldsymbol{\mathrm v}}_a(\xi-\eta  )\label{u00vw1}\;,\\
(p-r+\wt{u}_{00}(\xi-\mu )-\overline{u}_{00}(\xi-\nu  ))\boldsymbol{\mathrm v}_a(\xi)&=&(p-a)\wt{\boldsymbol{\mathrm v}}_a(\xi-\mu )\nn\\
&&-(r-a)\overline{\boldsymbol{\mathrm v}}_a(\xi-\nu  )\label{u00vw2}\;,\\
(r-q+\overline{u}_{00}(\xi-\nu  )-\wh{u}_{00}(\xi-\eta  ))\boldsymbol{\mathrm v}_a(\xi)&=&(r-a)\overline{\boldsymbol{\mathrm v}}_a(\xi-\nu  )\nn\\
&&-(q-a)\wh{\boldsymbol{\mathrm v}}_a(\xi-\eta  )\label{u00vw3}\;.
\end{eqnarray}
\end{subequations}
We now set $p=a$ then \eqref{u00vw1} and \eqref{u00vw2} become
\begin{subequations}\label{set2}
\begin{eqnarray}
p-q+\wt{u}_{00}(\xi-\mu )-\wh{u}_{00}(\xi-\eta  )&=&-(q-p)\frac{\wh{\boldsymbol{\mathrm v}}_p(\xi-\eta    )}{\boldsymbol{\mathrm v}_p(\xi)}\label{u00vw11}\;,\\
p-r+\wt{u}_{00}(\xi-\mu )-\overline{u}_{00}(\xi-\nu )&=&-(r-p)\frac{\overline{\boldsymbol{\mathrm v}}_p(\xi-\nu    )}{\boldsymbol{\mathrm v}_p(\xi)}\label{u00vw21}\;,
\end{eqnarray}
\end{subequations}
The combination of \eqref{u00vw11} and \eqref{u00vw21} gives
\begin{equation}\label{LKP}
\frac{p-q+\wt{u}_{00}(\xi-\mu )-\wh{u}_{00}(\xi-\eta  )}{p-r+\wt{u}_{00}(\xi-\mu )-\overline{u}_{00}(\xi-\nu  )}
=\frac{p-q+\wt{\overline u}_{00}(\xi-\mu -\nu  )-\wh{\overline u}_{00}(\xi-\eta  -\nu  )}{p-r+\wh{\wt{u}}_{00}(\xi-\mu -\eta  )-\wh{\overline{u}}_{00}(\xi-\eta -\nu   )}\;,
\end{equation}
which is the ``\emph{lattice KP equation}", \cite{NCWQ}, cf. also \cite{NCW}. 
\\
\\
From the definition of the function $\boldsymbol{\mathrm v}_p(\xi)$ in \eqref{tau1}, \eqref{u00vw11} and \eqref{u00vw21} can be written in terms 
of the $\tau$-function
\begin{subequations}
\begin{eqnarray}
p-q+\wt{u}_{00}(\xi-\mu )-\wh{u}_{00}(\xi-\eta  )&=&-(q-p)\frac{\wh{\wt{\tau}}(\xi-\mu  -\eta )}{\wh{\tau}(\xi-\eta  )}\frac{\tau(\xi)}{\wt{\tau}(\xi-\mu )}\label{u00vw113}\;,\\
p-r+\wt{u}_{00}(\xi-\mu )-\overline{u}_{00}(\xi-\nu  )&=&-(r-p)\frac{\wt{\overline{\tau}}(\xi-\mu  -\nu )}{\overline{\tau}(\xi-\nu  )}\frac{\tau(\xi)}{\wt{\tau}(\xi-\mu )}\label{u00vw213}\;.
\end{eqnarray}
\end{subequations}
From \eqref{u00vw3}, if we set $r=a$ we also have
\begin{eqnarray}
r-q+\overline{u}_{00}(\xi-\mu )-\wh{u}_{00}(\xi-\eta  )=-(q-r)\frac{\wh{\overline{\tau}}(\xi-\eta -\nu  )}{\wh{\tau}(\xi-\eta  )}\frac{\tau(\xi)}{\overline{\tau}(\xi-\nu  )}\label{u00vw114}\;.
\end{eqnarray}
The combination of \eqref{u00vw113} \eqref{u00vw213} \eqref{u00vw114} yields
\begin{eqnarray}\label{Hirota}
&&(p-q)\wh{\wt{\tau}}(\xi-\mu  -\eta  )\overline{\tau}(\xi-\nu  )+(r-p)\wt{\overline{\tau}}(\xi-\mu  -\nu  )\wh{\tau}(\xi-\eta )\nn\\
&&\;\;\;\;\;\;\;\;\;\;\;\;\;\;\;\;\;\;\;\;\;\;\;\;\;\;\;+(r-q)\wh{\overline{\tau}}(\xi-\eta  -\nu  )\wt{\tau}(\xi-\mu )=0\;,
\end{eqnarray}
which is the bilinear lattice KP equation, (originally coined DAGTE, cf. \cite{Hirota}). 

To summarize, we have established in this section a direct connection between the discrete-time Ruijsenaars 
model, embedded in a multi-time space, and well-known lattice systems of KP type. This shows that the 
rational discrete-time RS model yields a special class of rational solutions of the KP equation 
through the exploitation of the matrix $\boldsymbol{Y}(n,m,h)$, whose eigenvalues are the RS particle positions 
and which at the same time acts as a kernel for the lattice KP solutions. In this way we obtain  
solutions for all members of the family of KP lattices as classified in \cite{ABS2}, cf. also 
\cite{NC92}. In the trigonometric case of the discrete-time RS system the corresponding solutions are of 
soliton type, cf. \cite{FrankN8}. The connection between the (continuous-time) RS system and solitons has 
also been discussed in \cite{R3}.

\vspace{.2cm} 

\textbf{Remark:} We note that in the non-relativistic limit $\ld\to 0$ , discussed at the end of section 2, 
where we have identified the limiting behaviours 
$p\rightarrow e^{-\lambda p_{CM}}$ and $q\rightarrow e^{-\lambda q_{CM}}$ of the lattice parameters, 
we have as a consequence that ~$p-q\to -\ld(p_{CM}-q_{CM})+\mathcal{O}(\ld^2)$~. Hence, the non-relativistic limit 
$\ld\to 0$ coincides with the so-called ``skew continuum limit'' exploited in the context of the 
lattice equations of KP type, see e.g.  \cite{NCWQ,F5}. This fits with the picture painted in 
\cite{Sikarin1}, where instead of the fully discrete lattice KP equations, the rational CM system was 
treated as arising from a reduction of a semi-discrete KP equation which indeed 
can be obtained by performing a special continuum limit on one of the discrete 
variables.

\section{Discussion}
In this paper we have studied the Lagrangian structure for the Ruijsenaars-Schneider system, and shown that similarly to 
the Calogero-Moser system, which was treated in \cite{Sikarin1}, it possesses a Lagrangian 1-form structure, both on the discrete-time level as well 
as in the continuous-time case.  Thus, this is the second example of a system of ODEs which exhibits a Lagrangian multi-form structure in the sense 
of \cite{SF1} but in a lower-dimensional situation. The present example is important, because in contrast to the CM case where the Lagrange structure 
is closely related to the Lax representation (and hence inherit the closure relation from the zero-curvature condition), here the relation 
between the Lax matrices and the Lagrangians is less obvious, and the validity of the closure relation 
has to be verified by a separate calculation and is therefore more surprising. Thus, we believe that 
these results seem to confirm once again that these Lagrangian form structures are fundamental and 
ubiquitous among integrable systems. 

It is well known that the classical RS system is Liouville-integrable in the continuous-time case, \cite{R1,R2} and formally so in the 
discrete-time case \cite{FrankN8,Frank9} as well. 
With regard to the continuous-time model, the Lagrangian 1-form structure had to be established in 
a rather indirect way, namely by performing systematic limits on Lagrangians of the discrete-time 
system. We have already pointed out that establishing these Lagrange structures 
by Legendre transformation from the known Hamiltonians of the model is complicated, because it is not 
\emph{a priori} known how these 
Hamiltonian flows are embedded in a coherent structure, such that we get acceptable Lagrangian 
components of the 1-form. In this sense the 
discrete-time model can be viewed as a generating object for such Lagrangians for the continuous-time 
model. On the basis of those results, which in fact establish the proper form of 
the ``kinetic terms'' of the continuous hierarchy of Lagrangians it is possible to show that the 
Lagrangian 1-form structure precisely selects the general form of the integrable ``potentials'' 
when a priori arbitrary forms for those potentials are fed into the determining equations, cf. 
\cite{NYK}.  

      
In conclusion, let us state that in our view the importance of this new Lagrangian form structure 
resides in the understanding that it manifests the multidimensional consistency, in the sense of 
the papers \cite{NijWalk,BS}, at the level of the variational principle:  
It provides an answer to the problem of how to find a single Lagrangian framework for a situation 
where we have a multitude of compatible equations  imposed on one and the same 
(possibly vector-valued) function of many independent variables. In the case of ODEs, 
as is the case dealt with in the present paper, the structure is that of a Lagrangian 1-form 
describing systems of commuting flows in many time-variables (as many as the number of degrees of 
freedom of the system). It is obvious that for this structure to hold, the relevant 
Lagrangian components of the 1-form should have very specific forms, 
in order for the closure relation to hold subject to the equations of the motion. 
In fact, such admissable Lagrangians can be considered themselves to be  
solutions of the system of equations arising from the variational principle. 
In the continuous case the constitutive relations arising from this new variational principle, which 
involves variations not only with respect to the dependent variables but also with respect to the underlying 
geometry, were first given in \cite{Sikarin1,Rinthesis}. In a recent paper, \cite{Suris}, Yu. 
Suris from a slightly different point of view\footnote{Rather than considering the variations 
with respect to the geometry \cite{Suris} inspired by our results, considered Lagrangian 1-forms on arbitrary 
curves. We argue, however, that posing a least-action principle with respect to both dependent as 
well as independent variables is a conceptually important step, forming a new paradigm in 
variational calculus, cf. also \cite{Gelfand}, and constitutes potentially a novel principle of 
fundamental physics.} formulated the corresponding Legendre transform.   
That theories which exhibit structures as exemplified in the present paper, but also in higher 
dimensions, always correspond to integrable systems in the sense of other well-known integrability 
features (such as the applicability of the inverse scattering, existence of Lax pairs and higher 
symmetries, etc.) is challenging question that we hope to answer in future work.

\appendix
\section{The construction of the exact solution}\label{A}
\numberwithin{equation}{section}
In this Appendix we review the construction of the exact solution for the RS system. The basic relations following from the Lax pair \eqref{LM} together with the definitions \eqref{LM2} lead to
\begin{subequations}
\begin{eqnarray}
\mu  \boldsymbol{M}_0+\wt{\boldsymbol X}\boldsymbol{M}_0-\boldsymbol{M}_0\boldsymbol X=\wt{h}h^T\;,\label{MX}\\
\lambda \boldsymbol{L}_0+\boldsymbol X\boldsymbol{L}_0-\boldsymbol{L}_0\boldsymbol X=hh^T\;,\label{LX}
\end{eqnarray}
\end{subequations}
where $\boldsymbol X=\sum_{i=1}^Nx_iE_{ii}$ is the diagonal matrix of the particle positions. We have adopted here 
the freedom of making the model slightly more general by introducing in addition to the (relativistic) 
parameter $\lambda$ a new parameter $\mu $ replacing $\lambda$ in the $ \boldsymbol{M}_0$ matrix. 
On the other hand, from the Lax equation \eqref{iden1} and \eqref{iden2}, we obtain the relations
\begin{subequations}\label{XX1}
\begin{eqnarray}
\wt{\boldsymbol{L}}_0\boldsymbol{M}_0&=&\boldsymbol{M}_0\boldsymbol{L}_0\;,\label{L0M01}\\
\wt{\boldsymbol{L}}_0\wt{h}-\boldsymbol{M}_0h&=&-p\wt{h}\;,\label{L0M02}\\
h^T\boldsymbol{L}_0-\wt{h}^T\boldsymbol{M}_0&=&-ph^T\;,\label{L0M03}
\end{eqnarray}
\end{subequations}
where \eqref{L0M02} and \eqref{L0M03} are equivalent to the relations \eqref{idenforh2aa}. 
We now factorize the Lax matrices as follows: 
\begin{equation}\label{XX2}
\boldsymbol{L}_0=\boldsymbol U_0\bLam \boldsymbol U_0^{-1}\;,\;\;\;\;\;\mbox{and}\;\;\;\;\;\boldsymbol{M}_0=\wt{\boldsymbol U}_0\boldsymbol U_0^{-1}\;,
\end{equation}
where $\boldsymbol U_0$ is an invertible $N \times N$ matrix, and where the matrix $\bLam$ 
is constant: $\wt{\bLam}=\bLam$, as a consequence of \eqref{L0M01}. (Obviously, if $\boldsymbol{L}_0$ 
is diagonalizable $\bLam$ is just its diagonal matrix of eigenvalues). Next, introducing
\begin{equation}\label{Yrs}
\boldsymbol{Y}=\boldsymbol U_0^{-1}\boldsymbol{X}\boldsymbol U_0\;,\;\;\;\boldsymbol{r}=\boldsymbol U_0^{-1}\cdot h\;,\;\;\;\boldsymbol{s}^T=h\cdot \boldsymbol U_0\;,
\end{equation}
we obtain from \eqref{XX1} and \eqref{XX2},
\begin{equation}\label{rs}
(p\boldsymbol{I}+\bLam)\cdot \wt{\boldsymbol r}=\boldsymbol r\;,\;\;\;\;\boldsymbol{s}^T\cdot (\mathrm{p}\boldsymbol{I}+\bLam)=\wt{\boldsymbol s}^T\;,
\end{equation}
where $\boldsymbol I$ is the unit matrix, as well as from \eqref{MX} and \eqref{LX}, we have
\begin{eqnarray}\label{Yrs2}
\mu +\wt{\boldsymbol Y}-\boldsymbol Y&=&\wt{\boldsymbol r}\boldsymbol s^T\;,\\
\lambda\bLam+[\boldsymbol Y,\bLam]&=&\boldsymbol r\boldsymbol s^T\;.
\end{eqnarray}
Eliminating the dyadic $\boldsymbol r\boldsymbol s^T$ from \eqref{Yrs2} by making use of \eqref{rs}, we find the linear equation
\begin{equation}\label{Yn1}
\wt{\boldsymbol{Y}}=(p\boldsymbol{I}+\bLam)^{-1}\boldsymbol Y (p\boldsymbol{I}+\bLam)-\frac{p\mu }{p\boldsymbol{I}+\bLam}
+\frac{(\lambda -\mu )\bLam}{p\boldsymbol{I}+\bLam}\;,
\end{equation}
which can be immediately solved to give
\begin{equation}\label{Yn22}
\boldsymbol{Y}(n,m)=(p\boldsymbol{I}+\bLam)^{-n}\boldsymbol Y(0,m)(p\boldsymbol{I}+\bLam)^n-\frac{np\mu }{p\boldsymbol{I}+\bLam} 
+\frac{n(\lambda -\mu )\bLam}{p\boldsymbol{I}+\bLam}\;,
\end{equation}
subject to the constraint on the initial value matrix
\begin{equation}\label{constraint(n)}
[\boldsymbol Y(0,m),\bLam]=-\lambda\bLam+\mbox{rank 1}\;.
\end{equation}
We now consider the matrix $\boldsymbol{N}_0$ in \eqref{aKN2} which can be rewritten in the form
\begin{eqnarray}
\eta  \boldsymbol{N}_0+\wh{\boldsymbol X}\boldsymbol{N}_0-\boldsymbol{N}_0\boldsymbol X=\wh{h}h^T\;,\label{NX}
\end{eqnarray}
where $\eta $ is another relativistic parameter associating to the temporal Lax matrix $\boldsymbol{N}$. Iterating the same process as we did before, we find
\begin{equation}\label{Ym1}
\wh{\boldsymbol{Y}}=(q\boldsymbol{I}+\bLam)^{-1}\boldsymbol Y (q\boldsymbol{I}+\bLam)-\frac{q\eta  }{q\boldsymbol{I}+\bLam}
+\frac{(\lambda -\eta  )\bLam}{q\boldsymbol{I}+\bLam}\;.
\end{equation}
Combining \eqref{Yn1} and \eqref{Ym1}, we can solve for
\begin{eqnarray}\label{Ynm1}
\boldsymbol{Y}(n,m)&=&(p\boldsymbol{I}+\bLam)^{-n}(q\boldsymbol{I}+\bLam)^{-m}\boldsymbol Y(0,0)(q\boldsymbol{I}+\bLam)^{m}(p\boldsymbol{I}+\bLam)^n\nn\\
&&-\frac{np\mu }{p\boldsymbol{I}+\bLam} -\frac{mq\eta }{q\boldsymbol{I}+\bLam}+\frac{n(\lambda -\mu )\bLam}{p\boldsymbol{I}+\bLam}
+\frac{m(\lambda -\eta  )\bLam}{q\boldsymbol{I}+\bLam}\;.
\end{eqnarray}
If we take $\lambda =\mu =\eta $ we recover the relation \eqref{EXACT}.
%
\\
\\

Conversely, we can start from a given $N\times N$ diagonal matrix $\bLam$ with distinct entries, and an initial value matrix $\bsY(0,0)$ subject to the condition that
\begin{equation}
[\bsY(0,0),\bLam]=-\lambda\bLam+{\rm rank}\, 1\;,
\end{equation}
where $[\;,\;]$ represents the matrix commutator bracket. Let $\bU^{-1}(0,0)$ be the matrix that diagonalized $\bsY(0,0)$, i.e., such that
\begin{equation}\label{Ydiag} 
\bsY(0,0)=\bU^{-1}(0,0)\,\bsX(0,0)\,\bU(0,0))\quad,\quad \bsX(0,0)={\rm diag}(x_1(0,0),\dots, x_N(0,0))\  .  
\end{equation}
If the eigenvalues of $\bsY(0,0)$ are distinct (which we can take as an assumption on the initial condition) then $\bU^{-1}(0,0)$ is 
determined up to multiplication from the right by a diagonal matrix times a permutation matrix of the columns. (Fixing an ordering of the 
eigenvalues $x_i(0,0)$, $\bU^{-1}(0,0)$ unique only up to multiplication by a diagonal matrix from the right).  We can fix $\bU^{-1}(0,0)$ 
up to an overall multiplicative factor by demanding that 
\begin{equation}
[\boldsymbol{Y}(0,0)\,,\,\boldsymbol{\Lambda}]=-\lambda\boldsymbol{\Lambda} +\boldsymbol r(0,0)\,\boldsymbol s^T(0,0) \   .  
\end{equation}
Next, we consider the matrix function given by
\begin{eqnarray}\label{EXACTapp}
\boldsymbol{Y}(n,m)&=&(p\boldsymbol{I}+\bLam)^{-n}(q\boldsymbol{I}+\bLam)^{-m}\boldsymbol Y(0,0)(q\boldsymbol{I}+\bLam)^{m}(p\boldsymbol{I}+\bLam)^n\nn\\
&&-\frac{np\mu }{p\boldsymbol{I}+\bLam} -\frac{mq\eta }{q\boldsymbol{I}+\bLam}+\frac{n(\lambda -\mu )\bLam}{p\boldsymbol{I}+\bLam}
+\frac{m(\lambda -\eta  )\bLam}{q\boldsymbol{I}+\bLam}\;.  
\end{eqnarray} 
Let $\bU(n,m)$ be the matrix diagonalizing $\bsY(n,m)$ by an appropriate choice of an overall factor (as a function of $n$ and $m$) this matrix can be fixed such that it obeys:
\begin{equation}
\boldsymbol r(n,m)=(p\bI+\bLam)^{-n}(q\bI+\bLam)^{-m}\boldsymbol r(0,0)\;,\;\;\mbox{and}\;\;\;\boldsymbol s^T(n,m)=\boldsymbol s^T(0,0)(p\bI+\bLam)(q\bI+\bLam)\;,
\end{equation}
and
\begin{equation}\label{YCOMM}
[\boldsymbol{Y}(n,m)\,,\,\boldsymbol{\Lambda}]=-\lambda\boldsymbol{\Lambda} +\boldsymbol r(n,m)\,\boldsymbol s^T(n,m) \   .  
\end{equation}
From the expression \eqref{EXACTapp} we can derive the relations
\begin{subequations}\label{eq:LamY2}\begin{eqnarray}
&& (p\bI+\bLam)\,\wt{\bsY}-\bsY\,(p\bI+\bLam)=-p\mu +(\lambda -\mu )\bLam\  , \\ 
&& (q\bI+\bLam)\,\wh{\bsY}-\bsY\,(q\bI+\bLam)=-q\eta +(\lambda -\eta )\bLam  ,
\end{eqnarray}\end{subequations}
with the usual notation for the shifts in $n$ and $m$ over one unit. Together with the relation \eqref{YCOMM} 
this subsequently yields:
\begin{equation}\label{eq:YY2}
\wt{\bsY}-\bsY=-\mu +\wt{\boldsymbol r}\boldsymbol s^T\quad,\quad \wh{\bsY}-\bsY=-\eta +\wh{\boldsymbol r}\boldsymbol s^T\   . 
\end{equation}
Reversing these relations by rewriting them in terms of 
\begin{equation}\bsX(n,m)=\bU(n,m)\,\bsY(n,m)\,\bU^{-1}(n,m)\end{equation}
and now 
\textit{defining} the Lax matrix by  
\begin{equation}\label{eq:LKdefs}
\bL:= \bU\,\bLam\,\bU^{-1}\quad\   , 
\end{equation}
together with 
\begin{equation}\label{eq:MNdefs2}
\bM:= \wt{\bU}\,\bU^{-1}\quad,\quad \bN:=\wh{\bU}\,\bU^{-1}\   , 
\end{equation}
we recover the relations:
%
\begin{eqnarray}\label{XLK}
[ \bsX\,,\,\bL ] =-\lambda\bL+hh^T\;,
\end{eqnarray}
%
and
\begin{subequations}\label{XMN}
\begin{eqnarray}
\wt{\bsX}\,\bM-\bM\,\bsX &=&-\mu \bM+\wt hh^T\;,\\
\wh{\bsX}\,\bN-\bN\,\bsX &=&-\eta \bN+\wh hh^T\;,
\end{eqnarray}
\end{subequations}
which determine the matrices $\bM$ and $\bN$ as functions of the $x_i(n,m)$ as well as the off-diagonal 
parts of the matrices $\bL$ and $\bK$. 
\\
\\
Furthermore, from \eqref{eq:LamY2} we obtain
\begin{subequations}\label{LXMN}
\begin{eqnarray}
\wt{\bL}\wt{\bsX}\bM-\bM\bsX\bL&=&\left\{ \begin{array}{rcl}
-p\wt{h}h^T+(\lambda -\mu )\wt{\boldsymbol{L}} \,\bM\;,& &  \\ 
-p\wt{h}h^T+(\lambda -\mu )\,\bM\,\boldsymbol{L}\;, & &  
\end{array}\right.\\
\wh{\bK}\wh{\bsX}\bN-\bN\bsX\bK&=&\left\{ \begin{array}{rcl}
-q\wh{h}h^T+(\lambda -\eta  )\wh{\boldsymbol{L}} \,\bN\;& &  \\ 
-q\wh{h}h^T+(\lambda -\eta  )\,\bN\,\boldsymbol{L}\;, & &  
\end{array}\right.
\end{eqnarray}
\end{subequations}
which, when combined with the relations of \eqref{XLK}, yield
\begin{subequations}\label{LMML}
\begin{eqnarray}
\left(\wt{\bL}\bM-\bM\bL \right)\bsX+\left(\wt{\bL}\wt h-\bM h \right)h^T&=&-p\wt h h^T\;,\\
\left(\wh{\bK}\bN-\bN\bK \right)\bsX+\left(\wh{\bK}\wh h-\bN h \right)h^T&=&-q\wh h h^T\;.
\end{eqnarray}
\end{subequations}
On the other hand, using the relations \eqref{XMN} we also obtain
\begin{subequations}\label{LMML2}
\begin{eqnarray}
\wt{\bsX}\left(\wt{\bL}\bM-\bM\bL \right)+\wt{h}\left(h\bL-h^T\bM \right)&=&-p\wt h h^T\;,\\
\wh{\bsX}\left(\wh{\bK}\bN-\bN\bK \right)+\wh{h}\left(h\bK-h^T\bN \right)&=&-q\wh h h^T\;.
\end{eqnarray}
\end{subequations}
From the relations \eqref{LMML} and \eqref{LMML2} it follows that the Lax equations hold and their form is determined up to the diagonal part of the matrices $\bL$ and $\bK$.


\section{Examples}\label{example}
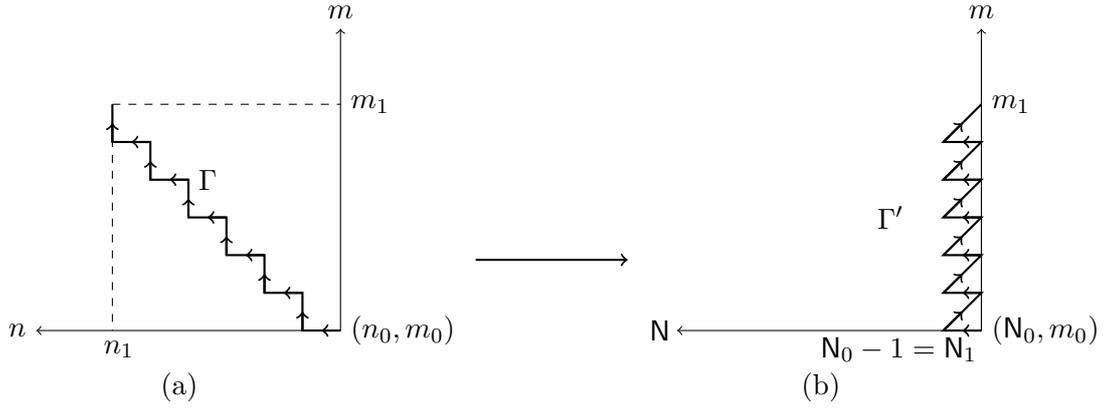
\begin{figure}[h]
\begin{center}
\begin{tikzpicture}[scale=0.5]
 \draw[->] (0,0) -- (-8,0) ;
  \draw (-9,0) node[anchor=west] {$n$};
 \draw[->] (0,0) -- (0,8) node[anchor=south] {$m$};
 \draw[thick] (0,0)  --(-1,0)--(-1,1)--(-2,1)--(-2,2)--(-3,2)--(-3,3)
--(-4,3)--(-4,4)--(-5,4)--(-5,5)--(-6,5)--(-6,6);
 \draw[->,thick] (0,0) -- (-.5,0);
 \draw[->,thick] (-1,0) -- (-1,.5);
\draw[->,thick]   (-1,1)--(-1.5,1);
\draw[->,thick]   (-2,1)--(-2,1.5);
\draw[->,thick]   (-2,2)--(-2.5,2);
\draw[->,thick]   (-3,2)--(-3,2.5);
\draw[->,thick]   (-3,3)--(-3.5,3);
\draw[->,thick]   (-4,3)--(-4,3.5);
\draw[->,thick]   (-4,4)--(-4.5,4);
\draw[->,thick]   (-5,4)--(-5,4.5);
\draw[->,thick]   (-5,5)--(-5.5,5);
\draw[->,thick]   (-6,5)--(-6,5.5);
 \draw (-5,-1.5)node[anchor=west] {(a)};
 \draw (-4,4)node[anchor=west] {$\Gamma$};
\draw (0,0)node[anchor=west] {$(n_0,m_0)$};
\draw (-6.5,-0.5)node[anchor=west] {$n_{1}$};
\draw [dashed](-6,5.5)--(-6,0);
\draw[dashed](-6,6)-- (0,6)node[anchor=west] {$m_1$};
\end{tikzpicture}
\begin{tikzpicture}[scale=0.5]
 \draw[white] (0,0) -- (2,0);
 \draw[white] (0,0) -- (0,2);
 \draw[->,thick] (0,4) -- (4,4);
\end{tikzpicture}
\begin{tikzpicture}[scale=0.5]
 \draw[->] (0,0) -- (-8,0) ;
  \draw (-9,0) node[anchor=west] {$\mathsf N$};
 \draw[->] (0,0) -- (0,8) node[anchor=south] {$m$};
 \draw[thick] (0,0) --(-1,0)--(0,1)--(-1,1)--(0,2)--(-1,2)--(0,3)--(-1,3)--(0,4)
--(-1,4)--(0,5)--(-1,5)--(0,6);
 \draw[->,thick] (0,0) -- (-.5,0);
 \draw[->,thick] (-1,0) -- (-.5,.5);
\draw[->,thick]   (0,1)--(-.5,1);
\draw[->,thick]   (-1,1)--(-.5,1.5);
\draw[->,thick]   (0,2)--(-.5,2);
\draw[->,thick]   (-1,2)--(-.5,2.5);
\draw[->,thick]   (0,3)--(-.5,3);
\draw[->,thick]   (-1,3)--(-.5,3.5);
\draw[->,thick]   (-0,4)--(-.5,4);
\draw[->,thick]   (-1,4)--(-.5,4.5);
\draw[->,thick]   (-0,5)--(-.5,5);
\draw[->,thick]   (-1,5)--(-.5,5.5);
\draw (-5,-1.5)node[anchor=west] {(b)};
 \draw (-3,3)node[anchor=west] {$\Gamma^\prime$};
\draw (0,0)node[anchor=west] {$(\mathsf{N}_0,m_0)$};
\draw (-4.5,-.5)node[anchor=west] {$\mathsf{N}_{0}-1=\mathsf N_1$};
\draw (0,6)node[anchor=west] {$m_1$};
\end{tikzpicture}
\end{center}
\caption{The effect of changing variables on the discrete curve I.}\label{curve_deformation_skew}
\end{figure}
In this Appendix, we will investigate how to derive the discrete Euler-Lagrange equation from the variational principle. For general discrete curves it is cumbersome to
implement the variational principle because of the notation it would require. We will, however, demonstrate how the principle works for a few simple cases: 
\textbf{(a)} the curve shown
 in Fig. (\ref{curve_deformation_skew}), \textbf{(b)} the curve shown in Fig. (\ref{curve_deformation_skewII}).
\\
\\
\textbf{Case(a)}: The curve shown in Fig. (\ref{curve_deformation_skew}): We now introduce a new variable $\mathsf N=n+m$
together with the change of notation
\[
\boldsymbol{x}(n,m)\mapsto\boldsymbol{x}(\mathsf N,m),\;\;\;\;\;\;\wt{\boldsymbol{x}}:=\boldsymbol{x}(\mathsf N+1,m)\;\;\;\mbox{and}\;\;\;
\wh{\boldsymbol x}:=\boldsymbol{x}(\mathsf N,m+1)\;,
\] 
and so we work with the curve given in
 Fig. (\ref{curve_deformation_skew}b). The action evaluated on this curve can be written in the form
\begin{eqnarray}\label{eq:3.14}
 S[\boldsymbol{x};\Gamma^\prime]=\sum_{m=m_0}^{m_1-1}-\mathscr{L}_{(\mathsf N)}(\boldsymbol{x}(\mathsf N_{0}-1,m),\boldsymbol{x}(\mathsf N_0,m))
+\sum_{m=m_0}^{m_1-1}\mathscr{L}_{(m)}(\boldsymbol{x}(\mathsf N_{0}-1,m),\boldsymbol{x}(\mathsf N_{0},m+1)),
\end{eqnarray}
where
\begin{eqnarray}\label{eq:3.14a}
\mathscr{L}_{(\mathsf N)}(\boldsymbol{x},\boldsymbol{y})&=&\sum_{i,j=1}^N\left( f(y_i-{x}_j)-f(y_i-{x}_j-\lambda)\right)-\frac{1}{2}\sum\limits_{\mathop {i,j = 1}\limits_{j \ne i} }^Nf(y_i-y_j+\lambda)\nn\\
&&-\frac{1}{2}\sum\limits_{\mathop {i,j = 1}\limits_{j \ne i} }^Nf(x_i-x_j+\lambda)-\ln\left|p\right|\sum_{i=1}^N(y_i -x_i) \;,\label{Lagn}\\
\mathscr{L}_{(m)}(\boldsymbol{x},\boldsymbol{y})&=&\sum_{i,j=1}^N\left( f(x_i-y_j)-f(x_i-y_j-\lambda)\right)-\frac{1}{2}\sum\limits_{\mathop {i,j = 1}\limits_{j \ne i} }^Nf(x_i-x_j+\lambda)\nn\\
&&-\frac{1}{2}\sum\limits_{\mathop {i,j = 1}\limits_{j \ne i} }^Nf(y_i-y_j+\lambda)-\ln\left|q \right|\sum_{i=1}^N(x_i -y_i)\;,\label{Lagm}
\end{eqnarray}
The minus sign in \eqref{eq:3.14} indicates the reverse direction of the Lagrangian $\mathsf{L}_{(\mathsf N)}$ along the horizontal links. 
Performing the variation $\boldsymbol{x}\mapsto\boldsymbol{x}+\delta\boldsymbol{x}$, we have
\begin{eqnarray}
&&\delta S=0=\nn\\
&&\sum_{m=m_0}^{m_1-1}\left(-\frac{\partial{\mathscr{L}_{(\mathsf N)}
(\boldsymbol{x}(\mathsf N_{0}-1,m),\boldsymbol{x}(\mathsf N_0,m))}}{\partial{\boldsymbol{x}(\mathsf N_0,m)}}\delta\boldsymbol{x}(\mathsf N_0,m)\right.\nn\\
&&\;\;\;\;\;\;\;\;\;\;\;\;\;\;\;\;\;\;\;\;\;\;\;\;\;\;\;\;\;\;\;\;\;\;\;\;\;\;\;\;\;\;\left.-\frac{\partial{\mathscr{L}_{(\mathsf N)}
(\boldsymbol{x}(\mathsf N_{0}-1,m),\boldsymbol{x}(\mathsf N_0,m))}}{\partial{\boldsymbol{x}(\mathsf N_{0}-1,m)}}\delta\boldsymbol{x}(\mathsf N_{0}-1,m)
\right)\nn\\
&&+\sum_{m=m_0}^{m_1-1}\left(\frac{\partial{\mathscr{L}_{(m)}
(\boldsymbol{x}(\mathsf N_{0}-1,m),\boldsymbol{x}(\mathsf N_{0},m+1)}}{\partial{\boldsymbol{x}(\mathsf N_0,m+1)}}\delta\boldsymbol{x}(\mathsf N_{0},m+1)\right.\nn\\
&&\;\;\;\;\;\;\;\;\;\;\;\;\;\;\;\;\;\;\;\;\;\;\;\;\;\;\;\;\;\;\;\;\;\;\;\;\;\;\left.+\frac{\partial{\mathscr{L}_{(m)}
(\boldsymbol{x}(\mathsf N_{0}-1,m),\boldsymbol{x}(\mathsf N_{0},m+1)}}{\partial{\boldsymbol{x}(\mathsf N_{0}-1,m)}}\delta\boldsymbol{x}(\mathsf N_{0}-1,m)
\right)\;.\nn\\
\end{eqnarray}
We now obtain the Euler-Lagrange equations
\begin{subequations}
\begin{eqnarray}
-\frac{\partial{\mathscr{L}_{(\mathsf N)}
(\boldsymbol{x}(\mathsf N_{0}-1,m),\boldsymbol{x}(\mathsf N_0,m))}}{\partial{\boldsymbol{x}(\mathsf N_0,m)}}
+\frac{\partial{\mathscr{L}_{(m)}
(\boldsymbol{x}(\mathsf N_{0}-1,m-1),\boldsymbol{x}(\mathsf N_{0},m)}}{\partial{\boldsymbol{x}(\mathsf N_0,m)}}=0\;,\\
-\frac{\partial{\mathscr{L}_{(\mathsf N)}
(\boldsymbol{x}(\mathsf N_{0}-1,m),\boldsymbol{x}(\mathsf N_0,m))}}{\partial{\boldsymbol{x}(\mathsf N_{0}-1,m)}}
+\frac{\partial{\mathscr{L}_{(m)}
(\boldsymbol{x}(\mathsf N_{0}-1,m),\boldsymbol{x}(\mathsf N_{0},m+1)}}{\partial{\boldsymbol{x}(\mathsf N_{0}-1,m)}}=0\;,
\end{eqnarray}
\end{subequations}
which produce 
\begin{subequations}
\begin{eqnarray}
\ln\left| \frac{p}{q}\right|&=&\sum_{j=1}^N\left(  \ln\left| \frac{\boldsymbol{x}_i(\mathsf N_0,m)-\boldsymbol{x}_j(\mathsf N_0-1,m)}{\boldsymbol{x}_i(\mathsf N_0,m)-\boldsymbol{x}_j(\mathsf N_0-1,m)+\lambda}\right|\right.\nn\\
&&\;\;\;\;\;\;\;\;\;\;\;\;\;\;\;\;\;\;\;\;\;\;\;\;\;+\left.\ln\left|\frac{\boldsymbol{x}_i(\mathsf N_0,m)-\boldsymbol{x}_j(\mathsf N_0-1,m-1)+\lambda}{\boldsymbol{x}_i(\mathsf N_0,m)-\boldsymbol{x}_j(\mathsf N_0-1,m-1)}\right|\right),\nn\\
\ln\left| \frac{p}{q}\right|&=&\sum_{j=1}^N\left( \ln\left| \frac{\boldsymbol{x}_i(\mathsf N_0-1,m)-\boldsymbol{x}_j(\mathsf N_0,m+1)}{\boldsymbol{x}_i(\mathsf N_0-1,m)-\boldsymbol{x}_j(\mathsf N_0,m+1)-\lambda}\right|\right.\nn\\
&&\;\;\;\;\;\;\;\;\;\;\;\;\;\;\;\;\;\;\;\;\;\;\;\;\;+\left.\ln\left|\frac{\boldsymbol{x}_i(\mathsf N_0-1,m)-\boldsymbol{x}_j(\mathsf N_0,m)-\lambda}{\boldsymbol{x}_i(\mathsf N_0-1,m)-\boldsymbol{x}_j(\mathsf N_0,m)}\right|\right)\;,\nn
\end{eqnarray}
\end{subequations}
which are equivalent to \eqref{CON1} and \eqref{CON2}, respectively.
\\
\\
\begin{figure}[h]
\begin{center}
\begin{tikzpicture}[scale=0.5]
 \draw[->] (0,0) -- (8,0) ;
  \draw (8,0) node[anchor=west] {$n$};
 \draw[->] (0,0) -- (0,8) node[anchor=south] {$m$};
 \draw[thick] (0,0)  --(1,0)--(1,1)--(2,1)--(2,2)--(3,2)--(3,3)
--(4,3)--(4,4)--(5,4)--(5,5)--(6,5)--(6,6);
 \draw[->,thick] (0,0) -- (.5,0);
 \draw[->,thick] (1,0) -- (1,.5);
\draw[->,thick]   (1,1)--(1.5,1);
\draw[->,thick]   (2,1)--(2,1.5);
\draw[->,thick]   (2,2)--(2.5,2);
\draw[->,thick]   (3,2)--(3,2.5);
\draw[->,thick]   (3,3)--(3.5,3);
\draw[->,thick]   (4,3)--(4,3.5);
\draw[->,thick]   (4,4)--(4.5,4);
\draw[->,thick]   (5,4)--(5,4.5);
\draw[->,thick]   (5,5)--(5.5,5);
\draw[->,thick]   (6,5)--(6,5.5);
 \draw (2.5,-1.5)node[anchor=west] {(a)};
 \draw (2.5,4)node[anchor=west] {$\Gamma$};
\draw (-2,-.5)node[anchor=west] {$(n_0,m_0)$};
\draw (5.5,-0.5)node[anchor=west] {$n_{1}$};
\draw [dashed](6,5.5)--(6,0);
\draw[dashed](6,6)-- (0,6);
\draw (-1.5,6)node[anchor=west] {$m_1$};
\end{tikzpicture}
\begin{tikzpicture}[scale=0.5]
 \draw[white] (0,0) -- (2,0);
 \draw[white] (0,0) -- (0,2);
 \draw[->,thick] (0,4) -- (4,4);
\end{tikzpicture}
\begin{tikzpicture}[scale=0.5]
 \draw[->] (0,0) -- (8,0) ;
  \draw (8,0) node[anchor=west] {$\mathsf N^\prime$};
 \draw[->] (0,0) -- (0,8) node[anchor=south] {$m$};
 \draw[thick] (0,0) --(1,0)--(0,1)--(1,1)--(0,2)--(1,2)--(0,3)--(1,3)--(0,4)
--(1,4)--(0,5)--(1,5)--(0,6);
 \draw[->,thick] (0,0) -- (.5,0);
 \draw[->,thick] (1,0) -- (.5,.5);
\draw[->,thick]   (0,1)--(.5,1);
\draw[->,thick]   (1,1)--(.5,1.5);
\draw[->,thick]   (0,2)--(.5,2);
\draw[->,thick]   (1,2)--(.5,2.5);
\draw[->,thick]   (0,3)--(.5,3);
\draw[->,thick]   (1,3)--(.5,3.5);
\draw[->,thick]   (0,4)--(.5,4);
\draw[->,thick]   (1,4)--(.5,4.5);
\draw[->,thick]   (0,5)--(.5,5);
\draw[->,thick]   (1,5)--(.5,5.5);
\draw (2.5,-1.5)node[anchor=west] {(b)};
 \draw (2,3)node[anchor=west] {$\Gamma^\prime$};
\draw (-3,-0.5)node[anchor=west] {$(\mathsf{N}^\prime_0,m_0)$};
\draw (1,-.5)node[anchor=west] {$\mathsf N^\prime_1=\mathsf{N}^\prime_{0}+1$};
\draw (-1.5,6)node[anchor=west] {$m_1$};
\end{tikzpicture}
\end{center}
\caption{The effect of changing variables on the discrete curve II.}\label{curve_deformation_skewII}
\end{figure}
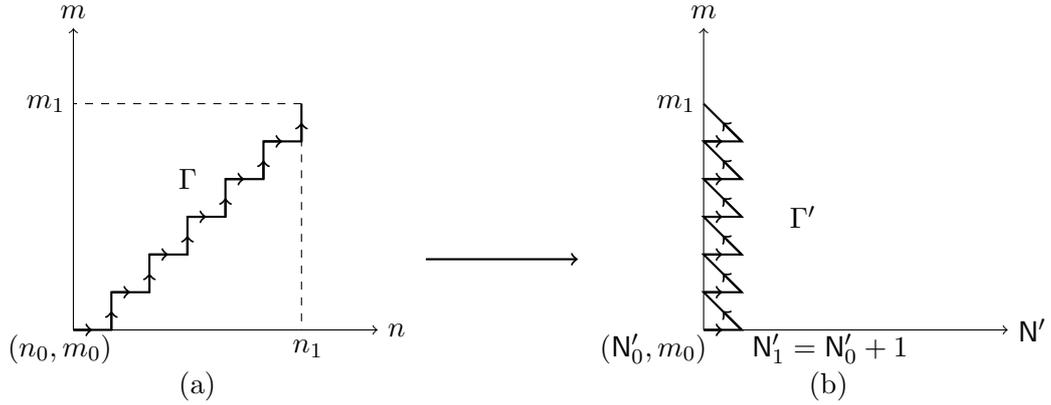
\\
\textbf{Case(b)}: The curve shown in Fig. (\ref{curve_deformation_skewII}):  
Introducing the variable $\mathsf N^\prime=n-m$, the corresponding curve is given in Fig. (\ref{curve_deformation_skewII}b).
 The action evaluated on the curve $\Gamma^\prime$ reads
\begin{eqnarray}\label{eq:3.18}
 S[\boldsymbol{x};\Gamma^\prime]&=&\sum_{m=m_0}^{m_1-1}\mathscr{L}_{(\mathsf N^\prime)}(\boldsymbol{x}(\mathsf N^\prime_0,m),\boldsymbol{x}(\mathsf N^\prime_{0}+1,m))\nn\\
&&\;\;\;\;\;\;\;\;\;\;\;\;\;\;\;\;\;\;\;\;\;\;+\sum_{m=m_0}^{m_1-1}\mathscr{L}_{(m)}(\boldsymbol{x}(\mathsf N^\prime_{0}+1,m),\boldsymbol{x}(\mathsf N^\prime_{0},m+1)),
\end{eqnarray}
where
\begin{eqnarray}\label{eq:3.14ab}
\mathscr{L}_{(\mathsf N^\prime)}(\boldsymbol{x},\boldsymbol{y})&=&\sum_{i,j=1}^N\left( f(x_i-{y}_j)-f(x_i-{y}_j-\lambda)\right)-\frac{1}{2}\sum\limits_{\mathop {i,j = 1}\limits_{j \ne i} }^Nf(y_i-y_j+\lambda)\nn\\
&&-\frac{1}{2}\sum\limits_{\mathop {i,j = 1}\limits_{j \ne i} }^Nf(x_i-x_j+\lambda)-\ln\left|p\right|\sum_{i=1}^N(x_i -y_i) \;,\label{Lagn}\\
\mathscr{L}_{(m)}(\boldsymbol{x},\boldsymbol{y})&=&\sum_{i,j=1}^N\left( f(x_i-y_j)-f(x_i-y_j-\lambda)\right)-\frac{1}{2}\sum\limits_{\mathop {i,j = 1}\limits_{j \ne i} }^Nf(x_i-x_j+\lambda)\nn\\
&&-\frac{1}{2}\sum\limits_{\mathop {i,j = 1}\limits_{j \ne i} }^Nf(y_i-y_j+\lambda)-\ln\left|q\right|\sum_{i=1}^N(x_i -y_i)\;,\label{Lagm}
\end{eqnarray}
Performing the variation $\boldsymbol{x}\mapsto\boldsymbol{x}+\delta\boldsymbol{x}$, we have
\begin{eqnarray}
&&\delta S=0=\nn\\
&&\sum_{m=m_0}^{m_1-1}\left(\frac{\partial{\mathscr{L}_{(\mathsf N^\prime)}
(\boldsymbol{x}(\mathsf N^\prime_0,m),\boldsymbol{x}(\mathsf N^\prime_{0}+1,m)}}{\partial{\boldsymbol{x}(\mathsf N^\prime_0,m)}}\delta\boldsymbol{x}(\mathsf N^\prime_0,m)\right.\nn\\
&&\;\;\;\;\;\;\;\;\;\;\;\;\;\;\;\;\;\;\;\;\;\;\;\;\;\;\;\;\;\;\;\;\;\;\;\;\;\;\;\;\;\;+\left.\frac{\partial{\mathscr{L}_{(\mathsf N^\prime)}
(\boldsymbol{x}(\mathsf N^\prime_0,m),\boldsymbol{x}(\mathsf N^\prime_{0}+1,m)}}{\partial{\boldsymbol{x}(\mathsf N^\prime_{0}+1,m)}}\delta\boldsymbol{x}(\mathsf N^\prime_{0}+1,m)
\right)\nn\\
&&+\sum_{m=m_0}^{m_1-1}\left(\frac{\partial{\mathscr{L}_{(m)}
(\boldsymbol{x}(\mathsf N^\prime_{0}+1,m),\boldsymbol{x}(\mathsf N^\prime_{0},m+1)}}{\partial{\boldsymbol{x}(\mathsf N^\prime_0,m+1)}}\delta\boldsymbol{x}(\mathsf N^\prime_{0},m+1)\right.\nn\\
&&\;\;\;\;\;\;\;\;\;\;\;\;\;\;\;\;\;\;\;\;\;\;\;\;\;\;\;\;\;\;\;\;\;\;\;\;\;\;\left.+\frac{\partial{\mathscr{L}_{(m)}
(\boldsymbol{x}(\mathsf N^\prime_{0}+1,m),\boldsymbol{x}(\mathsf N^\prime_{0},m+1)}}{\partial{\boldsymbol{x}(\mathsf N^\prime_{0}+1,m)}}\delta\boldsymbol{x}(\mathsf N^\prime_{0}+1,m)
\right)\;.\nn\\
\end{eqnarray}
We now obtain the Euler-Lagrange equations
\begin{subequations}
\begin{eqnarray}
\frac{\partial{\mathscr{L}_{(\mathsf N^\prime)}
(\boldsymbol{x}(\mathsf N^\prime_0,m),\boldsymbol{x}(\mathsf N^\prime_{0}+1,m)}}{\partial{\boldsymbol{x}(\mathsf N^\prime_0,m)}}
+\frac{\partial{\mathscr{L}_{(m)}
(\boldsymbol{x}(\mathsf N^\prime_{0}+1,m-1),\boldsymbol{x}(\mathsf N^\prime_{0},m)}}{\partial{\boldsymbol{x}(\mathsf N^\prime_0,m+1)}}=0\;,\\
\frac{\partial{\mathscr{L}_{(\mathsf N^\prime)}
(\boldsymbol{x}(\mathsf N^\prime_0,m),\boldsymbol{x}(\mathsf N^\prime_{0}+1,m)}}{\partial{\boldsymbol{x}(\mathsf N^\prime_{0}+1,m)}}
+\frac{\partial{\mathscr{L}_{(m)}
(\boldsymbol{x}(\mathsf N^\prime_{0}+1,m),\boldsymbol{x}(\mathsf N^\prime_{0},m+1)}}{\partial{\boldsymbol{x}(\mathsf N^\prime_{0}+1,m)}}=0\;,
\end{eqnarray}
\end{subequations}
which produce 
\begin{subequations}
\begin{eqnarray}
\ln\left| \frac{p}{q }\right|&=&\sum_{j=1}^N\left(  \ln\left| \frac{\boldsymbol{x}_i(\mathsf N_0^\prime,m)-\boldsymbol{x}_j(\mathsf N_0^\prime-1,m)}{\boldsymbol{x}_i(\mathsf N_0^\prime,m)-\boldsymbol{x}_j(\mathsf N_0^\prime-1,m)+\lambda}\right|\right.\nn\\
&&\;\;\;\;\;\;\;\;\;\;\;\;\;\;\;\;\;\;\;\;\;\;\left.+\ln\left|\frac{\boldsymbol{x}_i(\mathsf N_0^\prime,m)-\boldsymbol{x}_j(\mathsf N_0^\prime-1,m-1)+\lambda}{\boldsymbol{x}_i(\mathsf N_0^\prime,m)-\boldsymbol{x}_j(\mathsf N_0^\prime-1,m-1)}\right|\right),\nn\\
\ln\left| \frac{p}{q }\right|&=&\sum_{j=1}^N\left( \ln\left| \frac{\boldsymbol{x}_i(\mathsf N_0^\prime-1,m)-\boldsymbol{x}_j(\mathsf N_0^\prime,m+1)}{\boldsymbol{x}_i(\mathsf N_0^\prime-1,m)-\boldsymbol{x}_j(\mathsf N_0^\prime,m+1)-\lambda}\right|\right.\nn\\
&&\;\;\;\;\;\;\;\;\;\;\;\;\;\;\;\;\;\;\;\;\;\;+\left.\ln\left|\frac{\boldsymbol{x}_i(\mathsf N_0^\prime-1,m)-\boldsymbol{x}_j(\mathsf N_0^\prime,m)-\lambda}{\boldsymbol{x}_i(\mathsf N_0^\prime-1,m)-\boldsymbol{x}_j(\mathsf N_0^\prime,m)}\right|\right)\;,\nn
\end{eqnarray}
\end{subequations}
which are equivalent to \eqref{CON1} and \eqref{CON2}, respectively.

By working with the specific type of curves given in Fig. (\ref{curve_deformation_skew}) and Fig. (\ref{curve_deformation_skewII}), we can perform the variational principle
with either the new variables $(\mathsf N,m)$ or $(\mathsf N^\prime,m)$. We obtain the Euler Lagrange equations corresponding to each link of the discrete curve and we obtain 
constraint equations \eqref{CON1} and \eqref{CON2} describing the dynamic of the system from one direction to another direction of the discrete-time (while the equations of motion \eqref{eqmotion1} and \eqref{eqmotion12} represent the dynamic of the system on one discrete-time direction).

\begin{acknowledgements}

S. Yoo-Kong is supported by King Mongkut's University of Technology Thonburi Research Grant for a new acdemic staff 2012. 
F. Nijhoff is supported by a Royal Society/Leverhulme Trust Senior Research Fellowship and would like to thanks Department of Physics, Faculty of Scicence, King Mongkut's University of Technology Thonburi for hospitality as the final state of the paper has been done. We are grateful to S. Ruijsenaars 
for useful comments. 
\end{acknowledgements}


\end{document}